\documentclass[nofootinbib,reprint,amssymb,superscriptaddress,twocolumn]{revtex4-2}
\usepackage{graphicx}
\usepackage{subfig}
\usepackage{hyperref}
\usepackage{titletoc}
\titlecontents{section}
  [4em] 
  {} 
  {\contentslabel{4em}} 
  {\hspace*{-4em}} 
  {\titlerule*[0.5pc]{.}\contentspage} 
\usepackage{romannum}
\titlecontents{subsection}
  [7em] 
  {} 
  {\contentslabel{3.5em}} 
  {\hspace*{-3.5em}} 
  {\titlerule*[0.5pc]{.}\contentspage} 

\titlecontents{subsubsection}
  [7em] 
  {} 
  {\contentslabel{3.5em}} 
  {\hspace*{-3.5em}} 
  {\titlerule*[0.5pc]{.}\contentspage} 

\newcommand{\QFI}{\mathcal J_{\mathrm Q}}      
\usepackage{tabularx}
\usepackage{tikz}
\usepackage{amsmath,amsfonts,amssymb,mathtools}
\usepackage{amsthm}

\newtheorem{theorem}{Theorem}       \newtheorem{lemma}{Lemma}
\newtheorem{corollary}{Corollary} 
\usepackage{tasks}
\usepackage{csquotes}
\usepackage{diagbox}
\usepackage{float}
\usepackage[utf8]{inputenc}
\usepackage{csquotes}
\usepackage{float}
\usepackage{textcomp}
\usepackage{bbm,bm}
\usepackage[normalem]{ulem}
\usepackage{comment}
\usepackage{physics}
\usepackage{xcolor}
\usepackage{CJKutf8}

\usepackage{pgfplots}
\definecolor{blueprl}{RGB}{46,48,146}
\usepgfplotslibrary{groupplots}

\usepackage{lipsum}

      \def\QUICS{Joint Center for Quantum Information and Computer Science,
NIST/University of Maryland, College Park, Maryland 20742, USA}
\def\JQI{Joint Quantum Institute, NIST/University of Maryland, College Park, Maryland 20742, USA}
\def\ASTAR{Quantum Innovation Centre (Q.InC), Agency for Science, Technology and Research (A*STAR), 4 Fusionopolis Way, Kinesis \#05-01, Singapore 138635, Republic of Singapore}
\def\CQT{Centre for Quantum Technologies (CQT), National University of Singapore, Singapore 117543, Republic of Singapore}
\newcommand{\USNA}{Volgenau Department of Physics, United States Naval Academy, Annapolis, MD 21402, USA}
\begin{document}
\pagenumbering{arabic}
 \title{Super-Exponential Advantage of Squeezed Light in\\ Phase Estimation under Discrete Phase Randomisation}
 
\author{Lorc{\'a}n O. Conlon}
\email{lorcanconlon@gmail.com}
\affiliation{\JQI}
\affiliation{\QUICS}
\author{V Vijendran}
\affiliation{\CQT}
\affiliation{\ASTAR}
\author{Aritra Das}
\affiliation{ Center for Quantum Software and Innovation,
 University of Technology Sydney, Ultimo, NSW 2007, Australia}
\author{\begin{CJK}{UTF8}{gbsn}Yu-Xin Wang (王语馨)\end{CJK}}
\affiliation{\QUICS}
\author{Erfan~Abbasgholinejad}
\affiliation{\JQI}
\affiliation{\QUICS}
\author{Anthony J. Brady}
\affiliation{IonQ, Inc.}
\author{Jacob~Bringewatt}
\affiliation{\USNA}
\author{Sean~R.~Muleady}
\affiliation{\JQI}
\affiliation{\QUICS}
\author{Syed M. Assad}
\affiliation{\ASTAR}
\author{Alexey~V.~Gorshkov}
\affiliation{\JQI}
\affiliation{\QUICS}

\begin{abstract}
A paradigmatic task in quantum metrology is that of phase estimation. For this task it is commonly believed, under reasonable assumptions about prior information, that a quadratic advantage in sensitivity is the maximum possible advantage offered by the use of quantum resources. In this work, we consider a channel that implements a phase shift by an angle $2\pi k/M$, where $k$ is a uniformly random unknown integer, in addition to a small unknown phase shift that is our parameter of interest. We show that, subject to a constraint on the mean photon number per channel use, there can be a super-exponential separation between the quantum Fisher information (QFI) attainable using a squeezed vacuum state and that attainable using coherent states. This advantage also holds over any state obtained from a coherent state via phase covariant channels, including coherent states of smaller amplitude and displaced thermal states. Specifically, when the mean photon number per channel use is 1, the ratio of their QFI grows as $\Omega[(M/\sqrt{2}e)^M]$. For a mean photon number per channel use of $M/\log(M)$, the squeezed vacuum QFI is $\Omega(M^{3/2}/\sqrt{\log(M)})$, whereas for a coherent state it is at most $\exp[-M\log\log(M)+O(M)]$. We complement this QFI analysis with non-asymptotic estimation
results. For a mean photon number per channel use of $M/\log M$, an explicit
squeezed-vacuum protocol estimates the phase to accuracy $O(1/M)$ using
$O(\sqrt{M\log M})$ channel uses, whereas every coherent state protocol requires at least $\exp[M\log\log M-O(M)]$ channel uses, even allowing correlated probes and
collective measurements. We then consider the effects of noise (optical loss, thermal noise, and phase diffusion) and show that the ratio of their QFIs can still grow super-exponentially in $M$. Under fixed optical loss and thermal noise, choosing $\bar n=M^2/(\log M)^{3/2}$ allows squeezed vacuum with heterodyne detection to achieve accuracy $O(1/M)$ using sub-polynomially many channel uses, whereas every protocol in the specified classical family requires super-polynomially many. Finally, we demonstrate how experimentally feasible measurements such as variational photonic circuits or homodyne detection pave the way for a demonstration of this advantage.
\end{abstract}
\maketitle

\section{Introduction}
Quantum sensing employs quantum mechanical resources such as entanglement and squeezing to measure physical properties of systems better than is possible when restricted to classical resources~\cite{giovannetti2011advances}. Practical applications of quantum sensing have emerged in gravitational wave detection~\cite{aasi2013enhanced,tse2019quantum} and biosensing~\cite{casacio2021quantum}, alongside several other proof-of-principle demonstrations of quantum-enhanced sensing~\cite{daryanoosh2018experimental,nielsen2023deterministic,marciniak2022optimal,conlon2023approaching,conlon2023discriminating,yung2025saturating,das2026precision,yung2026beating,li2023optimal,franke_quantum-enhanced_2023,eckner_realizing_2023,robinson_direct_2024}. Conventional quantum metrology focuses on the estimation of a deterministic parameter unitarily encoded onto a quantum state~\cite{bollinger1996optimal}. A paradigmatic example is estimating an unknown rotation or phase shift. For $N$ qubits, we typically consider a rotation about the $Z$-axis~\cite{bollinger1996optimal}. In this setting, the standard quantum limit (SQL), corresponding to the mean-squared-error (MSE) attainable using unentangled probe states, scales as $1/N$. In contrast, when using a $N$-qubit entangled state, the Heisenberg limit can be achieved, whereby the MSE scales as $1/N^2$. For bosonic systems, the equivalent task is to estimate a phase shift using a state with mean photon number $\bar{n}$~\cite{demkowicz2015quantumlimits,barbieri2022optical}. In this setting, quantum resources can give a MSE that scales as $1/\bar{n}^2$ (the Heisenberg limit), whereas with classical resources the MSE scales as $1/\bar{n}$ (SQL).

While the advantage of the Heisenberg limit over the SQL is not insignificant, it is desirable to have as large a quantum advantage as possible. Along this line of thought, many works have sought out a beyond-quadratic advantage in quantum sensing. There are many different ways in which such a beyond-quadratic advantage can be engineered. These include coupling the unknown parameter of interest to many generators or high-weight generators~\cite{boixo2007generalized,roy2008exponentially}, operating near a quantum critical point~\cite{gong2008fidelity,gu2008fidelity}, coupling to time-dependent Hamiltonians~\cite{pang2017optimal}, and estimating exponentially many parameters~\cite{das2025holevo,chen2022quantum,chen2025efficient,gorecki2020optimal,brady2026precision,kwon2026universal}. In many of these proposals, the beyond-quadratic scaling is obtained by changing the encoding dynamics, exploiting many-body structure, or increasing the number of parameters. It is therefore natural to ask whether an exponential separation can arise in a single-parameter problem with a more standard encoding.

Recently, there have been several works attempting exactly this. Ref.~\cite{wang2024exponential} demonstrates an exponential advantage from entanglement in estimating correlated noise. Additionally, for estimating correlated parameters encoded unitarily across $N$ qubits, it is known that there is an exponential advantage when using entanglement~\cite{prabhu2026exponential}. In Refs.~\cite{kannan2026exponential2,cotler2026quantum,oh2024entanglement,Baiju2026Quantum}, exponential separations arise in learning or discriminating properties of bosonic channels or classical-field-induced distributions, rather than in the direct local estimation of a single parameter. These works leave open the question of whether there can exist an exponential separation in the direct estimation of a single unchanging parameter with an uncorrelated noise model. We note that Theorem 4 of Ref.~\cite{cotler2026quantum} demonstrates an exponential advantage in this setting for a decision-based problem. However, whether such an advantage can be attained for phase estimation or solely through a parameter estimation framework remains unclear.

In this work, we answer this question in the affirmative, demonstrating an exponential advantage of squeezed light over a family of classical states in phase estimation under a phase twirling channel. We consider a bosonic channel that rotates a state by an unknown angle $\theta_k=2\pi k/M$, where $k$ is an integer between $1$ and $M$. Our task is to estimate an additional rotation $\delta$, whereby every rotation angle is $\theta_k+\delta$, as shown in Fig.~\ref{fig:mpskscheme} (a) and (b). Remarkably, we demonstrate that with mean photon number up to $\bar{n}= M/\log(M)$, there is a super-exponential in $M$ separation between how well $\delta$ can be estimated with and without access to squeezing. Another interesting feature of the advantage presented here is that it displays some robustness to noise. Many quantum sensing protocols that achieve quadratic advantage in the noiseless setting display only a constant-factor advantage in the presence of noise~\cite{Huelga1997,dorner2009optimal,Escher2011,demkowicz2012elusive} unless certain conditions are satisfied~\cite{demkowicz2017adaptive,zhou2018achieving}. A similar fragility occurs in quantum learning: in many protocols where ideal quantum resources reduce the sample complexity from exponential to system-size independent, fixed physical noise can make the quantum sample complexity exponential, leaving an advantage only in the exponent~\cite{oh2024entanglement,seif2026entanglement}. Here, in contrast, we show that when $\bar{n}=M^2/(\log M)^{3/2}$, there is a super-polynomial advantage in terms of the number of samples required to achieve a certain accuracy $\epsilon=O(1/M)$ when using squeezed vacuum ($\nu_{\rm SV}^{\star}=M^{o(1)}$) compared to a family of classical states ($\nu_{\rm cl}^{\star}=M^{\omega(1)}$) in the presence of constant optical loss and thermal noise.

Taken together the features displayed by phase estimation under discrete phase randomisation make it an attractive problem for both near-term and long-term demonstrations of quantum advantage. On the practical side, this task is related to the calibration of quantum communication devices~\cite{cao2015discrete,slimane2012blind,wang2023discrete,bauml2024security,kanitschar2022optimizing} (Fig.~\ref{fig:mpskscheme} (c)). Random phase patterns are also used to calibrate spatial light modulators~\cite{schroff2024rapid}, suggesting that blind calibration of randomized phase devices may be relevant beyond communication systems.

The layout of this paper is as follows. In Section~\ref{sec:prelim} we introduce the preliminary material needed to understand our results. In Section~\ref{sec:metrologyworse} we present our main results, demonstrating super-exponential advantages both in the ideal setting and in the presence of noise. Then in Section~\ref{sec:implementation} we discuss possible near-term and future implementations of our protocols. Finally in Section~\ref{sec:discussion} we present some concluding remarks.

\section{Preliminary material}
\label{sec:prelim}

 In this work, we are interested in estimating an unknown parameter $\delta$, encoded in a quantum system $\rho=\rho(\delta)$. To extract information about $\delta$, we measure the quantum system. A measurement is described by a positive operator-valued measure (POVM) which is a set of positive operators $\{\Pi_k\}$ that sum to the identity $\sum_{k}\Pi_{k}=I_d$. The $k$-th measurement outcome occurs with probability
\mbox{$p_k=\text{Tr}[\rho\Pi_{k}]$}. Based on the observed measurement outcomes, the experimenter arrives at an estimated value of $\delta$, $\hat{\delta}$. There are many different figures of merit for characterising how good an estimation procedure is. Here we focus on two criteria that are particularly informative: 1) parameter estimation where we take as our figure of merit the MSE and 2) Probably Approximately Correct (PAC) learning where the figure of merit is the number of samples required to learn an unknown parameter to a certain accuracy with a certain probability.

\subsection{Parameter estimation}
In the parameter estimation framework, we restrict ourselves to local estimation around a known value $\delta_{0}$~\cite{fujiwara2006strong,holevo2011probabilistic}. A fundamental bound on the MSE, given a specific POVM, can be obtained from the classical Fisher information (CFI):
\begin{equation}
\label{eq:CFIdefin}
\mathcal{J}=\mathbb{E}\left[\left(\frac{\partial \log (P(\delta))}{\partial \delta}\right)^2\right],
\end{equation}
where $P(\delta)$ is the probability distribution of measurement outcomes and the expectation value is taken over all possible measurement outcomes. Defining the MSE as $v_{\delta}=\mathbb{E}[(\hat{\delta}-\delta)^2]\;$, we have the following bound:
\begin{equation}
    v_{\delta}\geq\frac{1}{\nu\mathcal{J}}\;,
\end{equation}
where $\nu$ is the number of experimental repetitions. As $\nu$ is a purely classical effect in this framework, we shall drop any dependence on $\nu$ going forward when discussing Fisher information.

While the CFI is very useful, it requires a specific POVM in order to be computed. To alleviate this requirement we turn to the quantum Fisher information (QFI), $\QFI$. Given a quantum state $\rho(\delta)$, one can compute the symmetric logarithmic derivative (SLD) operator $\mathcal{L}_{\delta}$ as
$\frac{\partial\rho}{\partial\delta}=\frac{1}{2}(\rho \mathcal{L}_{\delta}+\mathcal{L}_{\delta}\rho)\;$. From this, the QFI can be computed as
\begin{equation}
    \QFI=\text{Tr}[\rho \mathcal{L}_{\delta}^2]\;.
\end{equation}
The QFI then bounds the MSE for any POVM and unbiased estimator as
\begin{equation}
  v_{\delta}\geq\frac{1}{\mathcal{J}}\geq\frac{1}{\QFI}\;.
\end{equation}
For estimating a single unknown parameter, the measurement and classical post-processing required to achieve equality in the above bound are known.

\subsection{Probably Approximately Correct Learning}
The PAC framework quantifies estimation performance at finite sample size without the strict requirement of local unbiasedness. In PAC learning, we ask what is the minimum number of samples $\nu$ required to produce an estimate $\hat{\delta}$ such that
\begin{equation}
    \text{Pr}[|\hat{\delta}-\delta|\leq\epsilon] \geq1-\zeta\;,
\end{equation}
for every $\delta$ value in the allowed domain $\Theta$. When a protocol satisfies this definition, we will call it $(\epsilon,\zeta)$ accurate. This is a standard figure of merit in, for example, Pauli channel learning~\cite{chen2022exponential,chen2024tight}. For all instances of the PAC learning tasks considered here, we explicitly forbid quantum memory across different channel uses, i.e.~we only allow protocols where quantum states are prepared one at a time, they interact with the channel and are measured individually. Adaptively updating the quantum state based on previous measurement outcomes is allowed. Our classical lower bounds hold more generally, for classically correlated probes, collective measurements on all outputs and adaptive strategies.

\section{Exponential advantage in phase estimation with squeezing}
\label{sec:metrologyworse}
We now present our main result, demonstrating that we can modify a standard task in quantum sensing (phase estimation) with quadratic quantum sensitivity enhancement in such a way as to obtain a super-exponential separation between squeezed vacuum states and a family of classical states. 

\subsection{Channel model}

We consider phase estimation with a single bosonic mode in the presence of a uniform $\mathbb Z_M$ phase twirl. Each channel use applies an independently sampled rotation $2\pi k/M$, with the value of $k$ hidden, together with the unknown phase shift of interest $\delta$:
\begin{equation}
\begin{split}
     \mathcal E_{\delta,M}(\ket{\psi})
    &:=e^{\mathrm i\delta\hat n}\mathcal T_M(\ket{\psi})e^{-\mathrm i\delta\hat n},
    \qquad\\
    \mathcal T_M(\ket{\psi})
    &:=\frac1M\sum_{k=0}^{M-1}
    e^{2\pi\mathrm i k\hat n/M}\ket{\psi}\bra{\psi}e^{-2\pi\mathrm i k\hat n/M}\;,
\end{split}
\label{eq:sample-channel}
\end{equation}
where $\hat{n}$ is the photon number operator. Defining $M=2^m$, the twirling channel can also be thought of as a $m$-bit rotation channel. Note that, by definition, when the index $k$ is unknown, $\delta$ is only distinguishable within the interval $\delta\in[-\pi/M,\pi/M)$. As such, when considering PAC learning tasks, the domain of allowed $\delta$ values is $\Theta_M=[-a/M,a/M]$ for some value of $a$, $0<a<\pi/2$, fixed independently of $M$. When computing the QFI or CFI we are implicitly assuming that the value of $\delta$ is known locally. Throughout we assume $M$ to be even.

We emphasise that this is a standard phase estimation task: we can separate the twirling rotation from the rotation by $\delta$, $\mathcal E_{\delta, M} = \mathcal E_{\delta,1} \circ \mathcal E_{0, M}$. For example, when $M=1$ ($m=0$), this reduces exactly to standard phase estimation. Similarly, if the unknown rotation index $k$ is known in every round of the experiment, then the problem again reduces to the standard phase estimation setup. Therefore, the QFI of the squeezed vacuum states with mean photon number $\bar{n}$ that we consider below cannot exceed the untwirled QFI of $8\bar{n}(\bar n+1)$~\cite{nielsen2023deterministic}, raising the question of how we hope to arrive at an exponential advantage. The resolution is that for this task the performance of classical states such as coherent states gets exponentially worse in $M$, as illustrated in Fig.~\ref{fig:mpskscheme}(e). The Wigner function of the coherent state when averaged over unknown $k$ contains far less information about $\delta$ than that of the squeezed state input. On the contrary, as we demonstrate here, squeezed vacuum states can be resilient to the twirling process.

\begin{figure*}[t]
    \centering
    \includegraphics[width = \textwidth]{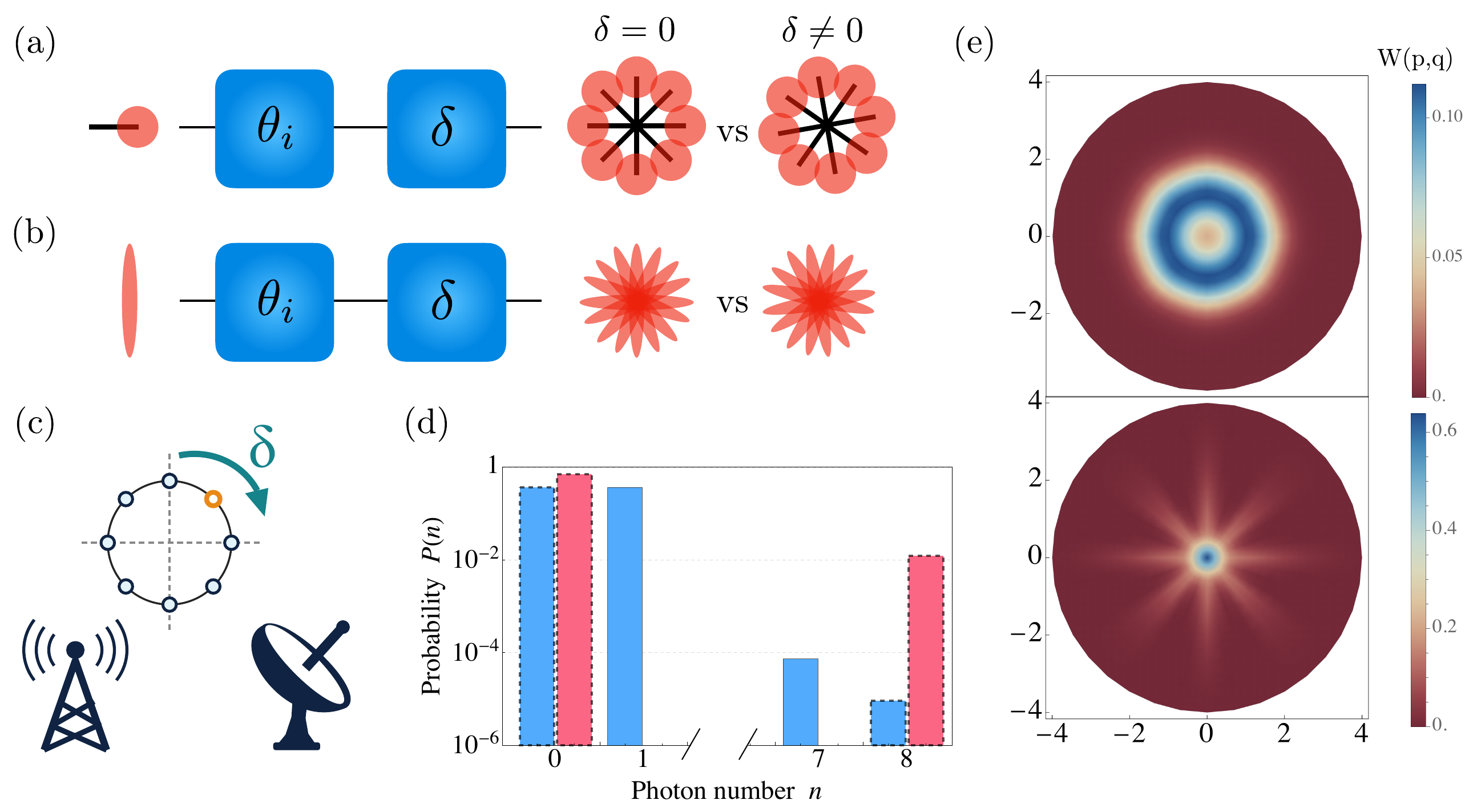}
    \caption{\textbf{Phase estimation in a $\mathbb{Z}_M$ phase twirling channel.} (a), (b) Estimation with coherent states and squeezed vacuum states respectively. The task is to determine $\delta$ in the presence of a randomised rotation $\theta_i$. (c) One possible application of this work is the calibration of quantum communication networks. (d) Photon number distribution of coherent states (yellow) and squeezed vacuum states (purple) for mean photon number of 1. (e) Wigner functions of a coherent state and squeezed vacuum state after passing through a channel with $M=8$ rotations. Both states have a mean photon number of $\sqrt{2}$.}
    \label{fig:mpskscheme}
\end{figure*}

Before computing the QFI, it will be useful to derive a direct sum structure over $M$ sectors for $  \mathcal E_{0,M}(\ket{\psi})$. We write the input state in the Fock basis as $
\ket{\psi}=\sum_{n}c_{n}\ket{n}$. After the $m$-bit rotation channel, we can write $\mathcal E_{0,M}(\ket{\psi})$ in the following direct sum structure (see Appendix~\ref{apen:mary1}):
\begin{equation}
     \mathcal E_{0,M}(\ket{\psi})=\bigoplus_{r=0}^{M-1}\lambda_r\ket{\tilde{u}_r}\bra{\tilde{u}_r}\;,
\end{equation}
where
\begin{equation}
\label{eq:utildemain}
\ket{\tilde{u}_r}=\frac{1}{\sqrt{\lambda_r}}\sum_{p=0}^\infty c_{r+pM}\ket{r+pM}\;,
\end{equation}
and
\begin{equation}
\lambda_r=\sum_{p=0}^\infty |c_{r+pM}|^2\;.
\end{equation}
Note that this state maintains coherences only within each $r$ sector. Coherences between different $r$ sectors are removed entirely.

\subsection{Ideal setting results}
We now compute the QFI for estimating $\delta$ under the action of the above channel. The additivity of the QFI for direct sum density matrices implies
\begin{equation}
\label{eq:QGIgeneralVar}
\begin{split}
        \QFI(\mathcal E_{\delta,M}(\ket{\psi}))&=\sum_{r=0}^{M-1}\lambda_r\QFI(\ket{\tilde{u}_r})\\
        &=4\sum_{r=0}^{M-1}\lambda_r\text{Var}_{\ket{\tilde{u}_r}}(\hat{n})\\
        &=4\mathbb{E}[\text{Var}(\hat{n}|\hat{n} \text{ mod }M)]\;.
\end{split}
\end{equation}
The above equation reveals the origin of the advantage of squeezed states over classical states such as coherent states: the random rotation channel preserves only coherences between Fock levels separated by multiples of $M$, so the QFI is determined by the photon-number variance remaining within each residue class modulo $M$. The rapidly decaying Poisson tail of a coherent state makes this conditional variance extremely small when $\bar{n}\ll M$, whereas the photon tail of squeezed vacuum maintains appreciable probability on number states separated by $M$. This is shown in Fig.~\ref{fig:mpskscheme} (d) for $M=8$.

 Utilising this expression, we arrive at the following simple formula for the QFI (see Appendix~\ref{apen:mary1}):
\begin{equation}
\label{eq:QFIsimple}
    \QFI=4M^2 \sum_{j,\lambda_j\neq0}\left(\sum_{p=0}^\infty |c_{j+pM}|^2p^2-\frac{\left(\sum_{p=0}^\infty |c_{j+pM}|^2p\right)^2}{\lambda_j}\right)\;.
\end{equation}
By summing to some finite cutoff rather than infinity, this expression offers an efficiently computable, and accurate, approximation for the QFI. When using this expression to determine the QFI numerically, we increase the cutoff until the obtained value converges. 

More importantly, for certain classes of states, we can upper- and lower-bound this expression. We begin by analysing a coherent state, written as
\begin{equation}
\ket{\alpha} =  \, e^{-|\alpha|^2/2} \sum_{n=0}^{\infty} 
\frac{ \alpha^n }{\sqrt{n! }} \ket{n}.
\end{equation}
 In Appendix~\ref{apen:maryC}, we prove that, when using a coherent state with $|\alpha|^2<\frac{M^2}{M+1}$, the QFI is upper bounded by
\begin{equation}
\label{eq:qfiubCoh}
\QFI^\text{C}\leq 4\frac{M^2\mathrm{e}^{-|\alpha|^2}|\alpha|^{2M}}{M!(1-|\alpha|^2\frac{M+1}{M^2})}=4\frac{M^2\mathrm{e}^{-\bar{n}}\bar{n}^{M}}{M!(1-\bar{n}\frac{M+1}{M^2})}\;,
\end{equation} 
where $\QFI^\text{C}=\QFI(\mathcal E_{\delta,M}(\ket{\alpha}))$, and we use the fact that $\bar{n}=|\alpha|^2$ for a coherent state. Based on this, we prove that 
\begin{equation}
\bar{n}(M)=\Omega(M)
\end{equation}
photons are required to keep the QFI constant. Specifically, if we allow only $\bar{n}(M)=M/\log(M)$ photons, the QFI will scale like 
\begin{equation}
\label{eq:qfiubCoh3}
\QFI^\text{C}\leq\exp\left[-M\log\log M+M-o(M)\right]\;.
\end{equation} 

Similarly, when using a coherent state with $\bar n=|\alpha|^2<\frac{M^2}{M+1}$, we prove in Appendix~\ref{apen:samplecomplexity} that the number of samples required for a $(\epsilon,\zeta)$ accurate estimate goes as
\begin{equation}
\nu_{\mathrm{C}}^\star
    \ge
   \frac{(1-2\zeta)^2}{16\epsilon^2M^2}
    \left(1-\frac{\bar n(M+1)}{M^2}\right)e^{\bar n}\frac{M!}{\bar{n}^M}\;.
\end{equation}
Recall that these guarantees hold over the restricted domain $\Theta_M=[-a/M,a/M]$, with fixed $0<a<\pi/2$, as specified in Eq.~\eqref{eq:sample-success}. This domain applies to all PAC learning tasks considered below and will not be restated for each result.

 We now consider a broader class of classical states, where the above results hold. From the monotonicity of the QFI $\QFI(\Lambda(\mathcal E_{\delta,M}(\ket{\alpha})))\leq\QFI(\mathcal E_{\delta,M}(\ket{\alpha}))$ for any parameter independent channel $\Lambda$. Furthermore, if $\Lambda$ is phase covariant we have $\QFI(\mathcal E_{\delta,M}(\Lambda(\ket{\alpha})))\leq\QFI(\mathcal E_{\delta,M}(\ket{\alpha}))$. Therefore, the QFI results for a coherent state proven above are upper bounds for the set of states generated by a parameter independent phase covariant channel acting on $\ket{\alpha}$. As we wish to consider noisy channels later in the manuscript, we shall restrict $\Lambda$ to represent only loss, thermal noise, phase diffusion, and phase rotations, together with classical mixtures of such channels; for multi‑round protocols the channel may act independently on each channel use, with classically correlated or adaptively chosen parameters. This restriction ensures that, in the presence of the thermal-loss channel considered below, every allowed probe output can be obtained by parameter-independent processing of the corresponding coherent-probe output. For the remainder of this work, we shall refer to this restricted set of states as classical states:
\begin{equation}
\label{eq:classicaldef}
  \rho_{\rm cl}\in\big\{\Lambda(\ket\alpha)\;:\;|\alpha|^2=\bar n,\ \Lambda\ \text{as above}\big\} 
\end{equation}
 
Under this definition $\bar n$ refers to the mean photon number of the generating state per channel use. Note that this definition explicitly forbids strategies such as only rarely sending a very high-energy coherent state. For example, the state $\rho=(1-p)\ket{0}\bra{0}+p\ket{\tilde{\alpha}}\bra{\tilde{\alpha}}$, has mean photon number of $p|\tilde{\alpha}|^2=|\alpha|^2$. For sufficiently small $p$, we can have $|\tilde{\alpha}|^2\gg|\alpha|^2$, and so this state does not suffer the same exponential decay as in Eq.~\eqref{eq:qfiubCoh3}. States of this form have been studied before~\cite{rivas2012sub,zhang2012unbounded,branford2021average,giovannetti2012sub}. There are two physically motivated reasons for ruling out states of this type. First, in applications with fragile samples such as biosensing~\cite{casacio2021quantum}, the relevant constraint is on the energy of every individual state, not only on the average over states~\cite{branford2021average}. Second, allowing a mean photon number conserving strategy that involves only occasionally sending bright squeezed vacuum states yields an unbounded QFI at fixed mean energy~\cite{rivas2012sub,zhang2012unbounded}. This suggests that a mean photon number constraint is not well motivated physically.

 Similarly, as we show in Appendix~\ref{apen:samplecomplexity}, the minimum number of samples for a $(\epsilon,\zeta)$ accurate estimate satisfies $\nu_{\text{cl}}^\star=\nu^\star_{\text{C}}$, where $\nu_{\text{cl}}^\star$ is the minimum samples required over all allowed classical protocols. Going forward we shall use $\nu_{\text{cl}}^\star$ to denote the sample complexity of the optimal classical protocol. Note that this optical classical protocol is still restricted to the family of classical states defined in Eq.~\eqref{eq:classicaldef}.
 
We now examine the benefits of introducing squeezing, subject to the same constraints on mean photon number per channel use, i.e. we do not allow occasionally sending a very bright squeezed state. For simplicity, here we discuss squeezed vacuum only. For details on displaced squeezed states, we refer the reader to Appendix~\ref{apen:maryS}. A squeezed vacuum state may be written as
\begin{equation}
    \ket{r}=\sum_{n=0}^\infty c_{n}^{\mathrm{SV}}\ket{2n}\;,
\end{equation}
with 
\begin{equation}
\label{eq:svcoeffsmain}
    c_{n}^{\mathrm{SV}}=\frac{1}{\sqrt{\cosh(r)}}(-\mathrm{e}^{i\phi}\text{tanh}(r))^n\frac{\sqrt{(2n)!}}{2^nn!}\;.
\end{equation}
In Appendix~\ref{apen:maryS}, we prove that, with a squeezed vacuum state, the QFI is lower bounded as 
\begin{equation}
    \QFI ^{\mathrm{SV}}
    \ge
    4M^2
    \sum_{\substack{0\le j<M\\ j\,\mathrm{even}}}
    \frac{p_j^{\mathrm{SV}}p_{j+M}^{\mathrm{SV}}}{\lambda_j}\;,
\end{equation}
where $p_j^{\mathrm{SV}}=|c_{j}^{\mathrm{SV}}|^2$. From this, we show that if the mean photon number grows as $\bar n(M)=M/\log (M)$, 
\begin{equation}
    \QFI^{\mathrm{SV}}
\ge 
    c\frac{M^{3/2}}{\sqrt{\log M}}\;,
\end{equation}
where $\QFI^\text{SV}=\QFI(\mathcal E_{\delta,M}(\ket{r}))$ and $c$ is some constant for sufficiently large $M$ (see Eq.~\eqref{eq:SV-strong-lower}). Similarly, the number of samples required for a $(\epsilon,\zeta)$ accurate estimate satisfies (see Appendix~\ref{apen:samplecomplexity})
\begin{equation}
\nu_{\mathrm{SV}}^\star
    =O\!\left(
        \sqrt{M\log M}\,
        \log\!\frac2\zeta
    \right).
\end{equation}

\begin{figure*}[t]
    \centering
    \includegraphics[width = \textwidth]{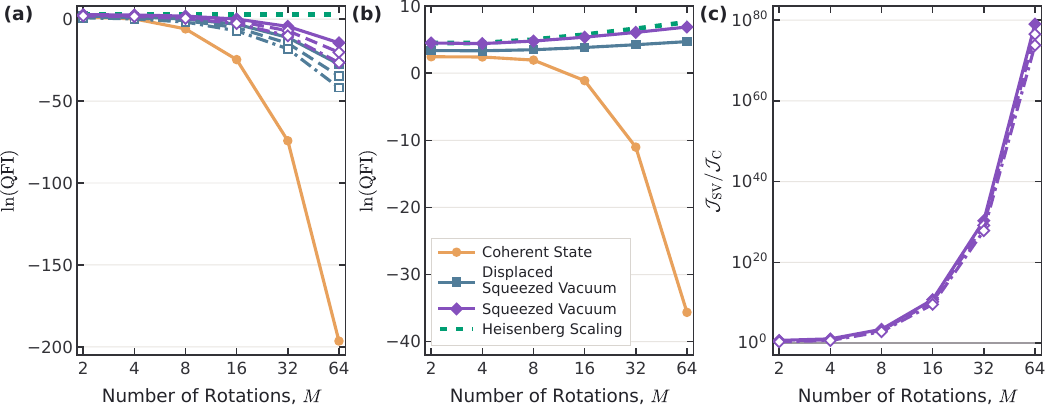}
    \caption{\textbf{Quantum Fisher information for phase estimation under phase twirling.} (a) Fixed mean photon number of 1. For the coherent state $\alpha=1$, and for the displaced squeezed state $\alpha=1/\sqrt{2}$, $r=\text{arcsinh}(1/\sqrt{2})$. Dashed and dot-dashed lines correspond to the QFI in the presence of loss with $\eta=0.9$ and $\eta=0.8$, respectively. $\bar{n}_\text{th}=0.1$ for both lines.  (b) Mean photon number scales as $\bar{n}=M/\log(M)$. In both plots, Heisenberg scaling corresponds to the QFI of a squeezed vacuum state when $M=1$, i.e $8\bar{n}(\bar n+1)$~\cite{nielsen2023deterministic}. (c) Ratio of quantum Fisher information with squeezed vacuum and coherent states from (a).     }
    \label{fig:mpsk}
\end{figure*}

With this restriction on mean photon number, we then arrive at our first main result
\begin{theorem}
\label{th1}
    Given mean photon number of $\bar{n}=M/\log(M)$, there is a super-exponential separation in the QFI of a squeezed vacuum and the family of classical optical states $\rho_\text{cl}$, as defined in Eq.~\eqref{eq:classicaldef}, for phase estimation under the action of the phase twirling channel $\mathcal E_{\delta,M}$. Specifically $\QFI^{\mathrm{SV}}=\Omega(M^{3/2}/\sqrt{\log(M)})$, and $\sup_{\rho\in\rho_\text{cl}}\QFI\leq \exp[-M\log\log M+O(M)]$\;. Furthermore, there is a super-exponential separation in the sample complexity required to solve the corresponding $(\epsilon,\zeta)$ PAC learning task (for $\epsilon=b/M$ on the domain $\Theta_M=[-a/M,a/M]$, with fixed $0<a<\pi/2$ and $0<b<a/2$), $\nu_\text{SV}^\star=O(\sqrt{M\log(M)}\log(2/\zeta))$, whereas $\nu_{\text{cl}}^\star\geq\exp(M\log\log(M)-O(M))$.
\end{theorem}

While sublinearly growing photon number ($\bar{n}=M/\log(M)$) is feasible for small $M$, most experiments operate with a fixed maximum squeezing level. As such, we next examine the ratio of $\QFI^{\mathrm{SV}}$ to $\sup_{\rho\in\rho_\text{cl}}\QFI$ for a fixed $\bar{n}$. We write 
\begin{equation}
\label{eq:qfiubCoh2}
\sup_{\rho\in\rho_\text{cl}}\QFI\leq\frac{4M^2\mathrm{e}^{-\bar{n}}\bar{n}^{M}}{M!(1-\bar{n}\frac{M+1}{M^2})}\leq \frac{8M^2\mathrm{e}^{-\bar{n}}\bar{n}^{M}}{M!}\;,
\end{equation} 
where we use the fact that $4\bar{n}\leq M$ implies $(1-\bar{n}\frac{M+1}{M^2})\geq1/2$. A similar result holds for the sample complexity to solve the PAC learning tasks. This allows us to arrive at the following theorem.
\begin{theorem}
\label{th2}
    Given mean photon number of $\bar{n}=1$, the ratio $\frac{\QFI^{\mathrm{SV}}}{\sup_{\rho\in\rho_\text{cl}}\QFI}$ grows super-exponentially in $M$ as
   \begin{equation}
        \frac{\QFI^{\mathrm{SV}}}
        {\sup_{\rho\in\rho_\text{cl}}\QFI}
        =
        \Omega\!\left[
        \left(\frac{M}{\sqrt{2}\mathrm{e}}\right)^M
        \right].
    \end{equation}
    For the corresponding $(\epsilon,\zeta)$ PAC learning task with $\epsilon=b/M$ on the domain $\Theta_M=[-a/M,a/M]$, with fixed $0<a<\pi/2$ and $0<b<a/2$, the ratio $\frac{\nu_{\text{cl}}^\star}{\nu^\star_\text{SV}}$ also grows as
    \begin{equation}
        \frac{\nu_{\text{cl}}^\star}{\nu^\star_\text{SV}}=\Omega\!\left[
        \left(
            \frac{M}{\sqrt{2}e}
        \right)^M
    \right]\;.
    \end{equation}
\end{theorem}
This is proven in Appendices~\ref{apen:constantnbar} and \ref{apen:samplecomplexity}. In Appendix~\ref{sec:additional-sv-optimality-notes}, we provide the motivation for restricting the main theorems to squeezed vacuum states, proving that among all fixed mean photon number displaced squeezed states the squeezed vacuum is locally optimal. We further prove that, for $\bar n=M/\log(M)$, as $M\to\infty$ the optimal displaced squeezed state has a vanishing fraction of its photons in displacement.

\subsection{Comparison of squeezed and coherent states}
 In Fig.~\ref{fig:mpsk} (a) and (b), we plot the QFI of a coherent state, a displaced squeezed state, and a squeezed vacuum state. For a fixed constant photon number, this is shown in Fig.~\ref{fig:mpsk}(a). Although the QFI decays exponentially for all three classes of states, the decay is considerably slower when squeezing is used, demonstrating Theorem~\ref{th2}. In Fig.~\ref{fig:mpsk}(b), we show that there is no decay in $M$ for the squeezed vacuum QFI when the mean photon number scales as $M/\log(M)$, demonstrating Theorem~\ref{th1}. Using a simple gradient descent algorithm, we also numerically optimise the ratio of displacement to squeezing for the displaced squeezed state. For larger photon numbers, we use these numerics to verify that the squeezed vacuum state is optimal. As is apparent from Fig.~\ref{fig:mpsk} (c), the ratio of the QFI with squeezing to that without is growing exponentially. The origin of this advantage is evident in the Wigner functions of the output quantum states, shown in Fig.~\ref{fig:mpskscheme}(e). After mixing with many different rotation angles, the coherent state carries almost no information about $\delta$. In contrast, the displaced squeezed state still carries information about $\delta$. Crucially, this exponential advantage does not imply precision beyond Heisenberg scaling. It arises because the squeezed vacuum retains substantially more phase information under twirling than a coherent state with the same mean photon number. The twirling channel preserves only coherences between Fock states whose photon numbers differ by integer multiples of $M$ (Eq.~\eqref{eq:QGIgeneralVar}). The broader photon-number distribution of the squeezed vacuum supports these widely separated coherences (Fig.~\ref{fig:mpskscheme} (d)), leading to the exponential separation in QFI.

\subsection{Accounting for experimental imperfections}
In any realistic experimental setting, there will be loss and thermal noise affecting the input states. As such, we now consider the scenario where the initial displaced squeezed state suffers loss, modelled as a beam splitter of transmissivity $\eta$. We model thermal noise as a thermal state with mean photon number $\bar{n}_\text{th}$ input to the second port of the beam splitter. We write the action of this channel after the phase twirling given an input state $\ket{\psi}$ as $\mathcal G_{\eta,\bar n_{\rm th}}\circ
 \mathcal E_{\delta,M}(\ket{\psi})$. This is described in detail in Appendix~\ref{apen:qfiloss}, along with a description of how the QFI is numerically calculated. As shown in Fig.~\ref{fig:mpsk}(a) and Fig.~\ref{fig:mpsk}(c), even in the presence of moderate noise and loss, there is still a significant advantage of using squeezed light.

In Appendix~\ref{apen:classicalnoisybound}, we bound the QFI attainable with classical states $\rho_\text{cl}$ 
under the above optical loss and thermal noise channel. Specifically we show that, for $\eta \bar{n}<M^2/(M+1)$,
\begin{equation}
 \sup_{\rho_{\rm cl}}
 \QFI\!\left[
 \mathcal G_{\eta,\bar n_{\rm th}}\circ
 \mathcal E_{\delta,M}(\rho_{\rm cl})\right]
 \le  \frac{4M^2e^{-\eta \bar{n}}(\eta \bar{n})^M}
 {M!\left[1-\eta \bar{n}(M+1)/M^2\right]}\;.
\end{equation}
Similarly, for this noise model, we prove in Appendix~\ref{apen:samplethermal1} that the number of samples required to produce a $(\epsilon,\zeta)$ accurate estimate is
    \begin{equation}
    \nu_{\rm cl}^{\star}
    \ge
    \frac{(1-2\zeta)^2}{16\epsilon^2M^2}
    \left(1-\eta \bar{n}\frac{M+1}{M^2}\right)
    e^{\eta \bar{n}}\frac{M!}{(\eta \bar{n})^M}\;.
\end{equation}

In Appendix~\ref{apen:QFIsqnoisybound}, we bound the QFI attainable when using squeezed vacuum as
\begin{equation}
\label{eq:QFInoise}
    \QFI^\text{SV}
    \!\left[
        \mathcal G_{\eta,\bar n_{\rm th}}
        \circ\mathcal E_{\delta,M}
        (\ket r\!\bra r)
    \right]
    \ge
    \frac{16M^2}{D}
    \binom{M}{M/2}
    \left(
        \frac{v_2-v_1}{2D}
    \right)^M\;,
\end{equation}
where 
\begin{equation}
\label{eq:covariancemain}
    \begin{split}
        v_1&=(1+2\bar{n}_\text{th})(1-\eta)+\eta\mathrm{e}^{-2r},\\
        v_2&=(1+2\bar{n}_\text{th})(1-\eta)+\eta\mathrm{e}^{2r}\;
    \end{split}
\end{equation}
and $D=(1+v_1)(1+v_2)$. Similarly, for the PAC learning task, we show in Appendix~\ref{apen:samplethermal1} that
\begin{equation}
    \nu_{\rm SV}^{\star,\eta,\bar n_{\rm th}}
    \le
    \left\lceil
    \frac{C_{a,b,\zeta}D}
         {16\binom{M}{M/2}}
    \left(
        \frac{D}
             {2\eta\sqrt{\bar{n}(1+\bar{n})}}
    \right)^M
    \right\rceil,
\label{eq:samplenoisy}
\end{equation}
where $C_{a,b,\zeta}$ is defined in Eq.~\eqref{eq:sampleloss-C} and gives the $(\epsilon,\zeta)$ dependence.

In Appendix~\ref{apen:QFIsqnoisybound}, we show that, under constant $\eta$ and $\bar n_\text{th}$, the bounds that resulted in Theorem~\ref{th1} no longer hold. As such, we cannot certify a similar constant versus exponentially decaying QFI advantage (see Eq.~\eqref{eq:apensqexpdecay}). The best performing bound that we obtained with $\bar n=M/\log(M)$ was $\QFI^{\text{SV}}\geq\exp(-\sqrt{A/\eta}\sqrt{M\log(M)}(1+o(1)))$, see Eqs.~\eqref{eq:exact-heterodyne-CFI} and \eqref{eq:xiM-general-scaling}, where $A$ depends only on $\eta$ and $\bar{n}_{\text{th}}$ (Eq.~\eqref{eq:sampleloss-A}). A similar result holds for the PAC sample complexity, see Appendix~\ref{sec:samplesthermalloss}. These bounds therefore certify that the QFI and sample complexity ratios remain super-exponential under fixed optical loss and thermal noise, although they no longer certify polynomial squeezed-vacuum sample complexity at $\bar n=M/\log M$. Similarly, for constant mean photon number, the advantage from Theorem~\ref{th2} remains. Specifically we have the following theorem

    \begin{theorem}
\label{th3}
    Given mean photon number of $\bar{n}=1$, under constant loss $0<\eta\leq1$ and thermal noise $\bar n_\text{th}$, the ratio $\frac{\QFI^{\mathrm{SV}}}{\sup_{\rho\in\rho_\text{cl}}\QFI}$ grows factorially in $M$ as
   \begin{equation}
        \frac{\QFI^{\text{SV}}
    \!\left[
        \mathcal G_{\eta,\bar n_{\rm th}}
        \circ\mathcal E_{\delta,M}
        (\ket r\!\bra r)
    \right]}
        { \sup_{\rho_{\rm cl}}
 \QFI\!\left[
 \mathcal G_{\eta,\bar n_{\rm th}}\circ
 \mathcal E_{\delta,M}(\rho_{\rm cl})\right]}
       =\Omega 
   \left[ \left(
        \frac{4\sqrt{2}M}{D\text{e}}
    \right)^M\right].
    \end{equation}
    Similarly, for the $(\epsilon,\zeta)$ PAC learning task with $\epsilon=b/M$ on the domain $\Theta_M=[-a/M,a/M]$, with fixed $0<a<\pi/2$ and $0<b<a/2$,
    the ratio $\frac{\nu_{\text{cl}}^\star}{\nu^\star_\text{SV}}$ grows as
    \begin{equation}
        \frac{\nu_{\text{cl}}^\star}{\nu^\star_\text{SV}}=\Omega\!
     \left[   \left(
            \frac{4\sqrt{2}M}{D\text{e}}
        \right)^M\right]
    .
    \end{equation}
\end{theorem}
This is proven in Appendices~\ref{apen:QFIsqnoisybound} and \ref{apen:samplenoisyfixedenergy}. Note that this does not mean that, for a fixed noisy channel, the squeezed vacuum probe outperforms any state in the classical family of Eq.~\eqref{eq:classicaldef}, this is only true for sufficiently large $M$. 

Theorem~\ref{th3} reveals that a super-exponential QFI and sample complexity ratio can survive even in the presence of loss and thermal noise. However, the QFI and sample complexity with squeezed vacuum decrease exponentially and increase exponentially, respectively. To prove this, note that the noisy QFI cannot be more than the QFI in the ideal setting, which by Eq.~\eqref{eq:fixedE-sv-upper} decays as $\tanh^M(r)$. It is therefore natural to ask how the mean photon number must grow with $M$ for constant noise parameters such that the QFI and sample complexity do not decay in $M$. We answer this question in the following theorem:
\begin{theorem}
\label{thm:high-energy-loss-separation}
Fix a transmissivity
$0<\eta\le1$ and thermal noise $\bar n_\text{th}\geq0$ independent of $M$. Let $\bar{n}$ satisfy
\begin{equation}
\label{eq:strangephotonconstraint}
   \frac{M^2}{(\log M)^2}\ll \bar{n}\ll \frac{M^2}{\log M}
\end{equation}
Let $\ket{r}=S(r)\ket0$ be the squeezed vacuum satisfying $\bar{n}=\sinh^2r$. Under the uniform $\mathbb Z_M$ phase twirl, the squeezed vacuum QFI satisfies $\QFI^{\rm SV}
    \geq
    M^{2-o(1)}$, and
   \begin{equation}
        \frac{\QFI
    \!\left[
        \mathcal G_{\eta,\bar n_{\text{th}}}
        \circ\mathcal E_{\delta,M}
        (\ket r\!\bra r)
    \right]}
        { \sup_{\rho_{\rm cl}}
 \QFI\!\left[
 \mathcal G_{\eta,\bar n_{\text{th}}}\circ
 \mathcal E_{\delta,M}(\rho_{\rm cl})\right]}
       = M^{\omega(1)}\;,
    \end{equation}
where $M^{\omega(1)}$ denotes a quantity larger than $M^{k}$ for
every fixed $k>0$. Furthermore the number of samples required to achieve a $(\epsilon,\zeta)$ accurate estimate with $\epsilon=b/M$ on the domain $\Theta_M=[-a/M,a/M]$, with fixed $0<a<\pi/2$ and $0<b<a/2$ using squeezed vacuum satisfies $\nu_{\rm SV }^{\star}\le M^{o(1)}$, and
\begin{equation}
    \frac{\nu_{\rm cl}^{\star}}
         {\nu_{\rm SV}^{\star}}
    \ge
    \exp\!\left[\Omega\!\left(\frac{M^2}{\bar n}\right)\right]
    =M^{\omega(1)}.
\end{equation}
\end{theorem}
This is proven in Appendices~\ref{sec:high-energy-loss-separation} and \ref{sec:sampleloss-high-energy}. While we admit that the mean photon number constraint in Eq.~\eqref{eq:strangephotonconstraint} is unusual, we have tuned the constraint to demonstrate an interesting physical effect. Many quantum advantages in sensing or learning follow a different functional form in the presence of noise~\cite{Huelga1997,dorner2009optimal,Escher2011,demkowicz2012elusive,oh2024entanglement,seif2026entanglement}. However, by choosing $\bar{n}=M^2/\log(M)^{3/2}$ the photon number distributions under constant loss for both the squeezed vacuum state and the classical family of states are tuned just to the point where we observe a sub-polynomial versus super-polynomial separation. 

Finally, we consider a phase diffusion channel parameterised by $\sigma$ that acts as $  \mathcal D_\sigma(\ket n\!\bra m)=e^{-\sigma^2(n-m)^2/2}\ket n\!\bra m$. 
\begin{theorem}
\label{thm:phase-diffusion-no-go}
Let $\sigma_M>0$ and let
$\Phi_{\delta,M,\sigma_M}$ denote the channel obtained by composing
the uniform $\mathbb Z_M$ phase twirl, wrapped Gaussian phase
diffusion of variance $\sigma_M^2$, and the unknown phase rotation
$\delta$. Then, if $M\sigma_M\to\infty$, 
\begin{equation}
    \sup_\rho
    \QFI
    \left[
        \Phi_{\delta,M,\sigma_M}(\rho)
    \right]
    \le
    2M^2e^{-M^2\sigma_M^2}
    (1+o(1)).
\end{equation}
Consequently, $M^2\sigma_M^2-2\log M\longrightarrow+\infty$
implies
\begin{equation}
    \sup_\rho
    \QFI
    \left[
        \Phi_{\delta,M,\sigma_M}(\rho)
    \right]
    \longrightarrow0.
\end{equation}
In particular, this holds for every fixed $\sigma_M=\sigma>0$,
independently of the input state's mean photon number or photon-number
variance.
\end{theorem}
This is proven in Appendix~\ref{sec:phase-diffusion-no-go}. Remarkably, this result is independent of mean photon number. This proves that phase diffusion noise provides a fundamental block to asymptotically achieving constant QFI for phase estimation under the $\mathbb{Z}_M$ phase twirling channel. In order to maintain constant QFI, it is sufficient to have $\sigma_M^2=O(M^{-2})$. A necessary condition is $\sigma_M^2= O(\log(M)/M^{2})$.

\section{Scalable advantage}
\label{sec:implementation}
Above we have computed the QFI in various different settings under various restrictions. However, for practical application of these results, we require a physical implementation. Here we propose physical implementations in two different regimes. We first consider near-term quantum processors whereby we have a fixed mean photon number and noisy state preparation and measurements. Second, we consider a future setting where bosonic quantum processors can be used to prepare and measure the states of interest.

\subsection{Near-term quantum processors}
\label{sec:neartermprocessor}
In the near term, it is not reasonable to expect quantum processors capable of coherently controlling arbitrarily many photons. As such, in this subsection we focus on input states with constant mean photon number and experimentally feasible measurements. For simplicity we focus on $\bar{n}=1$.

The optimal measurement for saturating the QFI for estimating a single parameter is known~\cite{helstrom1969quantum,braunstein1994statistical} (for multiple parameters, the situation is more subtle~\cite{conlon2025role,ragy2016compatibility}). In Appendix~\ref{sec:additional-null-sv}, we describe the optimal measurement in detail. After imposing a finite photon number cutoff, the optimal measurement is a finite-dimensional projective measurement. In principle, any such projective measurement can be implemented on an oscillator~\cite{krastanov2015universal}. However, in practice, this measurement may be complex to implement. As such, here we focus on more experimentally friendly measurements.

\begin{figure}[t]
    \centering
    \includegraphics[width =0.5\textwidth]{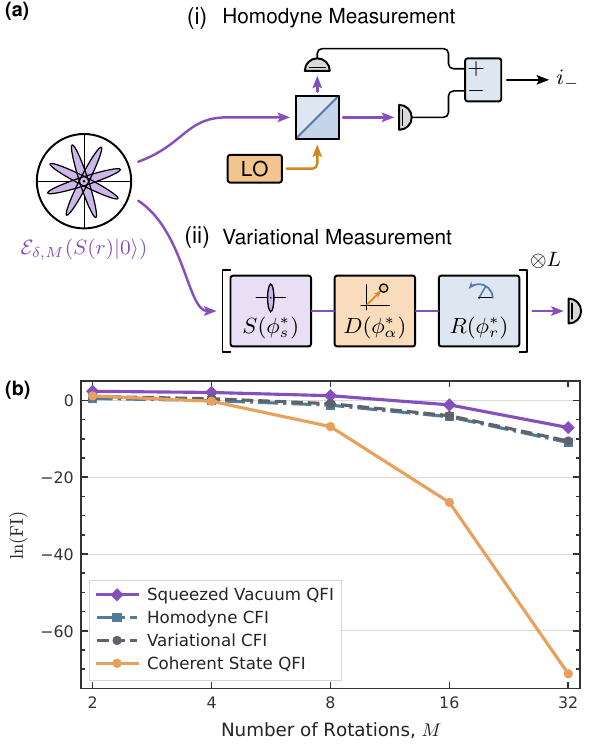}
    \caption{\textbf{Feasible measurement implementation.} (a \romannum{1}) Schematic showing homodyne detection. The input state is mixed on a 50:50 beamsplitter with a bright local oscillator (LO). (a \romannum{2}) The photonic circuit used to implement a close to optimal measurement. $L$ is the number of layers (number of times the circuit is repeated), $S,D$ and $R$ denoting a squeezing, displacement and rotation operation respectively, and $\phi^*$ denotes the numerically optimised angles. (b) CFI of the measurement strategies shown in (a \romannum{1}) and (a \romannum{2}), plotted alongside the QFI of a squeezed vacuum state and coherent state. We assume $\bar{n}_\text{th}=0.1$, $\eta=0.9$, $\alpha=1$ for the coherent state, and for the squeezed vacuum state $r=\text{arcsinh}(1)$. The variational measurement assumes $n_\text{max}=5$.}   \label{fig:optmeas}
\end{figure}

We first consider homodyne detection, accessible to most labs. We require only balanced homodyne detection at the fixed local-oscillator phase $\vartheta_\star=\pi/(2M)$, shown in Fig.~\ref{fig:optmeas}(a.\romannum{1}). Each run of the experiment produces a real quadrature
outcome $x_j$. The experimentalist then produces the statistic $Y_j:=\cos(t_\star x_j)$, which is evaluated entirely in classical post-processing ($t_\star=\sqrt{2M}$ is chosen in Appendix~\ref{sec:additional-homodyne-sv}). Finally,
\begin{equation}
    \widehat\Phi(t_\star)
    :=
    \frac1\nu\sum_{j=1}^{\nu}Y_j
\end{equation}
estimates
$\Phi_\delta(t_\star)=
\mathbb E_\delta[\cos(t_\star X_{\vartheta_\star})]$.
Near the operating point $\delta=0$, the phase can therefore be
estimated from the calibrated linear response
\begin{equation}
    \hat\delta
    =
    \frac{
        \widehat\Phi(t_\star)-\Phi_0(t_\star)
    }{
        \left.
        \partial_\delta\Phi_\delta(t_\star)
        \right|_{\delta=0}
    }.
\end{equation}
Thus the nontrivial part of the receiver is classical
post-processing rather than non-Gaussian optical measurement. As we prove in Appendix~\ref{sec:additional-homodyne-sv},
\begin{equation}
    \log
    \frac{\mathcal J^{\text{SV}}_{\rm hom}}
    {\displaystyle
    \sup_{\rho_{\rm cl}}
    \QFI}
    \ge
    M\log M-O(M)-M^2\sigma_M^2,    
\end{equation}
where $\mathcal J^{\text{SV}}_{\rm hom}$ is the CFI of homodyne detection with squeezed vacuum. Thus, without phase diffusion ($\sigma_M=0$), the super-exponential advantage survives even under constant loss. The drawback is that, although the ratio above grows super-exponentially, $\mathcal J^{\text{SV}}_{\rm hom}$ also decays exponentially in $M$, necessitating exponentially more samples to reach the same precision. Nevertheless, given the experimental accessibility of squeezed states and homodyne detection, this provides a path to demonstrating super-exponential advantage over the family of classical states defined in Eq.~\eqref{eq:classicaldef} requiring only modest resources as shown in Fig.~\ref{fig:optmeas}(b).

We next examine how variational circuits can be used to implement a close-to-optimal measurement, as shown in Fig.~\ref{fig:optmeas}(a \romannum{2}). We consider a variational circuit consisting of $L$ layers of squeezing, displacement and rotation operations followed by a photon counting measurement. We model the photon counting measurement as being characterised by the maximum number of detectable photons, $n_\text{max}$. Photon numbers greater than $n_\text{max}$ are assumed to be counted as $n_\text{max}$ photons. Based on this, we can compute the CFI numerically (see Appendix~\ref{apen:variational-receiver} for more detail), shown in Fig.~\ref{fig:optmeas}(b). This offers another potential near-term experimental realisation of this exponential advantage.

\subsection{Growing-energy implementations}

We finally consider implementations in which the source mean photon number can increase with $M$. Suppose that preparation, transmission, and detection imperfections are described by a fixed number of thermal attenuators whose parameters are independent of $M$, with nonzero total transmissivity. Their composition is another thermal attenuator with fixed effective parameters, so
Theorem~\ref{thm:high-energy-loss-separation} applies. For example, with
$\bar n=M^2/(\log M)^{3/2}$, squeezed-vacuum probes and heterodyne detection require at most $M^{o(1)}$ samples, whereas every protocol in the classical family of Eq.~\eqref{eq:classicaldef} requires $M^{\omega(1)}$ samples. The same asymptotic separation survives phase diffusion satisfying $\sigma_M^2=O(M^{-2})$. The implementation therefore requires increasing squeezing while maintaining controlled loss, thermal noise, and phase stability.

\section{Discussion}
\label{sec:discussion}
The example in this work is a well-studied phase estimation problem evaluated under a new specially-structured noise model. As such, the untwirled squeezed vacuum QFI of $8\bar n(\bar n+1)$ places an upper bound on the attainable QFI~\cite{giovannetti2011advances}. In spite of this, we show how simple modifications to the phase estimation problem can lead to super-exponentially growing differences between the QFI of a squeezed vacuum and the family of classical states defined in Eq.~\eqref{eq:classicaldef}.  This does not circumvent the Heisenberg limit: the twirling channel does not enhance the phase encoding, but instead degrades the sensitivity of the family of classical probes considered here far more strongly than squeezed vacuum probe states. Similarly, the robustness to constant noise proven in Theorem~\ref{thm:high-energy-loss-separation} appears to evade limitations of previous works~\cite{Huelga1997,dorner2009optimal,Escher2011,demkowicz2012elusive,oh2024entanglement,seif2026entanglement}. Again, however, this robustness arises because the twirling channel
suppresses the classical benchmark more strongly, rather than because
the phase signal is error-corrected.

For the problem considered here, the parameter of interest is encoded in a random manner, every time the channel is used. This is similar in nature to Refs.~\cite{oh2024entanglement,cotler2026quantum,prabhu2026exponential,kannan2026exponential2}. It is particularly informative to view our task through the framework of Ref.~\cite{kannan2026exponential2} where, for several different resource families, an exponential advantage in the learning of Fourier coefficients was demonstrated. In our work, after the phase twirl the total rotation angle is $\phi\in\{\delta+2\pi k/M\}_{k=0}^{M-1}$ with uniform weights. The angular Fourier coefficients of $\phi$ are
\begin{equation}
    \mathbb E\!\left[e^{\mathrm i\ell\phi}\right]
    =\frac1M\sum_{k=0}^{M-1}e^{\mathrm i\ell(\delta+2\pi k/M)}
    =\begin{cases}
        e^{\mathrm i\ell\delta}, & \ell\in M\mathbb Z\\
        0, & \text{otherwise}
    \end{cases}\;.
\end{equation}
Hence, estimating $\delta$ under the phase twirl is equivalent to estimating the phase of the $M$-th angular Fourier coefficient of the $\phi$ distribution. The two settings nevertheless differ in a meaningful manner. First, our parameter is encoded by a rotation rather than a displacement. Second, here the hardness for probes from the classical family in Eq.~\eqref{eq:classicaldef} is created by a parameter-independent noise channel acting on a standard phase-estimation problem rather than by the choice of a learning objective. It would be interesting in future to examine whether directly estimating the $M$-th order Fourier coefficient can display an exponential advantage without the phase twirling channel under a different encoding.

There are also similarities with Ref.~\cite{liu2024exponential}, where exponential advantages in learning were identified by gaining access to a purification of the system. For instance, if one knows the index $k$ in every round of the experiment, then the problem discussed in this paper reduces to regular phase estimation where the advantage of squeezed vacuum over coherent states is limited to quadratic. There has been significant recent interest in connecting quantum metrology and learning~\cite{meyer2025quantum,chen2026instance,kwon2026universal}. The strong overlap between our results for PAC learning and quantum metrology complement these known results.

In the future, the task studied here may find relevance for quantum communication networks---e.g.~for calibrating a $M$-phase shift keying (PSK) device. PSK is an important technique for encoding information in continuous variable (CV) quantum key distribution (QKD)~\cite{wang2023discrete,bauml2024security,kanitschar2022optimizing}. An $M$-PSK device prepares one of $M$ different quantum states, each of which corresponds to some unique classical information. Unfortunately, imperfections in state preparation can render a QKD protocol insecure~\cite{fung2007phase,peng2025practical}. As such, it is of great importance to ensure devices used in QKD are well calibrated. The task considered here is analogous to the calibration of the state preparation device for $M$-PSK.

\section{Acknowledgments}

L.O.C., E.A., and A.V.G.~were supported in part by ONR MURI, AFOSR MURI, ARL (W911NF-24-2-0107), DoE ASCR Quantum Testbed Pathfinder program (award No.~DE-SC0024220), NSF QLCI (awards No.~OMA-2120757 and No.~OSI-2553574), NSF STAQ program, and NQVL:QSTD:Design:FTL. L.O.C., E.A., and A.V.G.~also acknowledge support from the U.S.~Department of Energy, Office of Science, National Quantum Information Science Research Centers, Quantum Systems Accelerator (award No.~DE-SCL0000121) and from the U.S.~Department of Energy, Office of Science, Accelerated Research in Quantum Computing, Fundamental Algorithmic Research toward Quantum Utility (FAR-Qu). V.V. is supported by the National Research Foundation, Singapore, through the National Quantum Office, hosted in A*STAR, under its National Quantum Scholarship Scheme (NQSS) Funding Initiative (PhD). S.R.M. is supported by the NSF QLCI award OMA-2120757. J.B. notes that the views expressed in this work are those of the author and do not reflect the official policy of the United States Naval Academy or any department of the United States government.

 \appendix
 
\begin{widetext}
\newpage
  \section*{Appendix Table of Contents}

\startcontents[appendices]

\printcontents[appendices]{}{1}{\setcounter{tocdepth}{2}}
\vspace{10cm} 
 \clearpage

\section{QFI calculations in ideal setting}
\label{apen:mary}
In this appendix we first derive an expression for the QFI for arbitrary pure states in Appendix~\ref{apen:mary1}. In Appendices~\ref{apen:maryC} and \ref{apen:maryS} we specialise to coherent states and squeezed vacuum respectively, before comparing the two in Appendix~\ref{apen:constantnbar}.

\subsection{General expression for QFI}
\label{apen:mary1}

We begin by writing our input state in the Fock basis as
\begin{equation}
    \ket{\psi}=\sum_{n=0}^\infty c_{n}\ket{n}\;.
\end{equation}
For now we don't specify the input state. After the twirling channel in Eq.~\eqref{eq:sample-channel} the state becomes
\begin{equation}
\label{eq:inputrhopure}
\begin{split}
        \rho&=\mathcal T_M(\ket{\psi})\\
        &=\frac{1}{M}\sum_{k=1}^{M}\sum_{n,p}c_{n}c_p^*\mathrm{e}^{\mathrm{i}2\pi k(n-p)/M}\ket{n}\bra{p}\\
        &=\frac{1}{M}\sum_{n,p}c_{n}c_p^*\sum_{k=1}^{M}\mathrm{e}^{\mathrm{i}2\pi k(n-p)/M}\ket{n}\bra{p}\\
         &=\sum_{n,p}c_{n}c_p^*\delta_{(n-p)\text{ mod }M,0}\ket{n}\bra{p}
\end{split}
\end{equation}
Here $\delta_{(n-p)\text{ mod }M,0}$ is to be interpreted as Kronecker delta, distinct to the rotation $\delta$. Note that we have not included the  rotation $\delta$ yet. This acts to give the final state $ \mathcal E_{\delta,M}(\ket{\psi})=\mathrm{e}^{\mathrm{i}\hat{n}\delta}\rho\mathrm{e}^{-\mathrm{i}\hat{n}\delta}$ so that
\begin{equation}
        \frac{\partial\mathcal E_{\delta,M}(\ket{\psi})}{\partial\delta}\bigg|_{\delta=0}=i\left[\hat{n},\mathcal E_{0,M}(\ket{\psi})\right]=i\left[\hat{n},\rho\right]\;.
\end{equation}
The standard expression for the QFI is
\begin{equation}
\begin{split}
    \QFI=2\sum_{j,k}\frac{|(\partial_\delta \mathcal E_{\delta,M}(\ket{\psi}))_{jk}|^2}{\lambda_j+\lambda_k}=2\sum_{j,k}\frac{|\bra{\tilde{u}_j}[ \mathcal E_{0,M}(\ket{\psi}),\hat{n}]\ket{\tilde{u}_k}|^2}{\lambda_j+\lambda_k}\;,
\end{split}
\end{equation}
where $\lambda_k$ and $\ket{\tilde{u}_k}$ are the eigenvalues and eigenvectors of $\mathcal E_{0,M}(\ket{\psi})$ respectively. The sum is only over terms such that $\lambda_j+\lambda_k\neq0$. We can write $\mathcal E_{0,M}(\ket{\psi})$ in the following direct sum structure (note that we only explicitly write the non-zero eigenvalues)
\begin{equation}
\label{eq:directsumstructure}
\begin{split}
    \mathcal E_{0,M}(\ket{\psi})&=\bigoplus_{r=0}^{M-1}\rho_r\\
    &=\bigoplus_{r=0}^{M-1}\sum_{p,q=0}^\infty c_{r+pM}c_{r+qM}^*\ket{r+pM}\bra{r+qM}\\&=\bigoplus_{r=0}^{M-1}\ket{u_r}\bra{u_r}\\&=\bigoplus_{r=0}^{M-1}\lambda_r\ket{\tilde{u}_r}\bra{\tilde{u}_r}\;,
\end{split}
\end{equation}
where
\begin{equation}
\ket{u_r}=\sum_{p=0}^\infty c_{r+pM}\ket{r+pM}\;,
\end{equation}
with
\begin{equation}
\lambda_r=\sum_{p=0}^\infty |c_{r+pM}|^2\;,
\end{equation}
and
\begin{equation}
\label{eq:utilde}
\ket{\tilde{u}_r} =\frac{1}{\sqrt{\lambda_r}}\ket{u_r}\;.   
\end{equation}

From the additivity of the QFI, we arrive at the following expression.
\begin{equation}
\begin{split}
        \QFI( \mathcal E_{\delta,M}(\ket{\psi}))&=\sum_{r=0}^{M-1}\lambda_r\QFI(\ket{\tilde{u}_r})\\
        &=4\sum_{r=0}^{M-1}\lambda_r\text{Var}_{\ket{\tilde{u}_r}}(\hat{n})\;.
\end{split}
\end{equation}
For the state $\ket{\tilde{u}_r}$, we have (summing over non-zero $\lambda_r$ only)
\begin{equation}
\begin{split}
    \langle \hat{n}\rangle&=\frac{1}{\lambda_r}\sum_p|c_{r+pM}|^2(r+pM)\\
    &=r+\frac{M}{\lambda_r}\sum_p|c_{r+pM}|^2p\;,
\end{split}
\end{equation}
and
\begin{equation}
    \begin{split}
         \langle \hat{n}^2\rangle&=\frac{1}{\lambda_r}\sum_p|c_{r+pM}|^2(r+pM)^2\\
         &=r^2+\frac{1}{\lambda_r}\sum_p|c_{r+pM}|^2(2rpM+p^2M^2)\;.
    \end{split}
\end{equation}
Therefore, 
\begin{equation}
\begin{split}
       \text{Var}_{\ket{\tilde{u}_r}}(\hat{n})&=\frac{M^2}{\lambda_r}\sum_p|c_{r+pM}|^2p^2-\frac{M^2}{\lambda_r^2}\left(\sum_p|c_{r+pM}|^2p\right)^2\\
       &=\frac{M^2}{\lambda_r}\left(\sum_p|c_{r+pM}|^2p^2-\frac{\left(\sum_p|c_{r+pM}|^2p\right)^2}{\lambda_r}\right)\;.
\end{split}
\end{equation}

Therefore, the QFI can be written as
\begin{equation}
    \label{eq:QFIgen}\QFI=4M^2 \sum_{j,\lambda_j\neq0}\left(\sum_{p=0}^\infty |c_{j+pM}|^2p^2-\frac{\left(\sum_{p=0}^\infty |c_{j+pM}|^2p\right)^2}{\lambda_j}\right)\;.
\end{equation}
As the second term is negative this offers an immediate upper bound
\begin{equation}
\label{eq:QFIlb}
\begin{split}
       \QFI&\leq4M^2 \sum_{j,\lambda_j\neq0}\left(\sum_{p=0}^\infty |c_{j+pM}|^2p^2\right)\\
       &=4M^2 \sum_{j,\lambda_j\neq0}\left(\sum_{p=1}^\infty |c_{j+pM}|^2p^2\right)\\
       &\leq 4\sum_{j,\lambda_j\neq0}\left(\sum_{p=1}^\infty |c_{j+pM}|^2(j+Mp)^2\right)\\
       &=4\sum_{n\geq M} |c_{n}|^2n^2\;.
\end{split}
\end{equation}
The third line uses a $j+Mp>Mp$, and the fourth line is simply a relabelling.

A lower bound follows from considering just one term in the sum in Eq.~\eqref{eq:QFIgen}.
\begin{equation}
\label{eq:QFIlowerbound}
    \QFI\geq4M^2 \left(\sum_{p=0}^\infty |c_{pM}|^2p^2-\frac{\left(\sum_{p=0}^\infty |c_{pM}|^2p\right)^2}{\lambda_0}\right)\;.
\end{equation}

It will be useful for later in the manuscript to rewrite Eq.~\eqref{eq:QFIsimple} in a
pairwise form. Define
\begin{equation}
    a_{j,p}:=|c_{j+pM}|^2\;,
\end{equation}
such that $\lambda_j=\sum_{p\ge0}a_{j,p}$. Note that,
\begin{equation}
\begin{split}
    &\lambda_j\sum_p a_{j,p}p^2
-\left(\sum_p a_{j,p}p\right)^2 \nonumber\\
 =&
\frac12\sum_{p,q\ge0}
a_{j,p}a_{j,q}(p-q)^2
\\=&
\sum_{0\le p<q<\infty}
a_{j,p}a_{j,q}(q-p)^2.
\end{split}    
\end{equation}
Therefore, we can write the QFI as
\begin{equation}
\label{eq:pairwiseQFI}
\QFI
=
4M^2\sum_{j=0}^{M-1}\frac{1}{\lambda_j}
\sum_{0\le p<q<\infty}
a_{j,p}a_{j,q}(q-p)^2,
\end{equation}
where again the sum is only over non-zero $\lambda_j$.

\subsection{Coherent states}
\label{apen:maryC}
For a coherent state input, we can write the density matrix after the twirling channel as
\begin{equation}
\mathcal E_{0,M}(\ket{\alpha}) =  \, e^{-|\alpha|^2} \sum_{n,m=0}^{\infty} 
\frac{ \alpha^n (\alpha^*)^m}{\sqrt{n! \, m!}} 
\, \delta_{(n-m) \bmod M, 0} \, |n\rangle \langle m|.
\end{equation}
Following the above, we arrive at
\begin{equation}
\begin{split}
    \QFI^{\text{C}}=4M^2\sum_{j,\lambda_j\neq0}\bigg(&\mathrm{e}^{-|\alpha|^2}\sum_{p=0}^\infty \frac{|\alpha|^{2(j+pM)}}{(j+pM)!}p^2-\mathrm{e}^{-2|\alpha|^2}\frac{\left(\sum_{p=0}^\infty \frac{|\alpha|^{2(j+pM)}}{(j+pM)!}p\right)^2}{\lambda_j}\bigg)\;.
\end{split}
\end{equation}
with 
\begin{equation}
    \lambda_j=\mathrm{e}^{-|\alpha|^2}\sum_{p=0}^\infty\frac{|\alpha|^{2(j+pM)}}{(j+pM)!}\;.
\end{equation}

We now evaluate the upper bound in Eq.~\eqref{eq:QFIlb} for a coherent state:
\begin{equation}
\begin{split}
       \QFI^{\text{C}}&\leq4\sum_{n\geq M} |c_{n}|^2n^2\;.
\end{split}
\end{equation}
Denoting the terms in this sum by $t_i=|c_{i}|^2i^2$. We have that
\begin{equation}
\begin{split}
        \frac{t_{n+1}}{t_{n}}&=\frac{|c_{n+1}|^2(n+1)^2}{|c_{n}|^2n^2}\\
        &=|\alpha|^2\frac{n+1}{n^2}\;.
\end{split}
\end{equation}
For all $n\geq M$, the ratio $r_n=t_{n+1}/t_{n}$ is bounded by
\begin{equation}
   t_{n+1}/t_{n}\leq r_M= |\alpha|^2\frac{M+1}{M^2}\;.
\end{equation}
From the rules of geometric series (assuming $r_M<1$), we then have that
\begin{equation}
    \sum_{n\geq M}t_n\leq t_M\sum_{k=0}^{\infty}r_M^k=\frac{t_M}{1-r_M}\;.
\end{equation}
Filling everything in using
\begin{equation}
    t_M=|c_M|^2M^2=\frac{M^2\mathrm{e}^{-|\alpha|^2}|\alpha|^{2M}}{M!}
\end{equation}
we arrive at the following bound for coherent states (for $|\alpha|^2<\frac{M^2}{M+1}$)
\begin{equation}
    \label{eq:QFIcoherentUB}\QFI^{\text{C}}\leq4\frac{M^2\mathrm{e}^{-|\alpha|^2}|\alpha|^{2M}}{M!(1-|\alpha|^2\frac{M+1}{M^2})}\;.
\end{equation}
As an illustrative example, with $\alpha=1$, we see
\begin{equation}
    \QFI^{\text{C}}\leq\frac{4M^2\mathrm{e}^{-1}}{M!(1-\frac{M+1}{M^2})}\;.
\end{equation}
For large $M$, this is dominated by the $O(1/M!)$ term.

\subsubsection{How many photons are required for a non-vanishing QFI?}
For a coherent state with mean photon number $\bar{n}$, we use the upper bound in Eq.~\eqref{eq:QFIcoherentUB}. When $\bar{n}=o(M)$, we have
\begin{equation}
    \left(1-\bar{n}\frac{M+1}{M^2}\right)^{-1}=1+o(1)\;.
\end{equation}
This implies, via Stirling's approximation ($M!=\sqrt{2\pi M}(M/e)^M(1+o(1))$),
\begin{equation}
\label{eq:QFIcohexpdecay}
\begin{split}
    \QFI^{\mathrm{C}}
\le&
4\frac{M^2\mathrm{e}^{-|\alpha|^2}|\alpha|^{2M}}{M!(1-|\alpha|^2\frac{M+1}{M^2})}\\
\sim&
4\frac{M^2\mathrm{e}^{-\bar n}\bar n^{M}}{\sqrt{2\pi M}\left(\frac{M}{e}\right)^M}\\
=&
\frac{4}{\sqrt{2\pi}}M^{3/2}\mathrm{e}^{M-\bar n}\left(\frac{\bar n}{M}\right)^M\\
=&\frac{4}{\sqrt{2\pi}}
\exp\!\left[-\bar{n}+M+M\log\!\left(\frac{\bar{n}}{M}\right)+\frac32\log(M)\right].
\end{split}
\end{equation}
Therefore any sublinear scaling $\bar{n}=o(M)$ gives $\QFI^{\mathrm{C}}\to0$, and in particular $\bar{n}\sim M/\log M$ is insufficient. To keep the coherent-state QFI from vanishing, one requires at least
\begin{equation}
\bar{n}(M)=\Omega(M).
\end{equation}

\subsubsection{Constant mean photon number}
We now evaluate the classical benchmark with constant mean photon number $\bar{n}$. By convexity and monotonicity of QFI, the supremum over the classical family (Eq.~\eqref{eq:classicaldef}) is attained
by the coherent state of energy $\bar{n}$:
\begin{equation}
\sup_{\rho\in\rho_{\rm cl}(\bar{n})}
\QFI(\mathcal E_{\delta,M}(\rho))
=
\QFI^{\rm C}\left(\mathcal E_{\delta,M}\left(\ket{\sqrt{\bar{n}}}\right)\right).
\label{eq:fixedE-classical-sup}
\end{equation}
Equation~\eqref{eq:qfiubCoh} (also Eq.~\eqref{eq:QFIcoherentUB}) gives
\begin{equation}
\QFI^{\rm C}\left(\mathcal E_{\delta,M}\left(\ket{\sqrt{\bar{n}}}\right)\right)
\le
\frac{4M^2e^{-\bar{n}}\bar{n}^M}
{M!\left(1-\bar{n}(M+1)/M^2\right)}
=
\frac{4M^2e^{-\bar{n}}\bar{n}^M}{M!}[1+o(1)].
\label{eq:fixedE-coh-upper}
\end{equation}
Conversely, retaining only the $\ket 0\leftrightarrow\ket M$ pair in
the pairwise QFI representation (Eq.~\eqref{eq:pairwiseQFI}) gives
\begin{equation}
\QFI^{\rm C}\left(\mathcal E_{\delta,M}\left(\ket{\sqrt{\bar{n}}}\right)\right)
\ge
4M^2\frac{p_0p_M}{\lambda_0}.
\end{equation}
For fixed $\bar{n}$,
\begin{equation}
p_0=e^{-\bar{n}},
\qquad
p_M=\frac{e^{-\bar{n}}\bar{n}^M}{M!},
\qquad
\lambda_0=p_0[1+o(1)],
\end{equation}
and therefore
\begin{equation}
\QFI^{\rm C}\left(\mathcal E_{\delta,M}\left(\ket{\sqrt{\bar{n}}}\right)\right)
\ge
\frac{4M^2e^{-\bar{n}}\bar{n}^M}{M!}[1+o(1)].
\end{equation}
Combining the upper and lower bounds yields
\begin{equation}
\sup_{\rho\in\rho_{\rm cl}(\bar{n})}
\QFI(\mathcal E_{\delta,M}(\rho))
=
\frac{4M^2e^{-\bar{n}}\bar{n}^M}{M!}[1+o(1)] .
\label{eq:fixedE-classical-asymptotic}
\end{equation}

\subsection{Squeezed states}
\label{apen:maryS}
The Fock-basis coefficients of a displaced squeezed vacuum state, with real $\alpha$ and $r$, are
\begin{equation}
c_{s,m}(\alpha,r) = \frac{1}{\sqrt{\cosh r}} \sum_{n=0}^{\infty} 
\frac{\sqrt{(2n)!}}{2^n n!} (-\tanh r)^n \,
\langle m | D(\alpha) | 2n \rangle,
\end{equation}
where the displacement matrix elements are
\begin{equation}
\langle m | D(\alpha) | 2n \rangle =
\begin{cases} 
\sqrt{\frac{(2n)!}{m!}} \, \alpha^{m-2n} \, e^{-\alpha^2/2} \, 
L_{2n}^{(m-2n)}(\alpha^2), & m \ge 2n,\\[1em]
\sqrt{\frac{m!}{(2n)!}} \, (-\alpha)^{2n-m} \, e^{-\alpha^2/2} \, 
L_m^{(2n-m)}(\alpha^2), & m < 2n,
\end{cases}
\end{equation}
and $L_n^{(k)}(x)$ are the generalized Laguerre polynomials. The squeezed state is then written as
\begin{equation}
    \ket{s}=\sum_{n=0}^\infty c_{s,n}\ket{n}\;.
\end{equation}
Therefore the QFI can be evaluated efficiently using Eq.~\eqref{eq:QFIgen}.

Going forward, we focus on squeezed vacuum state for proving a lower bound. This is because, as we show in Appendix~\ref{sec:additional-sv-optimality-notes}, squeezed vacuum (locally and asymptotically) outperforms any displaced squeezed state for this task. Although the squeezed vacuum state is defined in Eq.~\eqref{eq:svcoeffsmain}, we repeat this here for completeness. For a squeezed vacuum state,
\begin{equation}
\ket{r}=S(r)\ket{0}
=
\frac{1}{\sqrt{\cosh r}}
\sum_{k=0}^{\infty}
\frac{\sqrt{(2k)!}}{2^k k!}
(-\tanh r)^k
\ket{2k},
\end{equation}
so that
\begin{equation}
c_{2k}^{\mathrm{SV}}
=
\frac{1}{\sqrt{\cosh r}}
\frac{\sqrt{(2k)!}}{2^k k!}
(-\tanh r)^k,
\qquad
c_{2k+1}^{\mathrm{SV}}=0.
\end{equation}
Hence
\begin{equation}
p_{2k}^{\mathrm{SV}}
:=
\left|c_{2k}^{\mathrm{SV}}\right|^2
=
\frac{1}{\cosh r}
\frac{(2k)!}{2^{2k}(k!)^2}
\tanh^{2k}r,
\qquad
p_{2k+1}^{\mathrm{SV}}=0.
\end{equation}
Writing the mean photon number as $\bar n=\sinh^2 r$, we have
\begin{equation}
\cosh r=\sqrt{1+\bar n},
\qquad
\tanh^2 r=\frac{\bar n}{1+\bar n},
\end{equation}
and therefore
\begin{equation}
\label{eq:p2ksqV}
p_{2k}^{\mathrm{SV}}
=
\frac{1}{\sqrt{1+\bar n}}
\frac{\binom{2k}{k}}{2^{2k}}
\left(\frac{\bar n}{1+\bar n}\right)^k.
\end{equation}

Keeping only the pair $j,j+M$ in each term in the inner sum in the pairwise QFI formula (Eq.~\eqref{eq:pairwiseQFI}) gives the following lower bound on the QFI
\begin{equation}
\label{eq:SqVacQFILB2}
    \QFI ^{\mathrm{SV}} (\mathcal E_{\delta,M}(\ket{r}))
    \ge
    4M^2
    \sum_{\substack{0\le j<M\\ j\,\mathrm{even}}}
    \frac{p_j^{\mathrm{SV}}p_{j+M}^{\mathrm{SV}}}{\lambda_j}
\end{equation}

Define
\begin{equation}
\label{eq:qd}
    q:=\frac{\bar n}{1+\bar n},
    \qquad
    d_k:=\frac{\binom{2k}{k}}{4^k}.
\end{equation}
Then
\begin{equation}
    p_{2k}^{\mathrm{SV}}
    =\frac{d_kq^k}{\sqrt{1+\bar n}},
    \qquad p_{2k+1}^{\mathrm{SV}}=0.
\end{equation}
The sequence $d_k$ is decreasing, so for every even $0\le j<M$,
\begin{equation}
    \frac{p^{\mathrm{SV}}_{j+(p+1)M}}
    {p^{\mathrm{SV}}_{j+pM}}
    \le q^{M/2}.
\end{equation}
Consequently (using the fact that $p^{\mathrm{SV}}_{j+pM}\leq p^{\mathrm{SV}}_{j}(q^{M/2})^p$),
\begin{equation}
    \lambda_j
    =\sum_{p\ge0}p^{\mathrm{SV}}_{j+pM}
    \le p^{\mathrm{SV}}_{j}\sum_{p\ge0}(q^{M/2})^p=\frac{p_j^{\mathrm{SV}}}{1-q^{M/2}}.
\end{equation}
Eq.~\eqref{eq:SqVacQFILB2} therefore gives
\begin{equation}
\begin{split}
    \QFI ^{\mathrm{SV}}(\mathcal E_{\delta,M}(\ket{r}))
    &\ge
    4M^2
    \sum_{\substack{0\le j<M\\ j\,\mathrm{even}}}
    \frac{p_j^{\mathrm{SV}}p_{j+M}^{\mathrm{SV}}}{\lambda_j}\\
    &\ge
    4M^2
    (1-q^{M/2})\sum_{\substack{0\le j<M\\ j\,\mathrm{even}}}
    p_{j+M}^{\mathrm{SV}}\\
    &\ge
    4M^2(1-q^{M/2})
    \Pr_{\mathrm{SV}}(M\le N<2M)\;,
\end{split}
\label{eq:SV-block-tail-lower}
\end{equation}
where $\Pr_{\mathrm{SV}}(M\le N<2M)$ is the probability that a photon number measurement on the squeezed vacuum state returns between $M$ and $2M$ photons. 

\subsubsection{How many photons are required for constant QFI?}
\label{apen:sqVachowmanyphot}

We now evaluate how many photons are required such that the above lower bound on the squeezed vacuum QFI is at least constant. Now let $K=\lfloor\bar n/2\rfloor$.  We choose $\bar n=o(M)$, such that
\begin{equation}
\label{apen:sqVachowmanyphot1}
    \Pr_{\mathrm{SV}}(M\le N<2M)
    \ge
    \sum_{k=M/2}^{M/2+K}p_{2k}^{\mathrm{SV}}.
\end{equation}
Uniformly on this interval, Stirling's approximation gives $d_k\ge c_1M^{-1/2}$ for some constant $c_1>0$, while $(1+\bar n)^{-1/2}\ge c_2\bar n^{-1/2}$.  Moreover,
\begin{equation}
    q^k\ge q^{M/2+\bar n/2}
    =q^{M/2}q^{\bar n/2},
\end{equation}
and $q^{\bar n/2}$ is bounded below by a positive constant as $\bar n\to\infty$.  Since there are $\Theta(\bar n)$ terms,
\begin{equation}
    \Pr_{\mathrm{SV}}(M\le N<2M)
    \ge
    c\sqrt{\frac{\bar n}{M}}\,q^{M/2}
\label{eq:SV-block-prob-lower}
\end{equation}
for some $c>0$ and all sufficiently large $M$. Putting everything together we find
\begin{equation}
    \QFI ^{\mathrm{SV}}(\mathcal E_{\delta,M}(\ket{r}))
    \ge
    4cM^{3/2}\sqrt{\bar{n}}(1-q^{M/2})q^{M/2}
\end{equation}

For $\bar n=M/\log M$,
\begin{equation}
\label{apen:sqVachowmanyphot2}
\begin{split}
    q^{M/2}
    &=\left(1+\frac1{\bar n}\right)^{-M/2}\\
    &=\exp\!\left[-\frac{M}{2\bar n}
       +O\!\left(\frac{M}{\bar n^2}\right)\right]\\
    &=M^{-1/2}(1+o(1)),
\end{split}
\end{equation}
and $1-q^{M/2}\to1$.  Substituting into Eqs.~\eqref{eq:SV-block-tail-lower}--\eqref{eq:SV-block-prob-lower},
\begin{equation}
    \QFI ^{\mathrm{SV}}(\mathcal E_{\delta,M}(\ket{r}))
    \ge
    c\frac{M^{3/2}}{\sqrt{\log M}}
\label{eq:SV-strong-lower}
\end{equation}
for all sufficiently large $M$.

We now prove that the choice $\bar n=M/\log(M)$ is not arbitrary, by demonstrating that when $\bar n=o(M/\log(M))$ the QFI decays super-polynomially. Denote $\ket{r}_{\bar n}$ as the squeezed vacuum state with $\bar n$ photons. The QFI upper bound (Eq.~\eqref{eq:QFIlb}) gives
\begin{equation}
\begin{split}
    \QFI^{\rm SV}(\mathcal E_{\delta,M}(\ket{r}_{\bar{n}}))
&\le
4\sum_{n\ge M}n^2p_n^{\rm SV}\\&=\frac{16 q^{M/2}}{\sqrt{1+\bar n}}\sum_{k=M/2}^\infty k^2 q^{k-M/2}d_k
\\&\le\frac{16 q^{M/2}}{\sqrt{1+\bar n}}\sum_{k=M/2}^\infty k^2 q^{k-M/2}
\\&=\frac{16 q^{M/2}}{\sqrt{1+\bar n}}\sum_{t=0}^\infty \left(t+\frac{M}{2}\right)^2 q^{t}
\\&=\frac{16 q^{M/2}}{\sqrt{1+\bar n}}\left(\frac{M^2}{4}(1+\bar n)+M\bar n(1+\bar n)+\bar n(1+\bar n)(1+2 \bar n)\right)\;.
\end{split}
\end{equation}
As $\bar n=o(M/\log(M))$ we can approximate this as
\begin{equation}
   \QFI^{\rm SV}(\mathcal E_{\delta,M}(\ket{r}_{\bar{n}}))=O\left( q^{M/2}M^2\sqrt{\bar n}\right)\;.
\end{equation}
Finally, we have that
\begin{equation}
  q^{M/2}=\left(1-\frac{1}{1+\bar n}\right)^{M/2} \leq \exp\left(-\frac{M}{2(\bar n+1)}\right)
\end{equation}
As $\bar n=o(M/\log(M))$, $M/\bar{n}$ grows strictly faster than $\log(M)$ ensuring super-polynomial decay.

Note that we can arrive at a slightly weaker bound by noting
\begin{equation}
\begin{split}
    \QFI ^{\mathrm{SV}}(\mathcal E_{\delta,M}(\ket{r}_{\bar{n}}))
    &\ge
    4M^2(1-q^{M/2})
    \Pr_{\mathrm{SV}}(M\le N<2M)\\
    &\geq4M^2(1-q^{M/2}) p_{M}^{\mathrm{SV}} \\
    &= \frac{4M^2}{\sqrt{1+\bar{n}}}(1-q^{M/2})\frac{\binom{M}{M/2}}{2^M}\left(\frac{\bar{n}}{1+\bar{n}}\right)^{M/2}\;.
\end{split}
\label{eq:SV-block-tail-lower2}
\end{equation}

{\color{black}

\subsubsection{Constant mean photon number}
For a squeezed vacuum with mean photon number $\bar n$, its photon-number
probabilities are
\begin{equation}
    p^{\rm SV}_{2k}
    =
    \frac{1}{\sqrt{1+\bar{n}}}\,d_k q^k,
    \qquad
    p^{\rm SV}_{2k+1}=0.
\label{eq:fixedE-sv-probs}
\end{equation}
We first determine the asymptotic squeezed-vacuum QFI. From
Eq.~\eqref{eq:SV-block-tail-lower},
\begin{equation}
\QFI^{\rm SV}(\mathcal E_{\delta,M}(\ket{r}_{\bar{n}}))
\ge
4M^2(1-q^{M/2})
\Pr_{\bar{n}}(M\le N<2M).
\label{eq:fixedE-block-lower}
\end{equation}
The block probability is
\begin{equation}
\Pr_{\bar{n}}(M\le N<2M)
=
\frac{q^{M/2}}{\sqrt{1+\bar{n}}}
\sum_{t=0}^{M/2-1}d_{M/2+t}q^t.
\label{eq:fixedE-block-prob}
\end{equation}
Since $\bar{n}$ is finite, $0<q<1$. Furthermore, $
\frac{d_{M/2+t}}{d_{M/2}}\longrightarrow 1$ and for every fixed $t$, while monotonicity of $d_k$ gives
$0\le
\frac{d_{M/2+t}}{d_{M/2}}
\le1$. Dominated convergence therefore yields
\begin{equation}
\sum_{t=0}^{M/2-1}d_{M/2+t}q^t
=
\frac{d_{M/2}}{1-q}[1+o(1)]
=
(1+\bar{n})d_{M/2}[1+o(1)].
\label{eq:fixedE-geometric-tail}
\end{equation}
Consequently,
\begin{equation}
\Pr_{\bar{n}}(M\le N<2M)
=
\sqrt{1+\bar{n}}\,
d_{M/2}q^{M/2}[1+o(1)].
\label{eq:fixedE-block-asymptotic}
\end{equation}
Since $q^{M/2}=o(1)$, Eq.~\eqref{eq:fixedE-block-lower} gives
\begin{equation}
\QFI^{\rm SV}(\mathcal E_{\delta,M}(\ket{r}_{\bar{n}}))
\ge
4M^2\sqrt{1+\bar{n}}\,
d_{M/2}q^{M/2}[1+o(1)].
\label{eq:fixedE-sv-lower}
\end{equation}

This lower bound is asymptotically tight. Indeed, the general tail
upper bound (Eq.~\eqref{eq:QFIlb}) gives
\begin{equation}
\QFI^{\rm SV}(\mathcal E_{\delta,M}(\ket{r}_{\bar{n}}))
\le
4\sum_{n\ge M}n^2p_n^{\rm SV}.
\end{equation}
Writing $n=M+2t$ and using Eq.~\eqref{eq:fixedE-sv-probs},
\begin{equation}
\begin{split}
\QFI^{\rm SV}(\mathcal E_{\delta,M}(\ket{r}_{\bar{n}}))
&\le
\frac{4M^2d_{M/2}q^{M/2}}{\sqrt{1+\bar{n}}}
\sum_{t=0}^{\infty}
\frac{d_{M/2+t}}{d_{M/2}}
q^t
\left(1+\frac{2t}{M}\right)^2  \\
&=
4M^2\sqrt{1+\bar{n}}\,
d_{M/2}q^{M/2}[1+o(1)],
\end{split}
\label{eq:fixedE-sv-upper}
\end{equation}
where the second line again follows by dominated convergence. Hence
\begin{equation}
\QFI^{\rm SV}(\mathcal E_{\delta,M}(\ket{r}_{\bar{n}}))
=
4M^2\sqrt{1+\bar{n}}\,
d_{M/2}
\left(\frac{\bar{n}}{1+\bar{n}}\right)^{M/2}
[1+o(1)] .
\label{eq:fixedE-sv-asymptotic}
\end{equation}
}

\subsection{Comparison of squeezed light with coherent light}
\label{apen:constantnbar}
\subsubsection{Proof of Theorem~\ref{th1}: Sub-linearly growing mean photon number} 
Theorem~\ref{th1} is proven in Eqs.~\eqref{eq:QFIcohexpdecay} and \eqref{eq:SV-strong-lower}.

\subsubsection{Proof of Theorem~\ref{th2}: Constant mean photon number} 
Combining Eqs.~\eqref{eq:fixedE-classical-asymptotic} and \eqref{eq:fixedE-sv-asymptotic} we arrive at
\begin{equation}
\frac{\QFI^{\rm SV}(\ket{r}_{\bar{n}})}
{\sup_{\rho\in\rho_{\rm cl}(\bar{n})}
 \QFI(\rho)}
=
\Omega_{\bar{n}}\!\left[
\left(
    \frac{M}{e\sqrt{\bar{n}(1+\bar{n})}}
\right)^M
\right].
\end{equation}
For $\bar{n}=1$, this reduces to
\begin{equation}
\frac{\QFI^{\rm SV}}
{\sup_{\rho\in\rho_{\rm cl}(1)}
 \QFI(\rho)}
=
\Omega\!\left[
\left(\frac{M}{\sqrt{2}e}\right)^M
\right],
\end{equation}
recovering Theorem~\ref{th2} in its previous form.

\section{Finite-sample separation in phase estimation}
\label{apen:samplecomplexity}

The QFI separation established in the previous appendix can be converted into a nonasymptotic separation in the number of independent channel uses. Because the twirled output is invariant under $\delta\mapsto\delta+2\pi/M$, a fixed absolute target error would eventually exceed the entire identifiable phase interval and would therefore define a trivial estimation problem. We instead require an error of order $1/M$ on a local interval of the same order. The classical lower bound below allows correlated classical probes from the family of Eq.~\eqref{eq:classicaldef} with collective output measurements, and also applies to adaptive classical measure-and-prepare protocols. The quantum upper bound does not require either asymptotic efficiency or an unbiasedness assumption.

The channel is that described in Eq.~\eqref{eq:sample-channel}. Fix constants
\begin{equation}
    0<a<\frac\pi2,
    \qquad
    0<b<\frac a2,
    \qquad
    0<\zeta<\frac12,
\label{eq:sample-constants}
\end{equation}
and define $\Theta_M=[-a/M,a/M]$. A protocol using $\nu$ channel calls succeeds if its estimator obeys
\begin{equation}
    \inf_{\delta\in\Theta_M}
    \Pr_\delta\!\left(
        |\widehat\delta-\delta|\le\frac bM
    \right)
    \ge 1-\zeta.
\label{eq:sample-success}
\end{equation}
For $\bar{n}>0$, the classical optical family of states is defined in Eq.~\eqref{eq:classicaldef}. More generally, given $\nu$ access to a channel we consider classically correlated states
\begin{equation}
    \rho_{\mathrm{cl}}^{(\nu)}
    =\int\bigotimes_{k=1}^{\nu}
    \Lambda_k^{(x)}(\ket{\alpha})\,
   \mu(dx) \;,
\label{eq:sample-correlated-classical}
\end{equation}
where each $\Lambda_k^{(x)}$ is parameter-independent and phase-covariant such that $\Lambda_k^{(x)}(\ket{\alpha})$ belongs to the classical family in Eq.~\eqref{eq:classicaldef}, and $x$ is a classical variable denoting classical correlations between rounds. Note that the phase twirl is independent across channel uses. Denote the minimum sample complexity over these probes and arbitrary collective measurements by $\nu_{\mathrm{cl}}^\star=\nu_{\mathrm{cl}}^\star(M,\bar{n};a,b,\zeta)$. The same notation also includes adaptive classical measure-and-prepare protocols with the per-round constraint~\eqref{eq:classicaldef}. Let $\nu_{\mathrm{SV}}^\star=\nu_{\mathrm{SV}}^\star(M,\bar{n};a,b,\zeta)$ denote the corresponding sample complexity when each input is a squeezed vacuum with mean photon number $\bar{n}$.

\begin{theorem}
\label{thm:sample-finite}
Let $M\ge2$ be even and suppose $0<\bar{n}<M^2/(M+1)$. Define
\begin{equation}
    \kappa_{M,\bar{n}}:=\frac{\bar{n}(M+1)}{M^2},
    \qquad
    j_{M,\bar{n}}
    :=\frac{4M^2e^{-\bar{n}}(\bar{n})^M}
    {M!\left(1-\kappa_{M,\bar{n}}\right)}.
\label{eq:sample-kappa-j}
\end{equation}
For the squeezed vacuum, recall the definition of $p_{2k}^{\text{SV}}(\bar{n})$ from Eq.~\eqref{eq:p2ksqV} with $p_{2k+1}^{\text{SV}}(\bar{n})=0$, and define the sector-pair visibility
\begin{equation}
    B_{M,\bar{n}}
    :=4\sum_{\substack{0\le j<M\\j\ \mathrm{even}}}
    \frac{p_j^{\text{SV}}(\bar{n})p^{\text{SV}}_{j+M}(\bar{n})}{p^{\text{SV}}_j(\bar{n})+p^{\text{SV}}_{j+M}(\bar{n})}.
\label{eq:sample-sector-visibility}
\end{equation}
Then
\begin{equation}
    \nu_{\mathrm{cl}}^\star
    \ge
    \frac{(1-2\zeta)^2M^2}{4b^2j_{M,\bar{n}}}
    =\frac{(1-2\zeta)^2}{16b^2}
    \left(1-\kappa_{M,\bar{n}}\right)e^{\bar{n}}\frac{M!}{\bar{n}^M}
\label{eq:sample-classical-lower}
\end{equation}
and
\begin{equation} \nu_{\mathrm{SV}}^\star
    \le
    \left\lceil
        \frac{2+\frac43b\cos a}{b^2\cos^2 a}
        \frac1{B_{M,\bar{n}}}
        \log\!\left(\frac2\zeta\right)
    \right\rceil\;.
\label{eq:sample-sv-upper}
\end{equation}
\end{theorem}

\begin{proof}
We first establish the classical lower bound. We know from Eq.~\eqref{eq:QFIcoherentUB} that for a coherent state with $x=|\alpha|^2$, whenever $0\le x\le \bar{n}<M^2/(M+1)$,
\begin{equation}
    \QFI \!\left(
        \mathcal E_{\delta,M}(\ket\alpha\!\bra\alpha)
    \right)
    \le
    \frac{4M^2e^{-x}x^M}
    {M!\left[1-x(M+1)/M^2\right]}
    \le j_{M,\bar{n}}.
\label{eq:sample-coherent-qfi-upper}
\end{equation}
The final inequality follows because $e^{-x}x^M$ and $[1-x(M+1)/M^2]^{-1}$ are both increasing on $[0,\bar{n}]$. Convexity, monotonicity and additivity of the QFI therefore imply, for every probe~\eqref{eq:sample-correlated-classical},
\begin{equation}
\begin{split}
    \QFI \!\left[
        \mathcal E_{\delta,M}^{\otimes\nu}
        (\rho_{\mathrm{cl}}^{(\nu)})
    \right]
    &\le\QFI[\mathcal E_{\delta,M}^{\otimes\nu}(\ket{\alpha}^{\otimes\nu})]
   \\ &\le\nu j_{M,\bar{n}}.
\end{split}
\label{eq:sample-multimode-qfi}
\end{equation}
  
Choose the two parameter values $\delta_\pm=\pm2b/M\in\Theta_M$. Any estimator satisfying Eq.~\eqref{eq:sample-success} distinguishes these hypotheses by the sign of $\widehat\delta$, with error probability at most $\zeta$. Writing $\rho_\pm^{(\nu)}$ for the respective output states, Helstrom's bound gives
\begin{equation}
    D\!\left(\rho_-^{(\nu)},\rho_+^{(\nu)}\right)
    :=\frac12\left\|\rho_-^{(\nu)}-\rho_+^{(\nu)}\right\|_1
    \ge1-2\zeta.
\label{eq:sample-trace-distance-lower}
\end{equation}
On the other hand, trace distance is bounded by the Bures angle, which is in turn bounded by the QFI length of the parameter path. Equation~\eqref{eq:sample-multimode-qfi} therefore gives
\begin{equation}
\begin{split}
    D\!\left(\rho_-^{(\nu)},\rho_+^{(\nu)}\right)
    &\le\frac12\int_{-2b/M}^{2b/M}
    \sqrt{
        \QFI 
        \!\left(\rho_\delta^{(\nu)}\right)
    }\,d\delta\\
    &\le\frac{2b}{M}\sqrt{\nu j_{M,\bar{n}}}.
\end{split}
\label{eq:sample-trace-distance-upper}
\end{equation}
Combining Eqs.~\eqref{eq:sample-trace-distance-lower} and~\eqref{eq:sample-trace-distance-upper} proves Eq.~\eqref{eq:sample-classical-lower}.

The same conclusion holds for adaptive classical measure-and-prepare strategies. If $H_{k-1}$ denotes the classical history before round $k$, the conditional Fisher information satisfies $\mathcal J_k(\delta\mid H_{k-1})\le j_{M,\bar{n}}$. The conditional scores have zero conditional mean, so their cross terms vanish and the complete measurement record has Fisher information
\begin{equation}
    \mathcal J_{\mathrm{tr}}(\delta)
    =\sum_{k=1}^{\nu}
    \mathbb E_\delta\!\left[
        \mathcal J_k(\delta\mid H_{k-1})
    \right]
    \le\nu j_{M,\bar{n}}.
\label{eq:sample-adaptive-fisher}
\end{equation}
For the corresponding distributions,
\begin{equation}
\begin{split}
    \operatorname{TV}(P_{\delta_-},P_{\delta_+})
    &\le\frac12\int_{\delta_-}^{\delta_+}
    \int|\partial_\delta p_\delta(y)|\,dy\,d\delta\\
    &\le\frac12\int_{\delta_-}^{\delta_+}
    \sqrt{\mathcal J_{\mathrm{tr}}(\delta)}\,d\delta
    \le\frac{2b}{M}\sqrt{\nu j_{M,\bar{n}}}.
\end{split}
\label{eq:sample-adaptive-tv}
\end{equation}
Again requiring $\operatorname{TV}(P_{\delta_-},P_{\delta_+})\ge1-2\zeta$, establishes the same lower bound. 

For the constructive quantum upper bound, let $c_j=\langle j|S(r)\ket0$ with $\sinh^2r=\bar{n}$. For each even $0\le j<M$, choose $\chi_j$ so that $c_{j+M}/c_j=e^{\mathrm i\chi_j}|c_{j+M}/c_j|$, and define
\begin{equation}
    \ket{j,\pm}
    :=\frac{\ket j\pm\mathrm i e^{\mathrm i\chi_j}\ket{j+M}}{\sqrt2},
    \qquad
    \Pi_{j,\pm}:=\ket{j,\pm}\!\bra{j,\pm},
    \qquad
    \Pi_\perp:=I-
    \sum_{\substack{0\le j<M\\j\ \mathrm{even}}}
    (\Pi_{j,+}+\Pi_{j,-}).
\label{eq:sample-sector-povm}
\end{equation}
The projectors are mutually orthogonal. Noting that the density matrix in the $\{\ket{j},\ket{j+M}\}$ subspace is $\rho_{j,j+M}=p_j\ket{j}\bra{j}+p_{j+M}\ket{j+M}\bra{j+M}+c_jc_{j+M}^*e^{-iM\delta}\ket{j}\bra{j+M}+c_{j+M}c_j^*e^{iM\delta}\ket{j+M}\bra{j}$, the probabilities are
\begin{equation}
    P_\delta(j,\pm)
    =\frac{s_j}{2}
    \pm\sqrt{p_jp_{j+M}}\sin(M\delta),
    \qquad
    s_j:=p_j+p_{j+M},
\label{eq:sample-sector-probabilities}
\end{equation}
where $p_j=p_j(\bar{n})$, and $P_\delta(\perp)$ is independent of $\delta$. Introduce the bounded weights
\begin{equation}
    g_j:=\frac{2\sqrt{p_jp_{j+M}}}{p_j+p_{j+M}}
    \in[0,1]
\label{eq:sample-sector-weight}
\end{equation}
and the single-shot statistic
\begin{equation}
    Z:=
    \begin{cases}
        g_j,&\text{for outcome }(j,+),\\
        -g_j,&\text{for outcome }(j,-),\\
        0,&\text{for outcome }\perp.
    \end{cases}
\label{eq:sample-sector-statistic}
\end{equation}
Equations~\eqref{eq:sample-sector-visibility}--\eqref{eq:sample-sector-weight} give the exact moments
\begin{equation}
    \mathbb E_\delta[Z]
    =B_{M,\bar{n}}\sin(M\delta),
    \qquad
    \mathbb E_\delta[Z^2]
    =B_{M,\bar{n}},
    \qquad
    \operatorname{Var}_\delta(Z)\le B_{M,\bar{n}}.
\label{eq:sample-sector-moments}
\end{equation}
For $\overline Z_\nu=\nu^{-1}\sum_{k=1}^{\nu}Z_k$, consider the clipped estimator
\begin{equation}
    \widehat\delta_\nu
    :=\frac1M\arcsin\!\left[
        \operatorname{clip}\!\left(
            \frac{\overline Z_\nu}{B_{M,\bar{n}}},
            -\sin a,\sin a
        \right)
    \right],
\label{eq:sample-explicit-estimator}
\end{equation}
where $\operatorname{clip}(x,l,u):=\min\{u,\max\{l,x\}\}$. Projection onto an interval is nonexpansive, meaning the distance between outputs is smaller than the distance between their inputs, i.e.
\begin{equation}
    \left|\operatorname{clip}\!\left(
            \frac{\overline Z_\nu}{B_{M,\bar{n}}},-\sin a,\sin a
        \right)-\sin(M\delta)\right|\leq\left| \frac{\overline Z_\nu}{B_{M,\bar{n}}}
        -\sin(M\delta)\right|
\end{equation}
 The function $\arcsin$ is $1/\cos a$-Lipschitz on $[-\sin a,\sin a]$, i.e.
 \begin{equation}
     |\arcsin(x)-\arcsin(y)|\leq\frac{1}{\cos(a)}|x-y|\;.
 \end{equation}
Hence, uniformly over $\delta\in\Theta_M$,
\begin{equation}
\begin{split}
    |\widehat\delta_\nu-\delta|&\le\left|\frac{1}{M}\arcsin\!\left[
        \operatorname{clip}\!\left(
            \frac{\overline Z_\nu}{B_{M,\bar{n}}},
            -\sin a,\sin a
        \right)
    \right]-\frac{1}{M}\arcsin\left( \frac{\mathbb{E}_\delta[Z]}
{B_{M,\bar{n}}}\right)\right|\\&
\le\frac{1}{M\cos(a)}\left|\left[
        \operatorname{clip}\!\left(
            \frac{\overline Z_\nu}{B_{M,\bar{n}}},
            -\sin a,\sin a
        \right)
    \right]-\left( \frac{\mathbb{E}_\delta[Z]}
{B_{M,\bar{n}}}\right)\right|\\&
\le\frac{1}{M\cos(a)}\left|\left(
            \frac{\overline Z_\nu}{B_{M,\bar{n}}},\right)-\left( \frac{\mathbb{E}_\delta[Z]}
{B_{M,\bar{n}}}\right)\right|\\&
    \le
    \frac{ |\overline Z_\nu-\mathbb E_\delta[Z]| }{MB_{M,\bar{n}}\cos a}.
\end{split}
\label{eq:sample-estimator-error}
\end{equation}
Because $|\overline Z_\nu-\mathbb E_\delta[Z]|\le2$, Bernstein's inequality yields 
\begin{equation}
    \Pr_\delta\!\left(
        |\overline Z_\nu-\mathbb E_\delta[Z]|\ge t
    \right)
    \le2\exp\!\left[
        -\frac{\nu t^2}{2B_{M,\bar{n}}+\frac43t}
    \right].
\label{eq:sample-bernstein}
\end{equation}
Setting $t=bB_{M,\bar{n}}\cos a$ and using Eq.~\eqref{eq:sample-estimator-error} gives
\begin{equation}
    \sup_{\delta\in\Theta_M}
    \Pr_\delta\!\left(
        |\widehat\delta_\nu-\delta|>\frac bM
    \right)
    \le
    2\exp\!\left[
        -\frac{
            \nu b^2\cos^2 a\,B_{M,\bar{n}}
        }{
            2+\frac43b\cos a
        }
    \right].
\label{eq:sample-uniform-confidence}
\end{equation}
Requiring the right-hand side to be at most $\zeta$ proves Eq.~\eqref{eq:sample-sv-upper}.
\end{proof}

\begin{corollary}
\label{cor:sample-poly-superexp}
Let $M\to\infty$ through even integers and set $\bar{n}=M/\log M$. For fixed $a,b,\zeta$ satisfying Eq.~\eqref{eq:sample-constants},
\begin{equation}
\nu_{\mathrm{SV}}^\star
    =O_{a,b}\!\left(
        \sqrt{M\log M}\,
        \log\!\frac2\zeta
    \right)
\label{eq:sample-poly-quantum}
\end{equation}
whereas
\begin{equation}
\begin{split}
    \nu_{\mathrm{cl}}^\star
    &\ge
    \frac{(1-2\zeta)^2}{16b^2}
    \left(1-\frac{M+1}{M\log M}\right)
    e^{M/\log M}M!
    \left(\frac{\log M}{M}\right)^M,\\
    \log\nu_{\mathrm{cl}}^\star
    &\ge
    M\log\log M-M+\frac{M}{\log M}
    +\frac12\log M+O_{b,\zeta}(1).
\end{split}
\label{eq:sample-superexp-classical}
\end{equation}
Consequently,
\begin{equation}
\frac{\nu_{\mathrm{cl}}^\star}
    {\nu_{\mathrm{SV}}^\star}
    \ge
    \exp\!\left[M\log\log M-O(M)\right]
\label{eq:sample-superexp-ratio}
\end{equation}
for every fixed confidence level $1-\zeta$.
\end{corollary}

\begin{proof}
Write $\bar{n}=M/\log M$, and recall from Eq.~\eqref{eq:qd} that $q=\bar{n}/(1+\bar{n})$ and $d_k:=\binom{2k}{k}/4^k$, so that $p^\text{SV}_{2k}=p_0^\text{SV}d_kq^k$. Since $d_k$ decreases with $k$,
\begin{equation}
    0\le\frac{p^\text{SV}_{j+M}}{p^\text{SV}_j}
    \le q^{M/2}
    =M^{-1/2}[1+o(1)]
\label{eq:sample-sector-ratio-asymptotic}
\end{equation}
uniformly over even $0\le j<M$. Therefore
\begin{equation}
    B_{M,\bar{n}}
    =4\Pr_{\bar{n}}(M\le N<2M)
    \left[1+O(M^{-1/2})\right],
\label{eq:sample-visibility-tail}
\end{equation}
where $N$ is the squeezed-vacuum photon number. Using Appendix~\ref{apen:sqVachowmanyphot}, specifically Eqs.~\eqref{apen:sqVachowmanyphot1} to \eqref{apen:sqVachowmanyphot2}, we arrive at
\begin{equation}
    B_{M,\bar{n}}
    \geq c
    \frac1{\sqrt{M\log M}}[1+o(1)].
\label{eq:sample-visibility-asymptotic}
\end{equation}
for some constant $c$. Substitution into Eq.~\eqref{eq:sample-sv-upper} proves Eq.~\eqref{eq:sample-poly-quantum}. 

The classical lower bound (Eq.~\eqref{eq:sample-superexp-classical}) follows directly from Eq.~\eqref{eq:sample-classical-lower}, since $\kappa_{M,\bar{n}}=(M+1)/(M\log M)$. Finally, Stirling's formula gives
\begin{equation}
    M!\left(\frac{\log M}{M}\right)^M
    =\sqrt{2\pi M}\,
    \exp\!\left[M\log\log M-M\right]
    [1+o(1)],
\label{eq:sample-stirling}
\end{equation}
which establishes the second line of Eq.~\eqref{eq:sample-superexp-classical} and the ratio~\eqref{eq:sample-superexp-ratio}.
\end{proof}

\begin{corollary}
\label{cor:sample-fixed-energy}
Fix $\bar{n}>0$ independently of $M$. For fixed $a,b,\zeta$ satisfying Eq.~\eqref{eq:sample-constants},
\begin{equation}
    \nu_{\mathrm{SV}}^\star(M,\bar{n};a,b,\zeta)
    =O_{a,b,\bar{n},\zeta}\!\left[
        \sqrt M
        \left(\frac{1+\bar{n}}{\bar{n}}\right)^{M/2}
    \right],
    \qquad
    \nu_{\mathrm{cl}}^\star
    =\Omega_{b,\bar{n},\zeta}\!\left[
        \sqrt M
        \left(\frac{M}{e\bar{n}}\right)^M
    \right].
\label{eq:sample-fixed-complexities}
\end{equation}
In particular,
\begin{equation}
    \frac{\nu_{\mathrm{cl}}^\star}{\nu_{\mathrm{SV}}^\star}
    =\Omega_{a,b,\bar{n},\zeta}\!\left[
        \left(
            \frac{M}{e\sqrt{\bar{n}(1+\bar{n})}}
        \right)^M
    \right],
\label{eq:sample-fixed-ratio}
\end{equation}
which reduces to $\Omega[(M/(\sqrt2e))^M]$ for $\bar{n}=1$.
\end{corollary}

\begin{proof}
The $j=0$ summand in Eq.~\eqref{eq:sample-sector-visibility} gives
\begin{equation}
    B_{M,\bar{n}}\ge\frac{4p_0(\bar{n})p_M(\bar{n})}{p_0(\bar{n})+p_M(\bar{n})},
    \qquad
    p_M(\bar{n})
    =\Theta_{\bar{n}}\!\left[
        M^{-1/2}
        \left(\frac{\bar{n}}{1+\bar{n}}\right)^{M/2}
    \right],
\label{eq:sample-fixed-visibility}
\end{equation}
while $p_0(\bar{n})=\Theta_{\bar{n}}(1)$. The squeezed-vacuum bound in Eq.~\eqref{eq:sample-fixed-complexities} follows from Eq.~\eqref{eq:sample-sv-upper}. Since $\kappa_{M,\bar{n}}=O_{\bar{n}}(M^{-1})$, Eq.~\eqref{eq:sample-classical-lower} and Stirling's formula give the classical bound in Eq.~\eqref{eq:sample-fixed-complexities}. Dividing the two bounds proves Eq.~\eqref{eq:sample-fixed-ratio}.
\end{proof}

\section{Analytic optimality of the squeezed vacuum}
\label{sec:additional-sv-optimality-notes}

In this appendix we present two analytic results that show the squeezed-vacuum states are optimal in certain regimes.  The first establishes local optimality at fixed mean photon number for every even $M$.  The second establishes asymptotic global optimality (in terms of the fraction of photons assigned to squeezing) in the regime $\bar n=M/\log M$.

We consider the family of pure single-mode Gaussian probes
\begin{equation}
    \ket{\psi_x}=D(\sqrt{x})S(r_x)\ket{0},
    \qquad
    x+\sinh^2 r_x=\bar n,
\end{equation}
where $x=\alpha^2$ is the mean photon number assigned to displacement and the remainder is assigned to squeezing.  Thus $x=0$ is the squeezed vacuum. Throughout this subsection $M=2^m$ is even, and the displacement and squeezing parameters are taken to be real.

{\color{black}
\begin{theorem}
\label{thm:sv-local-opt}
Fix an even $M=2L$ and a mean photon number $\bar n>0$. Let
$\QFI (x)$ denote the QFI, after the uniform
$\mathbb Z_M$ phase twirl, of
\begin{equation}
    D(\sqrt{x})S(r_x)\ket{0},
    \qquad
    \sinh^2r_x=\bar n-x,
    \qquad
    0\le x\le \bar n.
\end{equation}
Then
\begin{equation}
    \left.
    \frac{\mathrm d}{\mathrm dx}
    \QFI (x)
    \right|_{x=0}<0.
\end{equation}
Consequently, $x=0$ is a strict one-sided local maximizer.
\end{theorem}

\begin{proof}
Set
\begin{equation}
    q:=\frac{\bar n}{1+\bar n},
    \qquad
    M=2L.
\end{equation}
Since squeezed vacuum has support only on even photon numbers,
write $N=2K$. Define
\begin{equation}
    p_a(k;q):=(1-q)^a\frac{(a)_k}{k!}q^k,
\end{equation}
where $(a)_0:=1$ and
$(a)_k:=a(a+1)\cdots(a+k-1)$ for $k\ge1$. For squeezed vacuum,
$K\sim p_{1/2}(\cdot;q)$.

For any $a>0$, define
\begin{equation}
    W_{a,L}(q)
    :=
    \mathbb E_{p_a}
    \left[
        \operatorname{Var}(K\mid K\bmod L)
    \right].
\end{equation}
Since $N=2K$ and $M=2L$, Eq.~\eqref{eq:QGIgeneralVar} gives
\begin{equation}
    \QFI (0)
    =
    16W_{1/2,L}(q).
\label{eq:sv-W}
\end{equation}

To determine the first-order change away from squeezed vacuum,
put
\begin{equation}
    \bar{n}_{r_x}:=\bar n-x,
    \qquad
    q_x:=\frac{\bar{n}_{r_x}}{1+\bar{n}_{r_x}},
    \qquad
    t_x:=\sqrt{q_x}.
\end{equation}
The photon-number probability-generating function (PGF) of
$D(\sqrt{x})S(r_x)\ket{0}$ is
\begin{equation}
    G_x(z)
    =
    \sqrt{\frac{1-t_x^2}{1-t_x^2z^2}}
    \exp\left[
        x(1+t_x)\frac{z-1}{1+t_xz}
    \right].
\label{eq:DS-pgf-local}
\end{equation}
At $x=0$, let $t=\sqrt q$ and
\begin{equation}
    \kappa:=\frac{1+t}{1-t}>0.
\end{equation}
Using
\begin{equation}
    \left.\frac{\mathrm dq_x}{\mathrm dx}\right|_{x=0}
    =-(1-q)^2,
\end{equation}
direct differentiation of Eq.~\eqref{eq:DS-pgf-local} gives
\begin{equation}
\begin{split}
    G_x(z)
    ={}&G_0(z)
    +\kappa x
    \left[
        z\left(\frac{1-q}{1-qz^2}\right)^{3/2}
        -G_0(z)
        -(1-q^2)\partial_qG_0(z)
    \right]
    +O(x^2),
\end{split}
\label{eq:DS-pgf-first-order}
\end{equation}
where
\begin{equation}
    G_0(z)
    =
    \left(\frac{1-q}{1-qz^2}\right)^{1/2}.
\end{equation}
The first term in square brackets is the PGF of the odd-photon
distribution $N=2K+1$ with
$K\sim p_{3/2}(\cdot;q)$. Since $M$ is even, the even and odd
photon-number sectors occupy disjoint residue classes modulo $M$.
Hence the QFI for these different sectors can be considered separately, and Eq.~\eqref{eq:DS-pgf-first-order} gives
\begin{equation}
\begin{split}
    \frac{\QFI (x)}{16}
    ={}&
    W_{1/2,L}(q)
    +\kappa x
    \left[
        W_{3/2,L}(q)
        -W_{1/2,L}(q)
        -(1-q^2)W'_{1/2,L}(q)
    \right]
    +O(x^2).
\end{split}
\label{eq:qfi-local-expansion}
\end{equation}
Differentiating and evaluating at $x=0$ gives
\begin{equation}
    \frac{1}{16}\frac{\partial\QFI (x)}{\partial x}\bigg|_{x=0}
    =\kappa 
    \left[
        W_{3/2,L}(q)
        -W_{1/2,L}(q)
        -(1-q^2)W'_{1/2,L}(q)
    \right].
\end{equation}

It remains to prove that this derivative is negative. For a probability
distribution $p=(p_k)_{k\ge0}$ define
\begin{equation}
    W_L[p]
    :=
    \min_{\mu_0,\ldots,\mu_{L-1}}
    \sum_{k=0}^{\infty}
    p_k\left(k-\mu_{k\bmod L}\right)^2.
\label{eq:W-concave-form}
\end{equation}
This equals
$\mathbb E[\operatorname{Var}(K\mid K\bmod L)]$.
For fixed $\{\mu_j\}$ the right-hand side is linear in $p$,
so $W_L[p]$ is concave in $p$.

Moreover,
\begin{equation}
    \partial_qp_a(k;q)
    =
    \frac{a}{q(1-q)}
    \left[
        p_{a+1}(k;q)-p_a(k;q)
    \right].
\end{equation}
Applying the supporting-line inequality for the concave functional
$W_L$ at $a=1/2$ therefore gives
\begin{equation}
    W_{3/2,L}(q)-W_{1/2,L}(q)
    \le
    2q(1-q)W'_{1/2,L}(q).
\label{eq:W-half-bound}
\end{equation}

Finally, we show that $W'_{a,L}(q)>0$. The PGF of
$K\sim p_a(\cdot;q)$ is
\begin{equation}
    F_a(z;q)
    =
    \left(\frac{1-q}{1-qz}\right)^a.
\end{equation}
For $h>0$,
\begin{equation}
    \frac{F_a(z;q+h)}{F_a(z;q)}
    =
    \exp\left[
        \sum_{k=1}^{\infty}
        \lambda_k(h)(z^k-1)
    \right],
\end{equation}
where
\begin{equation}
    \lambda_k(h)
    :=
    \frac{a\big[(q+h)^k-q^k\big]}{k}.
\end{equation}
Thus
\begin{equation}
    K(q+h)\overset{d}{=}K(q)+Y_h,
\end{equation}
where $Y_h$ is an independent compound-Poisson random variable.
For independent integer-valued $X,Y$,
\begin{equation}
    W_L(X+Y)\ge W_L(X)+W_L(Y),
\label{eq:W-superadditive}
\end{equation}
because conditioning on $(X\bmod L,Y\bmod L)$ gives at least as
much information as conditioning only on $(X+Y)\bmod L$.

Let
\begin{equation}
    \Lambda_h:=\sum_{k\ge1}\lambda_k(h)=O(h).
\end{equation}
The compound-Poisson distribution satisfies
\begin{equation}
    \Pr(Y_h=0)=e^{-\Lambda_h},
    \qquad
    \Pr(Y_h=L)\ge e^{-\Lambda_h}\lambda_L(h).
\end{equation}
Since $0$ and $L$ belong to the same residue class modulo $L$,
the pairwise variance formula gives
\begin{equation}
    W_L(Y_h)
    \ge
    L^2e^{-2\Lambda_h}\lambda_L(h).
\end{equation}
Consequently,
\begin{equation}
\begin{split}
    W'_{a,L}(q)
    &\ge
    L^2
    \lim_{h\downarrow0}
    \frac{\lambda_L(h)}{h}\\
    &=
    aL^2q^{L-1}>0.
\end{split}
\label{eq:W-positive-derivative}
\end{equation}

Combining Eqs.~\eqref{eq:qfi-local-expansion},
\eqref{eq:W-half-bound}, and
\eqref{eq:W-positive-derivative},
\begin{equation}
\begin{split}
    \frac{1}{16\kappa}
    \left.
    \frac{\mathrm d\QFI }
         {\mathrm dx}
    \right|_{x=0}
    &\le
    \left[
        2q(1-q)-(1-q^2)
    \right]W'_{1/2,L}(q)\\
    &=
    -(1-q)^2W'_{1/2,L}(q)
    <0.
\end{split}
\end{equation}
Since $\kappa>0$, the result follows.
\end{proof}

}

Above we proved that squeezed vacuum is locally optimal. Next we show that any globally optimal
displaced-squeezed probe must assign an asymptotically vanishing
fraction of its photon budget to displacement.

\begin{theorem}
\label{thm:sv-asymptotic-opt}
Let $M=2^m\to\infty$ and
\begin{equation}
    \bar n_M=\frac{M}{\log M}.
\end{equation}
Consider any sequence of real displaced-squeezed probes
\begin{equation}
    D(\alpha_M)S(r_M)\ket{0},
    \qquad
    \alpha_M^2+\sinh^2r_M=\bar n_M,
\end{equation}
and let $\QFI (\alpha_M,r_M)$ denote the QFI
after the uniform $\mathbb Z_M$ phase twirl. Let
$\QFI ^{\mathrm{SV}}$ denote the QFI of the
squeezed vacuum with the same mean photon number.
 
For $0<\eta<1$,
\begin{equation}
    \frac{\QFI (\alpha_M,r_M)}
         {\QFI ^{\mathrm{SV}}}
    \le
    M^{-\frac{\eta}{2(1-\eta)}+o(1)}.
\label{eq:asymptotic-ratio-rate}
\end{equation}
Consequently, any sequence of global maximizers
$(\alpha_M^*,r_M^*)$ within this family satisfies
\begin{equation}
    \frac{(\alpha_M^*)^2}{\bar n_M}\longrightarrow0.
\label{eq:optimal-displacement-fraction-zero}
\end{equation}
\end{theorem}

\begin{proof}
Write
\begin{equation}
    x:=\alpha^2,
    \qquad
    \bar{n}_{r}:=\sinh^2r,
    \qquad
    x+\bar{n}_{r}=\bar n.
\end{equation}
From Eq.~\eqref{eq:SV-strong-lower}, which evaluates the QFI when $\bar n=M/\log M$,
\begin{equation}
    \QFI ^{\mathrm{SV}}
    \ge
    c_0\frac{M^{3/2}}{\sqrt{\log M}}
\label{eq:SV-lower-theorem4}
\end{equation}
for some constant $c_0>0$ and all sufficiently large $M$.

We next upper-bound the QFI of an arbitrary displaced-squeezed
state. Set
\begin{equation}
    t:=\tanh r=\sqrt{\frac{\bar{n}_{r}}{1+\bar{n}_{r}}}.
\end{equation}
Its photon-number PGF is
\begin{equation}
    G(z)
    =
    \sqrt{\frac{1-t^2}{1-t^2z^2}}
    \exp\left[
        x(1+t)\frac{z-1}{1+tz}
    \right]\;,
\label{eq:DS-pgf-asymptotic}
\end{equation}
which satisfies $G(s)=\sum_np_ns^n$. For $1<s<1/t$ with $M\log s\ge2$, the function
$n^2s^{-n}$ is decreasing for $n\ge M$. The bound in Eq.~\eqref{eq:QFIlb} gives
\begin{equation}
\label{eq:Chernoff-QFI}
\begin{split}
     \QFI 
    \le&4\sum_{n\geq M}n^2p_n\\
    \leq&4M^2s^{-M}\sum_{n\geq M}p_ns^n\\
    \leq&4M^2s^{-M}G(s)
\end{split}
\end{equation}
Whenever $\bar{n}_{r}=o(M)$, choose
\begin{equation}
    s=\frac{1-1/M}{t}.
\end{equation}
For all sufficiently large $M$, this satisfies the conditions
above, and
\begin{equation}
    s^{-M}\le Ct^M,
    \qquad
    1-t^2s^2\ge\frac1M.
\end{equation}
Moreover, since $s<1/t$,
\begin{equation}
    (1+t)\frac{s-1}{1+ts}
    \le
    \frac{1-t^2}{t}
    =
    \frac{1}{\sqrt{\bar{n}_{r}(1+\bar{n}_{r})}}
    \le
    \frac{1}{\bar{n}_{r}}.
\end{equation}
Therefore, for $\bar{n}_{r}>0$,
\begin{equation}
    \QFI (\alpha,r)
    \le
    CM^2\sqrt{\frac{M}{\bar{n}_{r}}}\,
    t^M\exp\left(\frac{x}{\bar{n}_{r}}\right).
\label{eq:DS-master-upper}
\end{equation}

Suppose first that
\begin{equation}
    \frac{x}{\bar n}\to\eta,
    \qquad 0<\eta<1,
\end{equation}
and write $\beta:=1-\eta$. Then
\begin{equation}
    \bar{n}_{r}
    =
    \beta\frac{M}{\log M}(1+o(1)).
\end{equation}
Consequently,
\begin{equation}
\begin{split}
    t^M
    &=
    \left(\frac{\bar{n}_{r}}{1+\bar{n}_{r}}\right)^{M/2}\\
    &=
    \exp\left[
        -\frac{M}{2\bar{n}_{r}}
        +O\left(\frac{M}{\bar{n}_{r}^2}\right)
    \right]\\
    &=
    M^{-1/(2\beta)+o(1)}.
\end{split}
\end{equation}
Also,
\begin{equation}
    \sqrt{\frac{M}{\bar{n}_{r}}}=O(\sqrt{\log M}),
    \qquad
    \exp\left(\frac{x}{\bar{n}_{r}}\right)=O(1).
\end{equation}
Equation~\eqref{eq:DS-master-upper} therefore gives
\begin{equation}
    \QFI (\alpha,r)
    \le
    M^{2-1/(2\beta)+o(1)}\sqrt{\log M}.
\end{equation}
Dividing by Eq.~\eqref{eq:SV-lower-theorem4},
\begin{equation}
\begin{split}
    \frac{\QFI (\alpha,r)}
         {\QFI ^{\mathrm{SV}}}
    &\le
    (\log M)
    M^{\frac12-\frac{1}{2\beta}+o(1)}\\
    &=
    M^{-\frac{\eta}{2(1-\eta)}+o(1)}
    \longrightarrow0.
\end{split}
\end{equation}

It remains to consider $\eta=1$. Then
\begin{equation}
    \bar{n}_{r}=o(\bar n)=o\left(\frac{M}{\log M}\right).
\end{equation}
Along any subsequence for which $\bar{n}_{r}\to\infty$,
Eq.~\eqref{eq:DS-master-upper} applies and
\begin{equation}
    M\log t
    =
    -\frac{M}{2\bar{n}_{r}}(1+o(1)),
    \qquad
    \frac{x}{\bar{n}_{r}}
    \le
    \frac{M}{\bar{n}_{r}\log M}
    =
    o\left(\frac{M}{\bar{n}_{r}}\right).
\end{equation}
Since $M/\bar{n}_{r}\gg\log M$, the QFI is
$M^{-\omega(1)}$ and is therefore negligible compared with
Eq.~\eqref{eq:SV-lower-theorem4}.

Along any subsequence on which $\bar{n}_{r}$ is bounded, choose a fixed
$s>1$ lying inside the convergence disk of
Eq.~\eqref{eq:DS-pgf-asymptotic}. Then
Eq.~\eqref{eq:Chernoff-QFI} gives
\begin{equation}
    \QFI 
    \le
    \operatorname{poly}(M)
    \exp\left[
        -\Theta(M)
        +O\left(\frac{M}{\log M}\right)
    \right],
\end{equation}
which is again negligible compared with
Eq.~\eqref{eq:SV-lower-theorem4}. Every subsequence contains a further subsequence of one of these
two types, so the QFI ratio tends to zero also for $\eta=1$.

Finally, suppose $(\alpha_M^*,r_M^*)$ is a sequence of global
maximizers and Eq.~\eqref{eq:optimal-displacement-fraction-zero}
fails. Then along some subsequence
\begin{equation}
    \frac{(\alpha_M^*)^2}{\bar n_M}\ge\epsilon>0.
\end{equation}
By compactness of $[0,1]$, a further subsequence converges to some
$\eta\in[\epsilon,1]$. The ratio limit proved above then implies
\begin{equation}
    \frac{\QFI (\alpha_M^*,r_M^*)}
         {\QFI ^{\mathrm{SV}}}
    \longrightarrow0,
\end{equation}
contradicting global optimality, since squeezed vacuum itself is an
admissible probe and hence the ratio must be at least one. Therefore
\begin{equation}
    \frac{(\alpha_M^*)^2}{\bar n_M}\longrightarrow0.
\end{equation}
\end{proof}

\section{QFI under Optical Loss, Thermal Noise, and Phase Diffusion}
\label{apen:qfiloss}
In this appendix we describe how optical loss and thermal noise affects the QFI attainable. In Appendix~\ref{apen:noisespecifics} we present the general model for optical loss and thermal noise used here. In Appendix~\ref{apen:mixedLB} we provide a lower bound on the QFI with mixed states. In Appendices~\ref{apen:classicalnoisybound} and \ref{apen:QFIsqnoisybound} respectively we present bounds on the QFI for coherent states and squeezed vacuum states. In Appendix~\ref{sec:high-energy-loss-separation} we present a scalable advantage in the presence of constant loss and thermal noise. Finally, in Appendix~\ref{sec:phase-diffusion-no-go}, we present bounds when considering phase diffusion.

\subsection{General Model of Optical Loss and Thermal Noise}
\label{apen:noisespecifics}
When accounting for loss and thermal noise, one model is to modify Eq.\eqref{eq:inputrhopure} such that the initial input state is not pure. We can write a mixed input state (e.g. a squeezed state after loss) as
\begin{equation}
  \rho=\sum_{n,n'}c_{n,n'}\ket{n}\bra{n'}\;.
\end{equation}
 After the $m$-bit rotations this state becomes
\begin{equation}
\label{eq:inputrhomixed}
\begin{split}
        \rho&=\frac{1}{M}\sum_{k=0}^{M-1}\sum_{n,p}c_{n,p}\mathrm{e}^{\mathrm{i}2\pi k(n-p)/M}\ket{n}\bra{p}\\
        &=\frac{1}{M}\sum_{n,p}c_{n,p}\sum_{k=0}^{M-1}\mathrm{e}^{\mathrm{i}2\pi k(n-p)/M}\ket{n}\bra{p}\\
         &=\sum_{n,p}c_{n,p}\delta_{(n-p)\text{ mod }M,0}\ket{n}\bra{p}
\end{split}
\end{equation}
As the input state is mixed, the output density matrix is no longer a direct sum of pure states. Hence, although the derivative,
\begin{equation}
        \frac{\partial\rho}{\partial\delta}=i[\hat{n},\rho]\;,
\end{equation}
and QFI,
\begin{equation}
\begin{split}
    \QFI=2\sum_{j,k}\frac{|(\partial_\delta\rho)_{jk}|^2}{\lambda_j+\lambda_k}=2\sum_{j,k}\frac{|\bra{\tilde{u}_j}[\rho,\hat{n}]\ket{\tilde{u}_k}|^2}{\lambda_j+\lambda_k}\;,
\end{split}
\end{equation}
take the same form as above, the simple formula derived above is no longer possible to use here. Nevertheless, we can numerically evaluate the QFI. Furthermore, as we shall see below, for certain classes of states we can bound the QFI under this noise model.

The model we use for loss is a beam splitter with thermal noise input to one mode. We consider a squeezed state with initial covariance matrix
\begin{equation}
    \Sigma_r=\begin{pmatrix}
        \mathrm{e}^{-2r}&0\\
        0&\mathrm{e}^{2r}
    \end{pmatrix}\;,
\end{equation}
as one input to the beam splitter of transmissivity $\eta$. The second input is a thermal state with mean thermal photon number $\bar{n}_\text{th}$, with covariance matrix
\begin{equation}
    \Sigma_\text{th}=\sigma_\text{th}^2\begin{pmatrix}
       1&0\\
        0&1
    \end{pmatrix}=(1+2\bar{n}_\text{th})\begin{pmatrix}
       1&0\\
        0&1
    \end{pmatrix}\;.
\end{equation}
After the beam splitter we are left with a state with covariance matrix
\begin{equation}
    \Sigma_\text{out}=\begin{pmatrix}
    v_1&0\\
        0& v_2
    \end{pmatrix}\;,
\end{equation}
where (see also Eq.~\eqref{eq:covariancemain})
\begin{equation}
    \begin{split}
    \label{eq:v_1v_2def}
        v_1&=(1+2\bar{n}_\text{th})(1-\eta)+\eta\mathrm{e}^{-2r}\\
        v_2&=(1+2\bar{n}_\text{th})(1-\eta)+\eta\mathrm{e}^{2r}\;.
    \end{split}
\end{equation}
This can be written as an effective squeezed thermal state with covariance matrix
\begin{equation}
    \Sigma=(1+2\bar{n}_\text{th}')\begin{pmatrix}
       \mathrm{e}^{-2r'}&0\\
        0&\mathrm{e}^{2r'}
    \end{pmatrix}\;.
\end{equation}
with
\begin{equation}
    \begin{split}
    \bar{n}_\text{th}'=\frac12(\sqrt{v_1v_2}-1)\\
    r'=\frac{1}{4}\log \left(\frac{v_2}{v_1}\right)\;.
    \end{split}
\end{equation}
Similarly, the displacement of a  displaced squeezed state changes from $\alpha$ to $\sqrt{\eta}\alpha$ after the action of this channel. Based on this, we can evaluate the QFI for a lossy displaced squeezed state numerically. 

\subsection{QFI Lower Bound for Mixed States}
\label{apen:mixedLB}
We now introduce a slightly different bound than the one derived in Appendix~\ref{apen:mary} to account for the fact that we are considering a direct sum of mixed states.

\begin{lemma}\label{lem:coherence-QFI-lower} Let $\rho$ be an arbitrary state after the uniform $\mathbb Z_M$ phase twirl, and let \begin{equation}
\rho(\delta) = e^{i\delta\hat n}\rho\,e^{-i\delta\hat n}. \end{equation} 
Writing $\rho_{nm}:=\bra{n}\rho\ket{m}$, the QFI for estimating $\delta$ satisfies 
\begin{equation} 
\QFI (\rho) \ge 4M^2 \sum_{j=0}^{M-1} \frac{|\rho_{j,j+M}|^2} {\rho_{jj}+\rho_{j+M,j+M}}  , \label{eq:coherence-QFI-lower} \end{equation} 
where terms with $\rho_{jj}+\rho_{j+M,j+M}=0$ are understood to vanish. \end{lemma} 
\begin{proof} 
For each $j=0,\ldots,M-1$, write 
\begin{equation} \rho_{j+M,j} = |\rho_{j,j+M}|e^{i\chi_j}, \end{equation} 
and define
\begin{equation}
\ket{\pm_j} := \frac{ \ket{j}\pm i e^{i\chi_j}\ket{j+M} }{\sqrt{2}}. \label{eq:pairwise-measurement-states}
\end{equation}
The subspaces $\operatorname{span}\{\ket{j},\ket{j+M}\}$ are mutually orthogonal for $0\le j<M$. Hence the projectors $\Pi_{\pm,j}:=\ket{\pm_j}\!\bra{\pm_j}$ together with $\Pi_\perp := I-\sum_{j=0}^{M-1}(\Pi_{+,j}+\Pi_{-,j})$ define a valid POVM (for the squeezed vacuum this is the measurement of Eq.~\eqref{eq:sample-sector-povm}).

 \begin{equation}
\rho_{j,j+M}(\delta) = e^{-iM\delta}\rho_{j,j+M}, \end{equation}
since $\bra j e^{i\delta\hat n}\rho e^{-i\delta\hat n}\ket{j+M}=e^{i\delta j}\rho_{j,j+M}e^{-i\delta(j+M)}$. The probabilities of the two outcomes associated with residue class $j$ are are therefore
\begin{equation} p_{\pm,j}(\delta) = \frac{\rho_{jj}+\rho_{j+M,j+M}}{2} \pm |\rho_{j,j+M}|\sin(M\delta). \label{eq:pairwise-outcome-probabilities} 
\end{equation} 
At $\delta=0$, \begin{equation} 
p_{\pm,j}(0) = \frac{\rho_{jj}+\rho_{j+M,j+M}}{2}, 
\end{equation} 
while 
\begin{equation} 
\left. \frac{\partial p_{\pm,j}}{\partial\delta} \right|_{\delta=0} = \pm M|\rho_{j,j+M}|. \end{equation} 
The complementary outcome has probability independent of $\delta$, because \begin{equation} 
p_{+,j}(\delta)+p_{-,j}(\delta) = \rho_{jj}+\rho_{j+M,j+M}. \end{equation} 
Hence it contributes no Fisher information. The classical Fisher information of this POVM at $\delta=0$ is therefore 
\begin{align} \mathcal J &= \sum_{j=0}^{M-1} \sum_{\sigma=\pm} \frac{ \left[ \partial_\delta p_{\sigma,j}(\delta) \right]^2_{\delta=0} }{ p_{\sigma,j}(0) } \nonumber\\ &= 4M^2 \sum_{j=0}^{M-1} \frac{|\rho_{j,j+M}|^2} {\rho_{jj}+\rho_{j+M,j+M}}. \end{align} 
Since the QFI is the supremum of the classical Fisher information over all POVMs, \begin{equation} 
\QFI (\rho) \ge \mathcal J, \end{equation} 
which proves Eq.~\eqref{eq:coherence-QFI-lower}. \end{proof}

\subsection{Classical Benchmark Under the Thermal-Loss Channel}
\label{apen:classicalnoisybound}

We now upper bound the QFI attainable by the classical family defined in Eq.~\eqref{eq:classicaldef}. For a coherent state of energy \(x=|\alpha|^2\), recall from
Eq.~\eqref{eq:qfiubCoh2} that, whenever \(0\le x<M^2/(M+1)\),
\begin{equation}
\QFI\!\left[\mathcal E_{\delta,M}
(\ket{\alpha})\right]
\le j_{M,x},    
\end{equation}
where $j_{M,x}$ is defined in Eq.~\eqref{eq:sample-kappa-j}.

The thermal attenuator \(\mathcal G_{\eta,\bar n_{\rm th}}\) is
phase covariant and hence commutes with both the uniform
\(Z_M\) twirl and the unknown phase rotation. Moreover, it may be
decomposed as a pure-loss channel \(\mathcal L_\eta\), followed by
a parameter-independent phase-insensitive additive-noise channel
\(\mathcal N_{\eta,\bar n_{\rm th}}\): $\mathcal G_{\eta,\bar n_{\rm th}}
=\mathcal N_{\eta,\bar n_{\rm th}}\circ\mathcal L_\eta$. Since the QFI cannot increase under parameter-independent completely positive trace-preserving maps, we have
\begin{equation}
\QFI\!\left[
\mathcal G_{\eta,\bar n_{\rm th}}
\circ\mathcal E_{\delta,M}
(\ket{\alpha})
\right]
\le
\QFI\!\left[
\mathcal E_{\delta,M}
(\ket{\sqrt{\eta}\alpha})
\right].    
\end{equation}
Pure loss therefore reduces the coherent-state energy from
\(|\alpha|^2\) to \(\eta|\alpha|^2\le\eta \bar{n}\). Hence, provided
\(\eta \bar{n}<M^2/(M+1)\),
\begin{equation}
\QFI\!\left[
\mathcal G_{\eta,\bar n_{\rm th}}
\circ\mathcal E_{\delta,M}
(\ket{\alpha})
\right]
\le j_{M,\eta \bar{n}}.    
\end{equation}

Finally, using linearity of the channel and convexity and monotonicity of the QFI,
\begin{equation}
\begin{aligned}
\QFI\!\left[
\mathcal G_{\eta,\bar n_{\rm th}}
\circ\mathcal E_{\delta,M}(\rho_{\rm cl})
\right]
&\le
\QFI\!\left[
\mathcal G_{\eta,\bar n_{\rm th}}
\circ\mathcal E_{\delta,M}(\ket{\alpha})
\right] \\
&\leq j_{M,\eta \bar{n}}.
\end{aligned}    
\end{equation}
Therefore,
\begin{equation}
\sup_{\rho_{\rm cl}}
\QFI\!\left[
\mathcal G_{\eta,\bar n_{\rm th}}
\circ\mathcal E_{\delta,M}(\rho_{\rm cl})
\right]
\le j_{M,\eta \bar{n}}.    
\end{equation}

Notice that the bound is independent of
\(\bar n_{\rm th}\). The additional thermal
noise is parameter independent and can only decrease the QFI, so
discarding it produces a conservative upper bound on the classical family of states considered here.

\subsection{Squeezed Vacuum QFI under Noise Channel}
\label{apen:QFIsqnoisybound}

We now wish to lower bound the QFI of squeezed vacuum under the same noise channel. As above, since the thermal-loss channel is phase covariant, it may be
applied before the $\mathbb Z_M$ twirl. Let $\rho^\text{SV}=\mathcal G_{\eta,\bar n_{\rm th}}
(\ket{r})$ denote the density matrix obtained after a squeezed vacuum passes through the lossy channel. Then, from Eq.~\eqref{eq:coherence-QFI-lower}, we can consider only a single term to arrive at the simple bound
\begin{equation}
\QFI^{\text{SV}}
\geq
4M^2
\frac{|\rho^\text{SV}_{0,0+M}|^2}
{\rho^\text{SV}_{00}+\rho^\text{SV}_{M,M}}
\ge
4M^2|\rho_{0M}|^2\;.
\end{equation}

The coherent-state overlap is
\begin{equation}
    \bra{\alpha}\rho^\text{SV}\ket{\alpha}
    =
    \frac{2}{\sqrt D}
    \exp\left[
        -\frac{2(\operatorname{Re}\alpha)^2}{1+v_1}
        -\frac{2(\operatorname{Im}\alpha)^2}{1+v_2}
    \right].
\label{eq:noisy-sv-coherent-overlap}
\end{equation}
On the other hand,
\begin{equation}
    e^{|\alpha|^2}\bra{\alpha}\rho^\text{SV}\ket{\alpha}
    =
    \sum_{m,n=0}^{\infty}
    \frac{\rho^\text{SV}_{mn}}{\sqrt{m!n!}}
    (\alpha^*)^m\alpha^n.
\label{eq:fock-generating-identity}
\end{equation}
Writing Eq.~\eqref{eq:noisy-sv-coherent-overlap} in terms of
$\alpha$ and $\alpha^*$, the coefficient of $\alpha^2$ in the
exponent of Eq.~\eqref{eq:fock-generating-identity} is
\begin{equation}
    \frac{v_1-v_2}{2D}\;,
\end{equation}
where $D=(1+v_1)(1+v_2)$. Therefore, for even $M$,
\begin{equation}
    \rho_{0M}^\text{SV}
    =
    \frac{2}{\sqrt D}
    \frac{\sqrt{M!}}{(M/2)!}
    \left(
        \frac{v_1-v_2}{2D}
    \right)^{M/2}.
\label{eq:noisy-rho0M}
\end{equation}
The $\mathbb Z_M$ twirl leaves this matrix element unchanged,
since the two Fock levels differ by exactly $M$ photons.

Hence
\begin{equation}
    |\rho_{0M}^\text{SV}|^2
    =
    \frac{4}{D}
    \binom{M}{M/2}
    \left(
        \frac{v_2-v_1}{2D}
    \right)^M.
\end{equation}
Therefore
\begin{equation}
    \QFI
    \!\left[
        \mathcal G_{\eta,\bar n_{\rm th}}
        \circ\mathcal E_{\delta,M}
        (\ket r\!\bra r)
    \right]
    \ge
    \frac{16M^2}{D}
    \binom{M}{M/2}
    \left(
        \frac{v_2-v_1}{2D}
    \right)^M\;.
\label{eq:noisy-sv-coherence-lower}
\end{equation}

Observe that
\[
v_2-v_1
=
\eta(\mathrm{e}^{2r}-\mathrm{e}^{-2r})
=
2\eta\sinh(2r)
=
4\eta\sqrt{\bar n(1+\bar n)}.
\]
Thus
\begin{equation}
    \QFI
    \!\left[
        \mathcal G_{\eta,\bar n_{\rm th}}
        \circ\mathcal E_{\delta,M}
        (\ket r\!\bra r)
    \right]
    \ge
    \frac{16M^2}{D}
    \binom{M}{M/2}
    \left(
        \frac{2\eta\sqrt{\bar n(1+\bar n)}}{D}
    \right)^M\;.
\end{equation}

Define
\begin{equation}
    A:=1+(1-\eta)(1+2\bar n_{\rm th}),
\label{eq:sampleloss-A}
\end{equation}
which allows $D$ to be rewritten as
\begin{equation}
    D=4\eta A\bar n+(A+\eta)^2.
\end{equation}

Now take $\bar n=M/\log M$. For even $M$,
\begin{equation}
    \binom{M}{M/2}\geq\frac{2^M}{\sqrt{2M}},
\end{equation}
and hence
\begin{equation}
    \binom{M}{M/2}
    \left(
        \frac{2\eta\sqrt{\bar n(1+\bar n)}}{D}
    \right)^M
    \geq
    \frac{1}{\sqrt{2M}}
    \left(
        \frac{4\eta\sqrt{\bar n(1+\bar n)}}{D}
    \right)^M\;.
\end{equation}
For fixed $\eta$ and $\bar n_{\rm th}$, as $\bar n\to\infty$,
\begin{equation}
    D
    =
    4\eta A\bar n
    \left(1+O(\bar n^{-1})\right),
    \qquad
    \sqrt{\bar n(1+\bar n)}
    =
    \bar n
    \left(1+O(\bar n^{-1})\right).
\end{equation}
We therefore have
\begin{equation}
    \left(
        \frac{4\eta\sqrt{\bar n(1+\bar n)}}{D}
    \right)^M
    =
    A^{-M}
    \exp\left(
        O\left(\frac{M}{\bar n}\right)
    \right)
    =
    A^{-M}\exp\left(O(\log M)\right)\;.
\end{equation}
Putting this all together,
\begin{equation}
\label{eq:apensqexpdecay}
    \QFI
    \!\left[
        \mathcal G_{\eta,\bar n_{\rm th}}
        \circ\mathcal E_{\delta,M}
        (\ket r\!\bra r)
    \right]
    \ge
    \exp\left(
        -M\log A-O(\log M)
    \right)\;.
\end{equation}
Thus, for $\bar n=M/\log M$, the single-coherence bound derived here
provides only an exponentially decaying lower bound. In particular,
this argument no longer certifies the non-vanishing squeezed-vacuum
QFI obtained in the lossless setting of Theorem~\ref{th1}.

For fixed $\bar n$ we have
\begin{equation}
    \QFI
    \!\left[
        \mathcal G_{\eta,\bar n_{\rm th}}
        \circ\mathcal E_{\delta,M}
        (\ket r\!\bra r)
    \right]
    \ge
    \frac{16M^2}{D}
    \frac{1}{\sqrt{2M}}
    \left(
        \frac{4\eta\sqrt{\bar n(1+\bar n)}}{D}
    \right)^M\;.
\end{equation}
Hence, for fixed $\eta$, $\bar n_{\rm th}$, and $\bar n$, this lower bound scales as
\begin{equation}
    \QFI
    \!\left[
        \mathcal G_{\eta,\bar n_{\rm th}}
        \circ\mathcal E_{\delta,M}
        (\ket r\!\bra r)
    \right]
    =
    \Omega\!\left(M^{3/2}B^M\right),
\end{equation}
for some constant $0<B<1$.

Now consider the ratio to the classical benchmark:
\begin{equation}
\begin{split}
\frac{
\QFI
\!\left[
    \mathcal G_{\eta,\bar n_{\rm th}}
    \circ\mathcal E_{\delta,M}
    (\ket r\!\bra r)
\right]
}{
\displaystyle
\sup_{\rho_{\rm cl}}
\QFI
\!\left[
    \mathcal G_{\eta,\bar n_{\rm th}}
    \circ\mathcal E_{\delta,M}(\rho_{\rm cl})
\right]
}
&\geq
\frac{
\displaystyle
\frac{16M^2}{D}
\frac{1}{\sqrt{2M}}
\left(
    \frac{4\eta\sqrt{\bar n(1+\bar n)}}{D}
\right)^M
}{
\displaystyle
\frac{
4M^2\mathrm{e}^{-\eta\bar n}(\eta\bar n)^M
}{
M!\left(
1-\eta\bar n\frac{M+1}{M^2}
\right)
}
}
\\[1mm]
&=
\frac{4}{D}
\frac{\mathrm{e}^{\eta\bar n}M!}{\sqrt{2M}}
\left(
1-\eta\bar n\frac{M+1}{M^2}
\right)
\left(
    \frac{4\sqrt{\bar n(1+\bar n)}}{D\bar n}
\right)^M .
\end{split}
\label{eq:noisy-fixed-nbar-ratio}
\end{equation}

For fixed $\eta$ and $\bar n$,
\begin{equation}
    \eta\bar n\frac{M+1}{M^2}\longrightarrow0,
\end{equation}
and therefore, for all sufficiently large $M$,
\begin{equation}
    1-\eta\bar n\frac{M+1}{M^2}\geq\frac12.
\end{equation}
Consequently,
\begin{equation}
\begin{split}
\frac{
\QFI
\!\left[
    \mathcal G_{\eta,\bar n_{\rm th}}
    \circ\mathcal E_{\delta,M}
    (\ket r\!\bra r)
\right]
}{
\displaystyle
\sup_{\rho_{\rm cl}}
\QFI
\!\left[
    \mathcal G_{\eta,\bar n_{\rm th}}
    \circ\mathcal E_{\delta,M}(\rho_{\rm cl})
\right]
}
&\geq
\frac{2}{D}
\frac{\mathrm{e}^{\eta\bar n}M!}{\sqrt{2M}}
\left(
    \frac{4\sqrt{\bar n(1+\bar n)}}{D\bar n}
\right)^M
\\
&=
\Omega\!\left[
\left(
    \frac{4}{D}
    \frac{M}{\mathrm e}
    \sqrt{\frac{1+\bar n}{\bar n}}
\right)^M
\right],
\end{split}
\end{equation}
where we used
\begin{equation}
    \frac{M!}{\sqrt{2M}}
    =
    \sqrt{\pi}
    \left(\frac{M}{\mathrm e}\right)^M
    \left[1+o(1)\right].
\end{equation}
For fixed $D$ and $\bar n$, the $M^M$ factor eventually dominates,
yielding the desired factorial separation.

\subsection{A Scalable Separation Under Fixed Loss and Thermal Noise}
\label{sec:high-energy-loss-separation}

We now consider an energy regime in which a squeezed-vacuum probe
retains non-vanishing Fisher information under fixed thermal loss, whereas
the QFI of every state in the classical family of Eq.~\eqref{eq:classicaldef} is
super-polynomially small. Let $M\to\infty$ be even and fix a transmissivity
$0<\eta\le1$ and thermal noise $\bar{n}_\text{th}\geq0$ independent of $M$. Let $\bar{n}$ simultaneously satisfy
\begin{equation}
    \frac{\bar{n}(\log M)^2}{M^2}\longrightarrow\infty,
    \qquad
    \frac{\bar{n}\log M}{M^2}\longrightarrow0.
\label{eq:loss-energy-window}
\end{equation}

\subsubsection{Squeezed Vacuum Lower Bound}
We shall first compute the CFI of heterodyne detection for the squeezed vacuum state, before proceeding to the classical states defined in Eq.~\eqref{eq:classicaldef}. Let $\ket{r_M}=S(r_M)\ket0$ be the squeezed vacuum satisfying $\sinh^2r_M=\bar{n}$. We consider the thermal loss channel described in Appendix~\ref{apen:noisespecifics}.

Recalling the definition of $A$ from Eq.~\eqref{eq:sampleloss-A} we have
\begin{equation}
    a_M:=1+v_1
    =A+\eta e^{-2r_M},
    \qquad
    b_M:=1+v_2
    =A+\eta e^{2r_M}\;.
\end{equation}
The quantities $a_M$ and $b_M$ have a simple physical
interpretation. Heterodyne detection realizes the POVM
\begin{equation}
    \Pi(\alpha)=\frac{1}{\pi}\ket{\alpha}\!\bra{\alpha},
\end{equation}
A heterodyne outcome is therefore a complex amplitude
$\alpha=re^{i\theta}$~\cite{Leonhardt1997,DArianoParisSacchi2004}.

For the lossy squeezed vacuum, the heterodyne probability density
before the phase twirl is
\begin{equation}
    p_{\rm het}(r,\theta)
    =
    \frac{2}{\pi\sqrt{a_Mb_M}}
    \exp\left[
        -2r^2\left(
            \frac{\cos^2\theta}{a_M}
            +
            \frac{\sin^2\theta}{b_M}
        \right)
    \right]\;.
\end{equation}
Thus the heterodyne outcomes form an elliptical Gaussian cloud in
phase space. The added unity in $a_M$ and $b_M$ is the vacuum noise
inherent to heterodyne detection.

We retain only the polar angle of each heterodyne outcome. Its
probability density is
\begin{equation}
\begin{split}
    h_M(\theta)
    &:=
    \int_0^\infty
    p_{\rm het}(r,\theta)\,r\,\mathrm dr
\\
    &=
    \frac{\sqrt{a_Mb_M}}
    {2\pi\left(
        b_M\cos^2\theta+a_M\sin^2\theta
    \right)}.
\end{split}
\label{eq:heterodyne-angular-marginal}
\end{equation}
Hence $h_M(\theta)$ is directly the experimentally observable
probability distribution of the angle of the heterodyne outcome,
after its radial coordinate has been discarded.

Defining
\begin{equation}
\label{eq:heterodyne-xi}
    \xi_M
    :=
    \frac{\sqrt{b_M}-\sqrt{a_M}}
         {\sqrt{b_M}+\sqrt{a_M}},
\end{equation}
the angular distribution has Fourier expansion
\begin{equation}
    h_M(\theta)
    =
    \frac{1}{2\pi}
    \left[
        1+
        2\sum_{k=1}^{\infty}
        (-\xi_M)^k\cos(2k\theta)
    \right].
\end{equation}
Thus $\xi_M$ quantifies the angular anisotropy of the heterodyne
cloud: $\xi_M=0$ for a circular distribution, while
$\xi_M\rightarrow1$ for an increasingly elongated distribution.

After the unknown phase rotation $\delta$ and the uniform
$\mathbb Z_M$ twirl, the observed angular distribution is
\begin{equation}
    \overline h_{M,\delta}(\theta)
    =
    \frac1M
    \sum_{j=0}^{M-1}
    h_M\left(
        \theta-\delta-\frac{2\pi j}{M}
    \right).
\label{eq:heterodyne-twirled-angle}
\end{equation}
For even $M$, the first Fourier component surviving this average has
angular frequency $M$ and amplitude
\begin{equation}
    R_M=\xi_M^{M/2}.
\end{equation}
 Up to an irrelevant fixed shift of the
angular origin, the twirled angular distribution is therefore
\begin{equation}
    h_{M,\delta}(\theta)
    =
    \frac{1-R_M^2}
    {2\pi\left[
        1-2R_M\cos\big(M(\theta-\delta)\big)+R_M^2
    \right]}\;.
\label{eq:wrapped-Cauchy-angle}
\end{equation}
Equation~\eqref{eq:wrapped-Cauchy-angle} is a wrapped-Cauchy
location family. We use only the angular coordinate of the heterodyne outcome. The resulting quantity is therefore a classical Fisher information, and hence a lower bound on the QFI. For the location family in
Eq.~\eqref{eq:wrapped-Cauchy-angle},
\begin{equation}
    \partial_\delta\log h_{M,\delta}(\theta)
    =
    \frac{
        2MR_M\sin[M(\theta-\delta)]
    }{
        1-2R_M\cos[M(\theta-\delta)]+R_M^2
    }.
\end{equation}
Direct integration gives the CFI for $\delta$ as
\begin{equation}
\begin{split}
    \mathcal J_{\rm het}
    &:=
    \int_0^{2\pi}
    h_{M,\delta}(\theta)
    \left[
        \partial_\delta\log h_{M,\delta}(\theta)
    \right]^2
    \mathrm d\theta
\\
    &=
    \frac{2M^2R_M^2}{(1-R_M^2)^2}
    =
    \frac{2M^2\xi_M^M}{(1-\xi_M^M)^2}.
\end{split}
\label{eq:exact-heterodyne-CFI}
\end{equation}
Since this is the Fisher information of an explicit measurement,
\begin{equation}
    \QFI^{\rm SV}
    \ge
    \mathcal J_{\rm het}.
\end{equation}

For $\bar{n}=\sinh^2r_M\to\infty$,
\begin{equation}
    a_M=A+O(\bar{n}^{-1}),
    \qquad
    b_M=4\eta \bar{n}+O(1).
\end{equation}
From Eq.~\eqref{eq:heterodyne-xi},
\begin{equation}
    -\log\xi_M
    =
    \log\left(
        \frac{\sqrt{b_M}+\sqrt{a_M}}
             {\sqrt{b_M}-\sqrt{a_M}}
    \right)
    =
    2\,\operatorname{arctanh}
    \sqrt{\frac{a_M}{b_M}}.
\label{eq:xi-arctanh}
\end{equation}
and so (using $\operatorname{arctanh}x=x+O(x^3)$)
\begin{equation}
    -\log\xi_M
    =
    \sqrt{\frac{A}{\eta \bar{n}}}
    +O(\bar{n}^{-3/2}).
\label{eq:xi-large-energy}
\end{equation}
Consequently,
\begin{equation}
    \xi_M^M
    =
    \exp\left[
        -\sqrt{\frac{A}{\eta}}\,
        \frac{M}{\sqrt{\bar{n}}}
        +O\left(\frac{M}{\bar{n}^{3/2}}\right)
    \right].
\label{eq:xiM-general-scaling}
\end{equation}

The two assumptions in Eq.~\eqref{eq:loss-energy-window} imply
\begin{equation}
    \sqrt{\log M}
    \ll
    \frac{M}{\sqrt{\bar{n}}}
    \ll
    \log M.
\end{equation}
Hence
\begin{equation}
    \xi_M^M
    =
    \exp\left[
        -\sqrt{\frac{A}{\eta}}\,
        \frac{M}{\sqrt{\bar n}}
        (1+o(1))
    \right].
\end{equation}
In particular,
\begin{equation}
    \xi_M^M\longrightarrow0,
    \qquad
    -\log(\xi_M^M)=o(\log M).
\end{equation}
Thus $\xi_M^M$ vanishes subpolynomially in $M$: it tends to zero,
but more slowly than $M^{-k}$ for every fixed $k>0$.
Substitution into Eq.~\eqref{eq:exact-heterodyne-CFI} gives
\begin{equation}
\label{eq:sqvachetLB}
    \mathcal J_{\rm het,M}^{\rm SV}
    =
    2M^2\xi_M^M(1+o(1))
    =
    M^{2-o(1)}.
\end{equation}

\subsubsection{Classical Upper Bound}

For the classical upper bound, we consider coherent states subject only to pure loss. This is because the thermal-loss channel can be
decomposed as a pure-loss channel \(\mathcal L_\eta\), followed by
a parameter-independent phase-insensitive additive-noise channel
\(\mathcal N_{\eta,\bar n_{\rm th}}\): $\mathcal G_{\eta,\bar n_{\rm th}}
=\mathcal N_{\eta,\bar n_{\rm th}}\circ\mathcal L_\eta$ as in Appendix~\ref{apen:classicalnoisybound}. As the QFI is non-increasing under $\mathcal N_{\eta,\bar n_{\rm th}}$, the QFI after $\mathcal G_{\eta,\bar n_{\rm th}}$ is at most the QFI after $L_\eta$, and so it suffices to drop $\mathcal N_{\eta,\bar n_{\rm th}}$ here.

After pure loss, a coherent state of input energy $\bar{n}$ is a
coherent state whose photon-number distribution is Poisson with mean
\begin{equation}
    \mu_M:=\eta \bar{n}.
\end{equation}
Let $N\sim{\rm Pois}(\mu_M)$ and let
$R:=N\bmod M$. For a twirled coherent state, Eq.~\eqref{eq:QGIgeneralVar}
\begin{equation}
    \QFI^{\rm C}
    =
    4\mathbb E\left[
        \operatorname{Var}(N\mid R)
    \right].
\label{eq:coherent-conditional-variance-high-energy}
\end{equation}

For every residue $r\in\{0,\ldots,M-1\}$, let $a_r$ be an integer
congruent to $r$ modulo $M$ that is closest to $\mu_M$. Since the
conditional mean minimizes conditional mean-squared error,
\begin{equation}
\begin{split}
    \frac{\QFI^{\rm C}}{4}
    &\le
    \mathbb E\left[
        \big(N-a_{N\bmod M}\big)^2
    \right].
\end{split}
\end{equation}
If $N\neq a_{N\bmod M}$, then
\begin{equation}
    |N-\mu_M|\ge\frac M2,
    \qquad
    |N-a_{N\bmod M}|
    \le2|N-\mu_M|.
\end{equation}
Therefore
\begin{equation}
    \QFI^{\rm C}
    \le
    16\,
    \mathbb E\left[
        (N-\mu_M)^2
        \mathbf 1_{\{|N-\mu_M|\ge M/2\}}
    \right].
\label{eq:coherent-tail-second-moment}
\end{equation}

For sufficiently large $M$, the assumed energy window gives
$\mu_M\ge M/2$. The Poisson Chernoff bound then yields
\begin{equation}
    \Pr\left(|N-\mu_M|\ge\frac M2\right)
    \le
    2\exp\left(-\frac{M^2}{12\mu_M}\right).
\end{equation}
Furthermore,
\begin{equation}
    \mathbb E[(N-\mu_M)^4]
    =
    \mu_M+3\mu_M^2
    \le4\mu_M^2
\end{equation}
for all sufficiently large $M$. By Cauchy--Schwarz,
Eq.~\eqref{eq:coherent-tail-second-moment} gives
\begin{equation}
    \QFI^{\rm C}
    \le
    32\sqrt2\,\mu_M
    \exp\left(-\frac{M^2}{24\mu_M}\right)\;.
\label{eq:coherent-high-energy-upper}
\end{equation}

Every coherent state with input energy at most $\bar{n}$ is obtained
from the coherent state of input energy $\bar{n}$ by an additional
parameter-independent pure-loss channel. QFI monotonicity and
convexity therefore extend Eq.~\eqref{eq:coherent-high-energy-upper}
to the full family of classical states defined in the main text:
\begin{equation}
    \sup_{\rho_{\rm cl}}
    \QFI( \mathcal G_{\eta,\bar n_{\rm th}}
        \circ\mathcal E_{\delta,M}
        (\rho_{\rm cl}))
    \le
    32\sqrt2\,\eta \bar{n}
    \exp\left(-\frac{M^2}{24\eta \bar{n}}\right).
\label{eq:classical-high-energy-upper}
\end{equation}
Since
\begin{equation}
    \bar{n}=o\left(\frac{M^2}{\log M}\right),
\end{equation}
we have
\begin{equation}
    \frac{M^2}{\bar{n}}=\omega(\log M).
\end{equation}
The polynomial prefactor in
Eq.~\eqref{eq:classical-high-energy-upper} is therefore dominated by
the exponential term, proving
\begin{equation}
    \sup_{\rho_{\rm cl}}
    \QFI(G_{\eta,\bar n_{\rm th}}
        \circ\mathcal E_{\delta,M}
        (\rho_{\rm cl}))
    =
    M^{-\omega(1)}.
\end{equation}

Combining with Eq.~\eqref{eq:sqvachetLB} we arrive at
\begin{equation}
    \frac{
        \mathcal J_{\rm het,M}^{\rm SV}
    }{
        \displaystyle
        \sup_{\rho_{\rm cl}}
        \QFI(G_{\eta,\bar n_{\rm th}}
        \circ\mathcal E_{\delta,M}
        (\rho_{\rm cl}))
    }
   \ge
    \exp\left[
        \Omega\left(\frac{M^2}{\bar n}\right)
    \right]
    = 
    M^{\omega(1)}.
\end{equation}

\subsection{A State-Independent No-Go Theorem for Phase Diffusion}
\label{sec:phase-diffusion-no-go}

Finally we consider a phase diffusion channel, $\mathcal D_{\sigma}$ that acts on a density matrix as
\begin{equation}
    \mathcal D_\sigma(\ket n\!\bra m)
    =
    e^{-\sigma^2(n-m)^2/2}
    \ket n\!\bra m.
\end{equation}
In a completely general model $\sigma$ may be a function of $M$, which we denote as $\sigma_M$. Assuming $\sigma_M>0$, we now bound the QFI for any input state under this channel.

Let $g_\sigma(\phi)$ denote the wrapped Gaussian phase distribution,
\begin{equation}
    g_\sigma(\phi)
    =
    \frac{1}{2\pi}
    \sum_{k\in\mathbb Z}
    e^{-k^2\sigma^2/2}e^{ik\phi}.
\label{eq:wrapped-Gaussian}
\end{equation}
The uniform phase twirl corresponds to an additional
random phase taking the values $2\pi j/M$, $j=0,\ldots,M-1$.
Hence the probability density of the total random phase is
\begin{equation}
    f_{M,\sigma}(\phi)
    :=
    \frac{1}{M}
    \sum_{j=0}^{M-1}
    g_\sigma\left(\phi-\frac{2\pi j}{M}\right).
\label{eq:twirled-phase-density-def}
\end{equation}
Substituting Eq.~\eqref{eq:wrapped-Gaussian} gives
\begin{equation}
\begin{split}
    f_{M,\sigma}(\phi)
    &=
    \frac{1}{2\pi}
    \sum_{k\in\mathbb Z}
    e^{-k^2\sigma^2/2}e^{ik\phi}
    \left[
        \frac{1}{M}
        \sum_{j=0}^{M-1}
        e^{-2\pi i jk/M}
    \right].
\end{split}
\end{equation}
The term in square brackets is equal to one when $k$ is a multiple
of $M$ and vanishes otherwise. Writing $k=\ell M$, we therefore obtain
\begin{equation}
\begin{split}
    f_{M,\sigma}(\phi)
    &=
    \frac{1}{2\pi}
    \sum_{\ell\in\mathbb Z}
    e^{-\ell^2M^2\sigma^2/2}
    e^{i\ell M\phi}
\\
    &=
    \frac{1}{2\pi}
    \left[
        1+
        2\sum_{\ell=1}^{\infty}
        e^{-\ell^2M^2\sigma^2/2}
        \cos(\ell M\phi)
    \right].
\end{split}
\label{eq:wrapped-twirled-phase-density}
\end{equation}

Since the discrete twirl, Gaussian phase diffusion, and unknown phase
rotation are all rotation operations they commute. Their
composition can therefore be written as
\begin{equation}
    \Phi_{\delta,M,\sigma}(\rho)
    =
    \int_0^{2\pi}
    f_{M,\sigma}(\phi-\delta)
    e^{-i\phi\hat n}\rho e^{i\phi\hat n}
    \,\mathrm d\phi.
\label{eq:phase-channel-random-unitary}
\end{equation}

Next introduce a classical register that stores the random phase:
\begin{equation}
    \widetilde\rho_\delta
    :=
    \int_0^{2\pi}
    f_{M,\sigma}(\phi-\delta)
    \ket{\phi}\!\bra{\phi}_{\rm C}
    \otimes
    e^{-i\phi\hat n}\rho e^{i\phi\hat n}
    \,\mathrm d\phi.
\label{eq:phase-labelled-extension}
\end{equation}
Tracing out the classical register gives
$\Phi_{\delta,M,\sigma}(\rho)$. The monotonicity of the QFI under partial trace
therefore implies
\begin{equation}
    \QFI
    \left[
        \Phi_{\delta,M,\sigma}(\rho)
    \right]
    \le
    \QFI(\widetilde\rho_\delta).
\end{equation}
Conditioned on $\phi$, the quantum state in
Eq.~\eqref{eq:phase-labelled-extension} is independent of $\delta$.
The QFI of the classical--quantum state is consequently just the
classical Fisher information of the phase label:
\begin{equation}
    \QFI(\widetilde\rho_\delta)
    =
    \int_0^{2\pi}
    \frac{
        [\partial_\delta
        f_{M,\sigma}(\phi-\delta)]^2
    }{
        f_{M,\sigma}(\phi-\delta)
    }
    \,\mathrm d\phi
    =
    \mathcal I_{\rm cl}[f_{M,\sigma}].
\end{equation}
This is
independent of $\rho$.

Equation~\eqref{eq:wrapped-twirled-phase-density} gives
\begin{equation}
    f_{M,\sigma}(\phi)
    \ge
    \frac{1-2S_{M,\sigma}}{2\pi}\;,
\end{equation}
where
\begin{equation}
    S_{M,\sigma}
    :=
    \sum_{\ell=1}^{\infty}
    e^{-\ell^2M^2\sigma^2/2}\;.
\end{equation}
Moreover, orthogonality of the Fourier modes gives
\begin{equation}
\begin{split}
    \int_0^{2\pi}
    [f_{M,\sigma}'(\phi)]^2
    \,\mathrm d\phi
    &=
    \frac{M^2}{\pi}
    \sum_{\ell=1}^{\infty}
    \ell^2e^{-\ell^2M^2\sigma^2}.
\end{split}
\end{equation}
Therefore
\begin{equation}
\begin{split}
    \mathcal I_{\rm cl}[f_{M,\sigma}]
    &=
    \int_0^{2\pi}
    \frac{[f_{M,\sigma}'(\phi)]^2}
         {f_{M,\sigma}(\phi)}
    \,\mathrm d\phi
\\
    &\le
    \frac{
        2M^2
        \displaystyle\sum_{\ell=1}^{\infty}
        \ell^2e^{-\ell^2M^2\sigma^2}
    }{
        1-
        2\displaystyle\sum_{\ell=1}^{\infty}
        e^{-\ell^2M^2\sigma^2/2}
    }\;.
\end{split}
\end{equation}
If $M\sigma_M\to\infty$, both sums are dominated by their
$\ell=1$ terms:
\begin{equation}
    \sum_{\ell=1}^{\infty}
    e^{-\ell^2M^2\sigma_M^2/2}
    =
    e^{-M^2\sigma_M^2/2}(1+o(1)),
\end{equation}
and
\begin{equation}
    \sum_{\ell=1}^{\infty}
    \ell^2e^{-\ell^2M^2\sigma_M^2}
    =
    e^{-M^2\sigma_M^2}(1+o(1)).
\end{equation}
If $M\sigma_M\to\infty$, substitution gives
\begin{equation}
    \sup_\rho
    \QFI
    \left[
        \Phi_{\delta,M,\sigma_M}(\rho)
    \right]
    \le
    2M^2e^{-M^2\sigma_M^2}
    (1+o(1)).
\end{equation}
Consequently,
\begin{equation}
    M^2\sigma_M^2-2\log M\longrightarrow+\infty
\end{equation}
implies
\begin{equation}
    \sup_\rho
    \QFI
    \left[
        \Phi_{\delta,M,\sigma_M}(\rho)
    \right]
    \longrightarrow0.
\end{equation}
In particular, this holds for every fixed $\sigma_M=\sigma>0$,
independently of the input state's mean photon number or photon-number
variance.

Note that if
\begin{equation}
    \sigma_M^2=O(M^{-2}),
\end{equation}
the QFI can remain constant.

\section{Finite-sample separation under optical loss and thermal noise}
\label{apen:samplecomplexityloss}

Appendix~\ref{apen:samplecomplexity} established a finite-sample separation for the ideal phase-twirling channel. We now prove that the separation survives the experimentally relevant thermal-loss model studied in Appendix~\ref{apen:qfiloss}. The statistical task is unchanged: throughout this appendix we use the same constants $a,b,\zeta$ and local parameter set $\Theta_M=[-a/M,a/M]$ from Eq.~\eqref{eq:sample-constants}, and success means Eq.~\eqref{eq:sample-success}. We assume that the thermal-loss channel acts independently on each use and that the $\mathbb Z_M$ phase randomisation is also independent across uses.

Define the noisy single-use channel
\begin{equation}
    \Phi_{\delta,M}^{\eta,\bar n_{\rm th}}
    :=
    \mathcal G_{\eta,\bar n_{\rm th}}
    \circ
    \mathcal E_{\delta,M},
\label{eq:sampleloss-channel}
\end{equation}
where $0<\eta\le1$ is the transmissivity and $\bar n_{\rm th}\ge0$ is the mean thermal photon number in the unused port. Since $\mathcal G_{\eta,\bar n_{\rm th}}$ is phase covariant, its position relative to the phase twirl and the unknown phase rotation is immaterial. Let
\begin{equation}
    \nu_{\rm cl}^{\star,\eta,\bar n_{\rm th}}(M,\bar{n};a,b,\zeta)
\end{equation}
denote the minimum number of noisy channel uses required by the classical family of Eq.~\eqref{eq:sample-correlated-classical}, allowing arbitrary classical correlations and arbitrary collective output measurements. As in Appendix~\ref{apen:samplecomplexity}, the same notation also includes adaptive classical measure-and-prepare protocols with the constraint that each state is the family of classical states defined by Eq.~\eqref{eq:classicaldef}. Let
\begin{equation}
    \nu_{\rm SV}^{\star,\eta,\bar n_{\rm th}}(M,\bar{n};a,b,\zeta)
\end{equation}
be the corresponding minimum when each probe is a squeezed vacuum with input mean photon number $\bar{n}$.

In Appendix~\ref{apen:samplethermal1} we establish general bounds when including thermal loss. In Appendices~\ref{apen:samplenoisyfixedenergy} and \ref{sec:samplesthermalloss} we specialise this to constant mean photon number and mean photon number of $\bar{n}=M/\log(M)$ respectively. Finally in Appendix~\ref{sec:sampleloss-high-energy} we present a scalable advantage in the presence of loss.

\subsection{Nonasymptotic thermal-loss bounds}
\label{apen:samplethermal1}

It is useful to isolate a single quantity that controls a constructive noisy squeezed-vacuum receiver. Let
\begin{equation}
    \sigma_{\bar{n}}^{\eta,\bar n_{\rm th}}
    :=\mathcal G_{\eta,\bar n_{\rm th}}
      (\ket{r_{\bar{n}}}\!\bra{r_{\bar{n}}}),
    \qquad
    \sinh^2 r_{\bar{n}}=\bar{n},
\label{eq:sampleloss-sigma}
\end{equation}
and define
\begin{equation}
    c_{M,\bar{n}}^{\eta,\bar n_{\rm th}}
    :=\left|\bra0\sigma_{\bar{n}}^{\eta,\bar n_{\rm th}}\ket M\right|,
    \qquad
    s_{M,\bar{n}}^{\eta,\bar n_{\rm th}}
    :=\bra0\sigma_{\bar{n}}^{\eta,\bar n_{\rm th}}\ket0
      +\bra M\sigma_{\bar{n}}^{\eta,\bar n_{\rm th}}\ket M,
\label{eq:sampleloss-c-s}
\end{equation}
and the visibility (i.e. the thermal-loss analogue of Eq.~\eqref{eq:sample-sector-visibility})
\begin{equation}
    B_{M,\bar{n}}^{\eta,\bar n_{\rm th}}
    :=
    \frac{4\left(c_{M,\bar{n}}^{\eta,\bar n_{\rm th}}\right)^2}
         {s_{M,\bar{n}}^{\eta,\bar n_{\rm th}}}.
\label{eq:sampleloss-visibility}
\end{equation}
Positivity of the $\{\ket0,\ket M\}$ submatrix gives
$c_{M,\bar{n}}^{\eta,\bar n_{\rm th}}\le s_{M,\bar{n}}^{\eta,\bar n_{\rm th}}/2$, and therefore $0<B_{M,\bar{n}}^{\eta,\bar n_{\rm th}}\le1$ whenever $\bar{n}>0$ and $\eta>0$.

\begin{theorem}
\label{thm:sampleloss-finite}
Let $M\ge2$ be even, $\bar{n}>0$, $0<\eta\le1$, and $\bar n_{\rm th}\ge0$. Assume $\eta \bar{n}<\frac{M^2}{M+1}$, and define
\begin{equation}
    C_{a,b,\zeta}
    :=
    \frac{2+\frac43 b\cos a}{b^2\cos^2 a}
    \log\!\left(\frac2\zeta\right).
\label{eq:sampleloss-C}
\end{equation}
Then
\begin{equation}
    \nu_{\rm cl}^{\star,\eta,\bar n_{\rm th}}
    \ge
    \frac{(1-2\zeta)^2}{16b^2}
    \left(1-\eta \bar{n}\frac{M+1}{M^2}\right)
    e^{\eta \bar{n}}\frac{M!}{(\eta \bar{n})^M}
\label{eq:sampleloss-classical-lower}
\end{equation}
and
\begin{equation}
    \nu_{\rm SV}^{\star,\eta,\bar n_{\rm th}}
    \le
    \left\lceil
        \frac{C_{a,b,\zeta}}
        {B_{M,\bar{n}}^{\eta,\bar n_{\rm th}}}
    \right\rceil .
\label{eq:sampleloss-sv-upper-exact}
\end{equation}
Moreover, with $v_1,v_2$ as in Eq.~\eqref{eq:v_1v_2def} and
\begin{equation}
    D
    :=(1+v_1)(1+v_2),
\label{eq:sampleloss-D}
\end{equation}
the exact noisy coherence derived in Appendix~\ref{apen:QFIsqnoisybound} gives
\begin{equation}
\begin{split}
    \left(c_{M,\bar{n}}^{\eta,\bar n_{\rm th}}\right)^2
    & =
    \frac{4}{D}
    \binom{M}{M/2}
    \left(
        \frac{2\eta\sqrt{\bar{n}(1+\bar{n})}}
             {D}
    \right)^M,\\
    B_{M,\bar{n}}^{\eta,\bar n_{\rm th}}
    &\ge
    \frac{16}{D}
    \binom{M}{M/2}
    \left(
        \frac{2\eta\sqrt{\bar{n}(1+\bar{n})}}
             {D}
    \right)^M.
\end{split}
\label{eq:sampleloss-coherence-visibility}
\end{equation}
Consequently Eq.~\eqref{eq:sampleloss-sv-upper-exact} has the fully explicit form
\begin{equation}
    \nu_{\rm SV}^{\star,\eta,\bar n_{\rm th}}
    \le
    \left\lceil
    \frac{C_{a,b,\zeta}D}
         {16\binom{M}{M/2}}
    \left(
        \frac{D}
             {2\eta\sqrt{\bar{n}(1+\bar{n})}}
    \right)^M
    \right\rceil .
\label{eq:sampleloss-sv-upper-explicit}
\end{equation}
\end{theorem}

\begin{proof}
We begin with the classical lower bound. Appendix~\ref{apen:qfiloss} proves that a coherent input with $|\alpha|^2\le \bar{n}$ obeys
\begin{equation}
    \QFI\!\left[
        \Phi_{\delta,M}^{\eta,\bar n_{\rm th}}
        (\ket\alpha\!\bra\alpha)
    \right]
    \le j_{M,\eta \bar{n}}:=
    \frac{4M^2e^{-\eta \bar{n}}(\eta \bar{n})^M}
    {M!\left[1-\eta \bar{n}(M+1)/M^2\right]}\;,
\label{eq:sampleloss-single-classical-qfi}
\end{equation}
where $j_{M,\eta \bar{n}}$ is introduced in Eq.~\eqref{eq:sample-kappa-j}.
The thermal attenuator can be written as pure loss followed by parameter-independent phase-insensitive additive noise. Hence the bound is independent of $\bar n_{\rm th}$: dropping the additive-noise stage can only increase the QFI. Eq.~\eqref{eq:sampleloss-classical-lower} then follows in an identical manner to the proof of the classical bound in Theorem~\ref{thm:sample-finite}. For adaptive classical measure-and-prepare protocols, the conditional Fisher information in every round is at most $j_{M,\eta \bar{n}}$; the martingale-score argument of Eq.~\eqref{eq:sample-adaptive-fisher} therefore gives a total transcript Fisher information at most $\nu j_{M,\eta \bar{n}}$, and the same two-point testing argument applies.

We now construct the noisy squeezed-vacuum protocol. Write
\begin{equation}
    \bra0\sigma_{\bar{n}}^{\eta,\bar n_{\rm th}}\ket M
    =c_{M,\bar{n}}^{\eta,\bar n_{\rm th}}e^{i\chi_M}.
\end{equation}
Because the two Fock levels differ by exactly $M$, the $\mathbb Z_M$ twirl leaves this coherence unchanged. Define
\begin{equation}
    \ket{\pm_M}
    :=
    \frac{\ket0\pm i e^{-i\chi_M}\ket M}{\sqrt2},
    \qquad
    \Pi_\pm:=\ket{\pm_M}\!\bra{\pm_M},
    \qquad
    \Pi_\perp:=I-\Pi_+-\Pi_-.
\label{eq:sampleloss-povm}
\end{equation}
Up to interchanging the labels $+$ and $-$, the outcome probabilities are
\begin{equation}
    P_\delta(\pm)
    =
    \frac{s_{M,\bar{n}}^{\eta,\bar n_{\rm th}}}{2}
    \pm
    c_{M,\bar{n}}^{\eta,\bar n_{\rm th}}\sin(M\delta),
\label{eq:sampleloss-probabilities}
\end{equation}
whereas $P_\delta(\perp)=1-s_{M,\bar{n}}^{\eta,\bar n_{\rm th}}$ is independent of $\delta$. Put
\begin{equation}
    g_M:=
    \frac{2c_{M,\bar{n}}^{\eta,\bar n_{\rm th}}}
         {s_{M,\bar{n}}^{\eta,\bar n_{\rm th}}}
    \in[0,1]
\end{equation}
and assign the single-shot statistic
\begin{equation}
    Z:=
    \begin{cases}
        g_M,&+,\\
        -g_M,&-,\\
        0,&\perp.
    \end{cases}
\label{eq:sampleloss-statistic}
\end{equation}
Equations~\eqref{eq:sampleloss-visibility} and~\eqref{eq:sampleloss-probabilities} give
\begin{equation}
    \mathbb E_\delta[Z]
    =B_{M,\bar{n}}^{\eta,\bar n_{\rm th}}\sin(M\delta),
    \qquad
    \mathbb E_\delta[Z^2]
    =B_{M,\bar{n}}^{\eta,\bar n_{\rm th}},
    \qquad
    \operatorname{Var}_\delta(Z)
    \le B_{M,\bar{n}}^{\eta,\bar n_{\rm th}}.
\label{eq:sampleloss-moments}
\end{equation}
For $\overline Z_\nu=\nu^{-1}\sum_{k=1}^{\nu}Z_k$, use the estimator
\begin{equation}
    \widehat\delta_\nu
    :=
    \frac1M\arcsin\!\left[
        \operatorname{clip}\!\left(
            \frac{\overline Z_\nu}
                   {B_{M,\bar{n}}^{\eta,\bar n_{\rm th}}},
            -\sin a,\sin a
        \right)
    \right].
\label{eq:sampleloss-estimator}
\end{equation}
Exactly as in Eq.~\eqref{eq:sample-estimator-error}, projection is nonexpansive and $\arcsin$ is $(\cos a)^{-1}$-Lipschitz on the clipped interval, so
\begin{equation}
    |\widehat\delta_\nu-\delta|
    \le
    \frac{|\overline Z_\nu-\mathbb E_\delta Z|}
    {M B_{M,\bar{n}}^{\eta,\bar n_{\rm th}}\cos a}.
\end{equation}
Since $|Z-\mathbb E_\delta Z|\le2$, Bernstein's inequality with $t=bB_{M,\bar{n}}^{\eta,\bar n_{\rm th}}\cos a$ yields
\begin{equation}
    \sup_{\delta\in\Theta_M}
    \Pr_\delta\!\left(
        |\widehat\delta_\nu-\delta|>\frac bM
    \right)
    \le
    2\exp\!\left[
        -\frac{\nu b^2\cos^2a\,B_{M,\bar{n}}^{\eta,\bar n_{\rm th}}}
        {2+\frac43b\cos a}
    \right].
\label{eq:sampleloss-confidence}
\end{equation}
This proves Eq.~\eqref{eq:sampleloss-sv-upper-exact}. Finally, $s_{M,\bar{n}}^{\eta,\bar n_{\rm th}}\le1$ implies $B_{M,\bar{n}}^{\eta,\bar n_{\rm th}}\ge4(c_{M,\bar{n}}^{\eta,\bar n_{\rm th}})^2$. Substituting the exact Gaussian coherence from Eq.~\eqref{eq:noisy-rho0M} proves Eqs.~\eqref{eq:sampleloss-coherence-visibility} and~\eqref{eq:sampleloss-sv-upper-explicit}.
\end{proof}

\subsection{Fixed-energy factorial separation}
\label{apen:samplenoisyfixedenergy}

We recall the definition of $A$ from Eq.~\eqref{eq:sampleloss-A}.
 Using $\cosh(2r_{\bar{n}})=2\bar{n}+1$, Eq.~\eqref{eq:sampleloss-D} becomes
\begin{equation}
    D
    =
    4\eta A\bar{n}
    +(A+\eta)^2.
\label{eq:sampleloss-D-closed}
\end{equation}

\begin{corollary}
\label{cor:sampleloss-fixed-energy}
Fix $\bar{n}>0$, $0<\eta\le1$, and $\bar n_{\rm th}\ge0$ independently of $M$. Then, for fixed $a,b,\zeta$ satisfying Eq.~\eqref{eq:sample-constants},
\begin{equation}
    \nu_{\rm SV}^{\star,\eta,\bar n_{\rm th}}
    =
    O_{a,b,\zeta,\bar{n},\eta,\bar n_{\rm th}}\!\left[
        \sqrt M
        \left(
            \frac{D}
                 {4\eta\sqrt{\bar{n}(1+\bar{n})}}
        \right)^M
    \right],
\label{eq:sampleloss-fixed-sv}
\end{equation}
whereas
\begin{equation}
    \nu_{\rm cl}^{\star,\eta,\bar n_{\rm th}}
    =
    \Omega_{b,\zeta,\bar{n},\eta}\!\left[
        \sqrt M
        \left(
            \frac{M}{e\eta \bar{n}}
        \right)^M
    \right].
\label{eq:sampleloss-fixed-classical}
\end{equation}
Consequently,
\begin{equation}
    \frac{\nu_{\rm cl}^{\star,\eta,\bar n_{\rm th}}}
         {\nu_{\rm SV}^{\star,\eta,\bar n_{\rm th}}}
    =
    \Omega_{a,b,\zeta,\bar{n},\eta,\bar n_{\rm th}}\!\left[
        \left(
            \frac{4M}{eD}
            \sqrt{\frac{1+\bar{n}}{\bar{n}}}
        \right)^M
    \right].
\label{eq:sampleloss-fixed-ratio}
\end{equation}
Thus any fixed nonzero transmissivity and any fixed thermal occupation preserve a factorial, hence super-exponential, sampling separation. For $\eta=1$ and $\bar n_{\rm th}=0$, $D_{\bar{n}}^{1,0}=4(1+\bar{n})$ and Eq.~\eqref{eq:sampleloss-fixed-ratio} reduces to the ideal result in Eq.~\eqref{eq:sample-fixed-ratio}.
\end{corollary}

\begin{proof}
For fixed parameters, the central-binomial asymptotic
\begin{equation}
    \binom{M}{M/2}
    =2^M\sqrt{\frac{2}{\pi M}}\,[1+O(M^{-1})]
\end{equation}
turns Eq.~\eqref{eq:sampleloss-coherence-visibility} into
\begin{equation}
    B_{M,\bar{n}}^{\eta,\bar n_{\rm th}}
    \ge
    \Omega_{\bar{n},\eta,\bar n_{\rm th}}\!\left[
        M^{-1/2}
        \left(
            \frac{4\eta\sqrt{\bar{n}(1+\bar{n})}}
                 {D}
        \right)^M
    \right].
\end{equation}
Equation~\eqref{eq:sampleloss-sv-upper-exact} gives Eq.~\eqref{eq:sampleloss-fixed-sv}. Since $\eta \bar{n}(M+1)/M^2=O(M^{-1})$, Eq.~\eqref{eq:sampleloss-classical-lower} and Stirling's formula give Eq.~\eqref{eq:sampleloss-fixed-classical}. Dividing the lower classical bound by the constructive quantum upper bound proves Eq.~\eqref{eq:sampleloss-fixed-ratio}.
\end{proof}

\subsection{The $\bar{n}=M/\log M$ regime under fixed thermal loss}
\label{sec:samplesthermalloss}
The energy scaling used for the strongest ideal separation is also robust in an operational sense, although the simple single-coherence receiver no longer certifies a polynomial quantum cost when there is fixed attenuation.

\begin{corollary}
\label{cor:sampleloss-MlogM}
Fix $0<\eta\le1$ and $\bar n_{\rm th}\ge0$, and let
\begin{equation}
    \bar{n}_M=\frac{M}{\log M}.
\end{equation}
Then the explicit $\ket0\leftrightarrow\ket M$ receiver of Theorem~\ref{thm:sampleloss-finite} obeys
\begin{equation}
    \log \nu_{\rm SV}^{\star,\eta,\bar n_{\rm th}}
    \le
    M\log A
    +O_{a,b,\zeta,\eta,\bar n_{\rm th}}(\log M),
\label{eq:sampleloss-MlogM-sv}
\end{equation}
while every strategy admissible under Eq.~\eqref{eq:classicaldef} satisfies
\begin{equation}
\begin{split}
    \log \nu_{\rm cl}^{\star,\eta,\bar n_{\rm th}}
    \ge{}&
    M\log\log M
    -M(1+\log\eta)
    +\frac{\eta M}{\log M}\\
    &+\frac12\log M
    +O_{b,\zeta,\eta}(1).
\end{split}
\label{eq:sampleloss-MlogM-classical}
\end{equation}
Consequently,
\begin{equation}
    \frac{\nu_{\rm cl}^{\star,\eta,\bar n_{\rm th}}}
         {\nu_{\rm SV}^{\star,\eta,\bar n_{\rm th}}}
    \ge
    \exp\!\left[M\log\log M-O(M)\right].
\label{eq:sampleloss-MlogM-ratio}
\end{equation}
For any fixed nontrivial attenuation $\eta<1$, one has $A>1$, so Eq.~\eqref{eq:sampleloss-MlogM-sv} is an exponential \emph{upper bound for this explicit receiver}; it should not be interpreted as a lower bound on the optimal squeezed-vacuum sample complexity. The conclusion needed here is stronger operationally: even this simple receiver is enough to retain the super-exponential separation between squeezed vacuum and the classical family in Eq.~\eqref{eq:classicaldef}.
\end{corollary}

\begin{proof}
Set $L=\log M$ and $\bar{n}=M/L$. From Eq.~\eqref{eq:sampleloss-D-closed},
\begin{equation}
    D
    =4\eta A\bar{n}\,[1+O(\bar{n}^{-1})],
\end{equation}
while
\begin{equation}
    \sqrt{\bar{n}(1+\bar{n})}=\bar{n}[1+O(\bar{n}^{-1})].
\end{equation}
Hence
\begin{equation}
    \frac{4\eta\sqrt{\bar{n}(1+\bar{n})}}
         {D}
    =
   A^{-1}
    [1+O(\bar{n}^{-1})].
\label{eq:sampleloss-MlogM-base}
\end{equation}
Using the central-binomial asymptotic in Eq.~\eqref{eq:sampleloss-coherence-visibility} therefore gives
\begin{equation}
    \log B_{M,\bar{n}}^{\eta,\bar n_{\rm th}}
    \ge
    -M\log A-O(\log M),
\end{equation}
where $M/\bar{n}=\log M$. Equation~\eqref{eq:sampleloss-sv-upper-exact} proves Eq.~\eqref{eq:sampleloss-MlogM-sv}.

For the classical family, Eq.~\eqref{eq:sampleloss-classical-lower} gives
\begin{equation}
    \nu_{\rm cl}^{\star,\eta,\bar n_{\rm th}}
    \ge
    \frac{(1-2\zeta)^2}{16b^2}
    \left(1-O(L^{-1})\right)
    e^{\eta M/L}
    M!
    \left(\frac{L}{\eta M}\right)^M.
\end{equation}
Stirling's formula yields Eq.~\eqref{eq:sampleloss-MlogM-classical}. Subtracting Eq.~\eqref{eq:sampleloss-MlogM-sv} proves Eq.~\eqref{eq:sampleloss-MlogM-ratio}.
\end{proof}

\subsection{Subpolynomial squeezed vacuum samples versus superpolynomial coherent state samples}
\label{sec:sampleloss-high-energy}

We finally give a stronger scalable statement using a standard heterodyne receiver. This is the finite-sample counterpart of the high-energy pure-loss calculation in Appendix~\ref{sec:high-energy-loss-separation}, and the argument extends directly to fixed thermal noise.

\begin{theorem}
\label{thm:sampleloss-high-energy}
Let $M=2^m\to\infty$, fix $0<\eta\le1$ and $\bar n_{\rm th}\ge0$, and let the squeezed-vacuum input energy $\bar{n}={\bar{n}}_M$ satisfy
\begin{equation}
    \frac{\bar{n}_M(\log M)^2}{M^2}\longrightarrow\infty,
    \qquad
    \frac{\bar{n}_M\log M}{M^2}\longrightarrow0.
\label{eq:sampleloss-high-energy-window}
\end{equation}
Then there is an explicit heterodyne-angle estimator satisfying Eq.~\eqref{eq:sample-success} with
\begin{equation}
    \nu_{\rm SV}^{\star,\eta,\bar n_{\rm th}}
    \le
    \exp\!\left[
        \sqrt{\frac{A}{\eta}}
        \frac{M}{\sqrt{\bar{n}_M}}
        (1+o(1))
    \right]
    =M^{o(1)}.
\label{eq:sampleloss-high-energy-sv}
\end{equation}
Every protocol from the family of Eq.~\eqref{eq:classicaldef}, including correlated probes with collective measurements and adaptive measure-and-prepare strategies, satisfies
\begin{equation}
    \nu_{\rm cl}^{\star,\eta,\bar n_{\rm th}}
    \ge
    \frac{(1-2\zeta)^2}{128\sqrt2\,b^2}
    \frac{M^2}{\eta \bar{n}_M}
    \exp\!\left(
        \frac{M^2}{24\eta \bar{n}_M}
    \right)
    =M^{\omega(1)}.
\label{eq:sampleloss-high-energy-classical}
\end{equation}
In particular,
\begin{equation}
    \frac{\nu_{\rm cl}^{\star,\eta,\bar n_{\rm th}}}
         {\nu_{\rm SV}^{\star,\eta,\bar n_{\rm th}}}
    \ge
\exp\!\left[\Omega\!\left(\frac{M^2}{\bar{n}_M}\right)\right]
    =M^{\omega(1)}.
\label{eq:sampleloss-high-energy-ratio}
\end{equation}
\end{theorem}

\begin{proof}
Let $\sinh^2r_M=\bar{n}_M$. After thermal loss the squeezed vacuum is a centered Gaussian state with covariance entries given by Eq.~\eqref{eq:v_1v_2def}. Recall the definition of $a_M$, $b_M$, $h_M(\theta)$, $\xi_M$, and $R_M$ from Appendix~\ref{sec:high-energy-loss-separation}. We use the same definition here.

The $\mathbb Z_M$ twirl removes all Fourier components below angular frequency $M$. Up to a known, parameter-independent shift $\theta_0$ of the angular origin that absorbs the sign of the surviving Fourier coefficient, the twirled angular density is the Poisson kernel
\begin{equation}
    h_{M,\delta}(\theta)
    =
    \frac{1-R_M^2}
    {2\pi\left[
        1-2R_M\cos(M(\theta-\delta-\theta_0))+R_M^2
    \right]}.
\label{eq:sampleloss-high-energy-poisson}
\end{equation}
The first trigonometric moment is therefore
\begin{equation}
    \mathbb E_\delta\!\left[
        \sin\big(M(\Theta-\theta_0)\big)
    \right]
    =R_M\sin(M\delta).
\label{eq:sampleloss-high-energy-moment}
\end{equation}
For independent heterodyne angles $\Theta_1,\ldots,\Theta_\nu$, define
\begin{equation}
    Y_k:=\sin\big(M(\Theta_k-\theta_0)\big),
    \qquad
    \overline Y_\nu:=\frac1\nu\sum_{k=1}^{\nu}Y_k,
\end{equation}
and use the clipped estimator
\begin{equation}
    \widehat\delta_\nu
    :=
    \frac1M\arcsin\!\left[
        \operatorname{clip}\!\left(
            \frac{\overline Y_\nu}{R_M},
            -\sin a,\sin a
        \right)
    \right].
\label{eq:sampleloss-high-energy-estimator}
\end{equation}
Since $|Y_k|\le1$, Hoeffding's inequality gives
\begin{equation}
    \sup_{\delta\in\Theta_M}
    \Pr_\delta\!\left(
        |\widehat\delta_\nu-\delta|>\frac bM
    \right)
    \le
    2\exp\!\left(
        -\frac{\nu b^2R_M^2\cos^2a}{2}
    \right).
\label{eq:sampleloss-high-energy-hoeffding}
\end{equation}
Hence
\begin{equation}
    \nu_{\rm SV}^{\star,\eta,\bar n_{\rm th}}
    \le
    \left\lceil
        \frac{2}{b^2\cos^2a}\,
        \xi_M^{-M}
        \log\!\left(\frac2\zeta\right)
    \right\rceil.
\label{eq:sampleloss-high-energy-sv-exact}
\end{equation}

For $\bar{n}_M\to\infty$,
\begin{equation}
    a_M=A+O(\bar{n}_M^{-1}),
    \qquad
    b_M=4\eta \bar{n}_M+O(1),
\end{equation}
so
\begin{equation}
\begin{split}
    -\log\xi_M
    &=2\,\operatorname{arctanh}\sqrt{\frac{a_M}{b_M}}\\
    &=\sqrt{\frac{A}{\eta \bar{n}_M}}
      \,[1+o(1)].
\end{split}
\label{eq:sampleloss-high-energy-xi-asymptotic}
\end{equation}
The window~\eqref{eq:sampleloss-high-energy-window} implies
\begin{equation}
    \sqrt{\log M}\ll\frac{M}{\sqrt{\bar{n}_M}}\ll\log M.
\end{equation}
Substitution into Eq.~\eqref{eq:sampleloss-high-energy-sv-exact} proves Eq.~\eqref{eq:sampleloss-high-energy-sv}.

For the classical family, the additional thermal noise is parameter independent, so its QFI is no larger than after pure loss alone. The Poisson-tail argument of Eq.~\eqref{eq:classical-high-energy-upper} therefore applies unchanged and gives, for all sufficiently large $M$,
\begin{equation}
    \sup_{\rho_{\rm cl}}
    \QFI\!\left[
        \Phi_{\delta,M}^{\eta,\bar n_{\rm th}}(\rho_{\rm cl})
    \right]
    \le
    32\sqrt2\,\eta \bar{n}_M
    \exp\!\left(-\frac{M^2}{24\eta \bar{n}_M}\right).
\label{eq:sampleloss-high-energy-qfi-classical}
\end{equation}
Applying the same correlated-probe and two-point testing argument as in Eqs.~\eqref{eq:sample-coherent-qfi-upper}--\eqref{eq:sample-trace-distance-upper} yields Eq.~\eqref{eq:sampleloss-high-energy-classical}. Finally, writing $x_M=M/\sqrt{\bar{n}_M}$, the logarithm of the ratio between the classical lower bound and the heterodyne upper bound is at least
\begin{equation}
    \frac{x_M^2}{24\eta}
    -\sqrt{\frac{A}{\eta}}\,x_M[1+o(1)]
    +O(\log x_M).
\end{equation}
Since $x_M\to\infty$, the quadratic term dominates, proving Eq.~\eqref{eq:sampleloss-high-energy-ratio}.
\end{proof}

\section{Exact local SLD measurement for squeezed vacuum}
\label{sec:additional-null-sv}

Let
\begin{equation}
    \rho_{\rm SV}
    =
    S(re^{\mathrm i\phi_s})\ket0\!\bra0
    S^\dagger(re^{\mathrm i\phi_s}),
    \qquad r>0,
\end{equation}
and let
\begin{equation}
    \rho_0:=\mathcal T_M(\rho_{\rm SV}),
    \qquad
    \rho_\delta
    :=
    e^{\mathrm i\delta\hat n}
    \rho_0
    e^{-\mathrm i\delta\hat n}.
\end{equation}
As shown in Appendix~\ref{apen:mary1}, for even $M$ the twirled
state has the rank-one direct-sum decomposition
\begin{equation}
    \rho_0
    =
    \bigoplus_{r_0\in\mathcal R_M}
    \lambda_{r_0}
    \ket{\widetilde u_{r_0}}
    \!\bra{\widetilde u_{r_0}},
    \qquad
    \mathcal R_M:=\{0,2,\ldots,M-2\},
\label{eq:null-direct-sum}
\end{equation}
where
\begin{equation}
    \ket{\widetilde u_{r_0}}
    =
    \frac{1}{\sqrt{\lambda_{r_0}}}
    \sum_{p=0}^{\infty}
    c_{r_0+pM}\ket{r_0+pM},
    \qquad
    \lambda_{r_0}
    =
    \sum_{p=0}^{\infty}|c_{r_0+pM}|^2.
\end{equation}
For $r>0$, every sector in $\mathcal R_M$ is occupied.

For each occupied sector define
\begin{equation}
    \mu_{r_0}
    :=
    \bra{\widetilde u_{r_0}}
    \hat n
    \ket{\widetilde u_{r_0}},
    \qquad
    \sigma_{r_0}^2
    :=
    \operatorname{Var}_{\widetilde u_{r_0}}(\hat n),
\label{eq:sld-sector-moments}
\end{equation}
and
\begin{equation}
    \ket{v_{r_0}}
    :=
    \frac{
        (\hat n-\mu_{r_0})
        \ket{\widetilde u_{r_0}}
    }{
        \sigma_{r_0}
    }.
\label{eq:sld-tangent-vector}
\end{equation}
Then
\begin{equation}
 \langle\widetilde u_{r_0}|v_{r_0}\rangle=0,
    \qquad
    \langle v_{r_0}|v_{r_0}\rangle=1.
\end{equation}
Moreover, $\ket{v_{r_0}}$ belongs to the same residue sector as
$\ket{\widetilde u_{r_0}}$ and is orthogonal to the rank-one support
of that sector. Hence
\begin{equation}
    \rho_0\ket{v_{r_0}}=0.
\end{equation}
Equation~\eqref{eq:sld-tangent-vector} gives
\begin{equation}
    \hat n\ket{\widetilde u_{r_0}}
    =
    \mu_{r_0}\ket{\widetilde u_{r_0}}
    +
    \sigma_{r_0}\ket{v_{r_0}}.
\end{equation}
Therefore
\begin{equation}
\begin{split}
    \dot\rho_0
    &:=
    \left.
    \partial_\delta\rho_\delta
    \right|_{\delta=0}
\\
    &=
    \mathrm i[\hat n,\rho_0]
\\
    &=
    \mathrm i
    \sum_{r_0\in\mathcal R_M}
    \lambda_{r_0}\sigma_{r_0}
    \left(
        \ket{v_{r_0}}\!\bra{\widetilde u_{r_0}}
        -
        \ket{\widetilde u_{r_0}}\!\bra{v_{r_0}}
    \right).
\end{split}
\label{eq:sld-state-derivative}
\end{equation}
Using
\begin{equation}
    \rho_0\ket{\widetilde u_{r_0}}
    =
    \lambda_{r_0}\ket{\widetilde u_{r_0}},
    \qquad
    \rho_0\ket{v_{r_0}}=0,
\end{equation}
one verifies directly that the SLD-defining equation
\begin{equation}
    \dot\rho_0
    =
    \frac12
    \left(
        \mathcal L_0\rho_0+
        \rho_0\mathcal L_0
    \right).
\end{equation}
is satisfied by
\begin{equation}
    \mathcal L_0
    =
    2\mathrm i
    \sum_{r_0\in\mathcal R_M}
    \sigma_{r_0}
    \left(
        \ket{v_{r_0}}\!\bra{\widetilde u_{r_0}}
        -
        \ket{\widetilde u_{r_0}}\!\bra{v_{r_0}}
    \right).
\label{eq:canonical-sld}
\end{equation}
The nonzero-eigenvalue eigenvectors of this SLD operator are
\begin{equation}
    \ket{w_{r_0,\pm}}
    :=
    \frac{
        \ket{\widetilde u_{r_0}}
        \pm\mathrm i\ket{v_{r_0}}
    }{\sqrt2},
    \qquad
    \mathcal L_0\ket{w_{r_0,\pm}}
    =
    \pm2\sigma_{r_0}\ket{w_{r_0,\pm}}.
\label{eq:sld-eigenvectors}
\end{equation}
Thus Eq.~\eqref{eq:canonical-sld} is an SLD operator.

Accordingly, we can define the following projective measurement
\begin{equation}
    \mathsf M_{\rm SLD}
    :=
    \left\{
        \Pi_{r_0,+}^{\rm SLD},
        \Pi_{r_0,-}^{\rm SLD}
    \right\}_{r_0\in\mathcal R_M}
    \cup
    \left\{
        \Pi_0^{\rm SLD}
    \right\},
\label{eq:sld-povm}
\end{equation}
where
\begin{equation}
    \Pi_{r_0,\pm}^{\rm SLD}
    :=
    \ket{w_{r_0,\pm}}\!\bra{w_{r_0,\pm}},
    \qquad
    \Pi_0^{\rm SLD}
    :=
    I-
    \sum_{r_0\in\mathcal R_M}
    \left(
        \Pi_{r_0,+}^{\rm SLD}
        +
        \Pi_{r_0,-}^{\rm SLD}
    \right),
\end{equation}
 At $\delta=0$,
\begin{equation}
    p_{r_0,\pm}(0)
    :=
    \operatorname{Tr}
    \left[
        \Pi_{r_0,\pm}^{\rm SLD}\rho_0
    \right]
    =
    \frac{\lambda_{r_0}}{2},
\end{equation}
and
\begin{equation}
    \left.
    \partial_\delta
    p_{r_0,\pm}(\delta)
    \right|_{\delta=0}
    =
    \pm\lambda_{r_0}\sigma_{r_0}.
\end{equation}
Hence
\begin{equation}
\begin{split}
    \mathcal J
    \!\left(
        \mathsf M_{\rm SLD};0
    \right)
    &=
    \sum_{r_0\in\mathcal R_M}
    \sum_{\pm}
    \frac{
        \lambda_{r_0}^2\sigma_{r_0}^2
    }{
        \lambda_{r_0}/2
    }
\\
    &=
    4
    \sum_{r_0\in\mathcal R_M}
    \lambda_{r_0}\sigma_{r_0}^2.
\end{split}
\end{equation}
The complementary outcome has probability $O(\delta^4)$ because
the complete first-order tangent space is already contained in the
span of the vectors
$\ket{\widetilde u_{r_0}}$ and $\ket{v_{r_0}}$.
It therefore contributes zero Fisher information at the operating
point.

Finally,
\begin{equation}
    \operatorname{Tr}
    \left(
        \rho_0\mathcal L_0^2
    \right)
    =
    4
    \sum_{r_0\in\mathcal R_M}
    \lambda_{r_0}\sigma_{r_0}^2,
\end{equation}
so the classical Fisher information of the stated measurement equals
the QFI.

\section{Experimentally feasible receivers}
\label{apen:feasiblemeas}

 This appendix analyses in detail the two receivers used in Sec.~\ref{sec:neartermprocessor}: balanced homodyne
detection (Appendix~\ref{sec:additional-homodyne-sv}), and a variational Gaussian circuit followed by finite-resolution
photon counting (Appendix~\ref{apen:variational-receiver}). We analyse both receivers in the presence of the noise models studied in Appendix~\ref{apen:qfiloss}.

\subsection{Balanced homodyne detection}
\label{sec:additional-homodyne-sv}

Balanced homodyne detection mixes the signal with a strong coherent local
oscillator on a $50{:}50$ beam splitter and records the difference
photocurrent. In the strong-local-oscillator limit, one measures the single
quadrature
\begin{equation}
     X_\vartheta
    :=
    \frac{
        \hat a e^{-\mathrm i\vartheta}
        +\hat a^\dagger e^{\mathrm i\vartheta}
    }{\sqrt2},
\end{equation}
where $\vartheta$ is the local-oscillator phase. Only one quadrature is
measured on each run. This is sufficient for local estimation because the
unknown phase rotates the quadrature distribution relative to the fixed
measurement axis. Choosing the phase $\vartheta_\star=\pi/(2M)$ places the first
angular harmonic surviving the $\mathbb Z_M$ twirl at maximal slope around
$\delta=0$.

Let $\rho_G$ be a centered single-mode Gaussian state whose covariance
eigenvalues are $v_1<v_2$. Equation~\eqref{eq:v_1v_2def} includes
the ideal squeezed vacuum, pure loss, and thermal loss as special cases. When the twirling channel imparts a rotation $\frac{2\pi k}{M}$, $X_\vartheta$ is a centered Gaussian with variance
\begin{equation}
    V_k(\delta)
    =\frac14\left\{
        v_2+v_1
        -(v_2-v_1)
        \cos\left[2\left(
        \vartheta-\delta-\frac{2\pi k}{M}
        \right)\right]
    \right\}.
\end{equation}
Define 
\begin{equation}
    A(t):=\frac{v_2+v_1}{8}t^2,
    \qquad
    z(t):=\frac{v_2-v_1}{8}t^2\;,
\label{eq:homodyne-A-z}
\end{equation}
such that the characteristic function conditioned on $k$ is 
\begin{equation}
\begin{split}
    \Phi_k(t;\delta)
    &:=
    \mathbb E[e^{itX_\vartheta}\mid k]
\\
    &=
    \exp\left[-\frac{t^2V_k(\delta)}{2}\right]
\\
    &=
    e^{-A(t)}
    \exp\left\{
        z(t)
        \cos\left[
        2\left(
        \vartheta-\delta-\frac{2\pi k}{M}
        \right)\right]
    \right\}\;.
\end{split}
\end{equation}

Since the rotation label $k$ is uniformly random and unobserved,
the characteristic function of the actual homodyne distribution is
\begin{equation}
    \Phi_\delta(t)
    =
    \frac1M
    \sum_{k=0}^{M-1}\Phi_k(t;\delta).
\end{equation}
Using
\begin{equation}
    e^{z\cos u}
    =
    I_0(z)+
    2\sum_{n=1}^\infty I_n(z)\cos(nu),
\end{equation}
where $I_k$ is the modified Bessel function, and the discrete Fourier identity
\begin{equation}
    \frac1M
    \sum_{k=0}^{M-1}
    e^{-4\pi i nk/M}
    =
    \begin{cases}
        1,&2n\equiv0\pmod M,\\
        0,&\text{otherwise},
    \end{cases}
\end{equation}
we find, for even $M$,
\begin{equation}
    \Phi_\delta(t)
    =
    e^{-A(t)}
    \left[
        I_0(z(t))
        +
        2\sum_{\ell=1}^\infty
        I_{\ell M/2}(z(t))
        \cos\bigl(
            \ell M(\vartheta-\delta)
        \bigr)
    \right].
\end{equation}

Phase diffusion convolves the orientation with a wrapped Gaussian
and therefore multiplies angular order $\ell M$ by
$e^{-\ell^2M^2\sigma_M^2/2}$ giving the final characteristic function as
\begin{equation}
\begin{split}
    \Phi_{\delta,\sigma_M}(t)
    &:={}
    \mathbb E_{\delta,\sigma_M}
    [e^{\mathrm i tX_\vartheta}]
\notag\\
    &={}
    e^{-A(t)}\left[
       I_0(z(t))
       +2\sum_{\ell=1}^{\infty}
       e^{-\ell^2M^2\sigma_M^2/2}
       I_{\ell M/2}(z(t))
       \cos\!\bigl(\ell M(\vartheta-\delta)\bigr)
       \right],
\end{split}
\label{eq:homodyne-characteristic}
\end{equation}
Since the homodyne distribution is even in $x$,
$\mathbb E[\sin(tX_\vartheta)]=0$, and hence
\begin{equation}
    \Phi_{\delta,\sigma_M}(t)
    =
    \mathbb E_{\delta,\sigma_M}
    [\cos(tX_\vartheta)].
\end{equation}
Differentiation at
$\vartheta_\star=\pi/(2M)$ leaves only odd values of $\ell$ and gives
\begin{equation}
    \left|
    \left.\partial_\delta
    \mathbb E_{\delta,\sigma_M}
    [\cos(tX_{\vartheta_\star})]
    \right|_{\delta=0}
    \right|
    =
    2M e^{-A(t)}
    |\Delta_{M/2,\sigma_M}(z(t))|\;.
\end{equation}
where
\begin{equation}
    \Delta_{M/2,\sigma_M}(z)
    :=
    \sum_{j=0}^{\infty}
    (-1)^j(2j+1)
    e^{-(2j+1)^2M^2\sigma_M^2/2}
    I_{(2j+1)M/2}(z).
\label{eq:homodyne-Delta-noisy}
\end{equation}

For any square-integrable statistic $T$ and score
$\ell_\delta=\partial_\delta\log p_\delta$,
\begin{equation}
    \partial_\delta\mathbb E_\delta[T]
    =
    \mathbb E_\delta
    [(T-\mathbb E_\delta[T])\ell_\delta].
\end{equation}
Cauchy--Schwarz therefore gives
\begin{equation}
    \mathcal J_{\rm hom}(\delta)
    \ge
    \frac{
        (\partial_\delta\mathbb E_\delta[T])^2
    }{
        \operatorname{Var}_\delta(T)
    }.
\end{equation}
Choosing $T(x)=\cos(tx)$ and using $\operatorname{Var}(T)\le1$ proves

\begin{equation}
    \mathcal J_{\rm hom}^{(\eta,\bar n_{\rm th},\sigma_M)}(0)
    \ge
    4M^2e^{-2A(t)}
    \Delta_{M/2,\sigma_M}(z(t))^2\;.
\label{eq:homodyne-CFI-lower}
\end{equation}

The bounded statistic $\cos(tX)$ is used only to certify a rigorous CFI
lower bound. Experimentally, every homodyne outcome $X$ may be retained and
the full one-dimensional likelihood may be used. Setting $\sigma_M=0$
recovers the loss and thermal-noise formulas without phase diffusion.

\subsubsection{Fixed mean photon number}

For fixed $v_1<v_2$, define
\begin{equation}
    d:=\frac{v_2-v_1}{2},
    \qquad
    \Psi(d):=
    \sqrt{1+d^2}
    +\log\!\left(
        \frac{d}{1+\sqrt{1+d^2}}
    \right),
\end{equation}
and
\begin{equation}
    B(v_1,v_2)
    :=
    \exp\!\left[
        \Psi(d)-\frac{v_2+v_1}{2}
    \right].
\label{eq:homodyne-fixed-B}
\end{equation}
The Debye expansion of $I_{M/2}(Md/2)$, together with exponential suppression of
the higher odd harmonics in Eq.~\eqref{eq:homodyne-Delta-noisy}, gives, for
$t_M=\sqrt{2M}$,
\begin{equation}
    \mathcal J_{\rm hom}^{(\eta,\bar n_{\rm th},\sigma_M)}(0)
    \ge
    e^{-M^2\sigma_M^2}
    \frac{4M}{\pi\sqrt{1+d^2}}
    B(v_1,v_2)^M[1+o(1)]\;.
\label{eq:homodyne-fixed-general}
\end{equation}
For ideal squeezed vacuum,
\begin{equation}
    d=\sinh(2r),
    \qquad
    B(e^{-2r},e^{2r})=\tanh r,
\end{equation}
so
\begin{equation}
    \mathcal J_{\rm hom}^{\rm SV,\sigma_M}(0)
    \ge
    e^{-M^2\sigma_M^2}
    \frac{4M}{\pi\cosh(2r)}
    \tanh^M r\,[1+o(1)]\;.
\label{eq:homodyne-fixed-cfi-improved}
\end{equation}
For fixed $\bar n=\sinh^2r>0$, fixed thermal-loss parameters with
$\eta>0$, and $\sigma_M=0$, the factor $B(v_1,v_2)$ is a positive constant
independent of $M$. Combining Eq.~\eqref{eq:homodyne-fixed-general} with the
classical upper bound gives
\begin{equation}
    \log
    \frac{\mathcal J_{\rm hom}}
    {\displaystyle
    \sup_{\rho_{\rm cl}}
    \QFI}
    \ge
    M\log M-O(M)-M^2\sigma_M^2.
\label{eq:homodyne-fixed-ratio-asymptotic}
\end{equation}
Thus $\sigma_M^2=O(\log M/M^2)$ changes the factorial separation by at
most a polynomial factor.

\subsection{Variational Gaussian receiver with photon counting}
\label{apen:variational-receiver}

Let $\rho_\delta$ be the output state to be measured. We consider an
$L$-layer Gaussian unitary
\begin{equation}
    U(\boldsymbol\theta)
    :=
    \prod_{\ell=1}^{L}
    R(\varphi_\ell)D(\alpha_\ell)S(\zeta_\ell),
\end{equation}
where
\begin{equation}
    R(\varphi)=e^{-\mathrm i\varphi\hat n},
    \qquad
    D(\alpha)=e^{\alpha\hat a^\dagger-\alpha^*\hat a},
    \qquad
    S(\zeta)=e^{(\zeta^*\hat a^2-\zeta\hat a^{\dagger2})/2}.
\end{equation}
The ordering convention is fixed by the circuit in
Fig.~\ref{fig:optmeas}; changing it only changes the parameterization of the
same variational family.

A photon counter with maximum resolved value $n_{\max}$ is described by
\begin{equation}
\begin{split}
    \Pi_n(\boldsymbol\theta)
    &=
    U(\boldsymbol\theta)^\dagger
    \ket n\!\bra n
    U(\boldsymbol\theta),
    \qquad 0\le n<n_{\max},
\notag\\
    \Pi_{n_{\max}}(\boldsymbol\theta)
    &=
    U(\boldsymbol\theta)^\dagger
    \left(
        \sum_{n=n_{\max}}^\infty\ket n\!\bra n
    \right)
    U(\boldsymbol\theta).
\end{split}
\end{equation}
The outcome probabilities and local CFI are
\begin{equation}
    p_n(\delta;\boldsymbol\theta)
    =
    \operatorname{Tr}
    [\rho_\delta\Pi_n(\boldsymbol\theta)],
\end{equation}
\begin{equation}
    \mathcal J_{\rm var}(\boldsymbol\theta)
    =
    \sum_{n=0}^{n_{\max}}
    \frac{
        [\partial_\delta
        p_n(\delta;\boldsymbol\theta)]^2
    }{
        p_n(\delta;\boldsymbol\theta)
    }
    \bigg|_{\delta=0},
\label{eq:variational-CFI}
\end{equation}
with zero-probability terms interpreted by continuity. Since
$\dot\rho_0=\mathrm i[\hat n,\rho_0]$,
\begin{equation}
    \left.\partial_\delta p_n\right|_{\delta=0}
    =
    \mathrm i\operatorname{Tr}
    \left(
        [\hat n,\rho_0]
        \Pi_n(\boldsymbol\theta)
    \right),
\end{equation}
so Eq.~\eqref{eq:variational-CFI} can be evaluated without finite
differencing. The reported receiver is
\begin{equation}
    \mathcal J_{\rm var}^{(L,n_{\max})}
    :=
    \max_{\boldsymbol\theta}
    \mathcal J_{\rm var}(\boldsymbol\theta).
\end{equation}

Optical loss, thermal noise, and phase diffusion are included simply by
taking $\rho_0$ in the above expressions to be the $\delta=0$ state in:
    \begin{equation}
    \rho_\delta
    :=
    \mathcal D_{\sigma_M}
    \circ\mathcal G_{\eta,\bar n_{\rm th}}
    \circ\mathcal E_{\delta,M}
    (\ket{r}\!\bra{r}),
\end{equation}
 In the Fock basis, phase diffusion is
implemented exactly as
\begin{equation}
    (\rho_0^{\sigma_M})_{nm}
    =
    e^{-\sigma_M^2(n-m)^2/2}(\rho_0)_{nm}.
\end{equation}
The numerical results in Fig.~\ref{fig:optmeas} already include fixed optical
loss and thermal noise. The same optimization can be repeated with phase
diffusion, but no analytic asymptotic guarantee is claimed for this
variational family.

Unlike homodyne detection, this receiver requires active Gaussian control
and photon-number resolution, but it can also approach the QFI closely at the
small photon numbers considered in Sec.~\ref{sec:neartermprocessor}.

\end{widetext}

\bibliography{Reference}
\bibliographystyle{naturemag}

\end{document}